%% file: main.tex
\documentclass[11pt]{article} 
\usepackage{ifdraft}           
\usepackage[utf8]{inputenc}%
\usepackage[tracking=true, letterspace=100, expansion=false]{microtype}
\usepackage[T1]{fontenc}
\usepackage{amsmath}%
\usepackage{amsfonts}%
\usepackage{amsthm}%
\usepackage{booktabs}
\usepackage{tabularx}
\usepackage{caption}
\usepackage[labelformat=simple]{subcaption}
\usepackage{multirow}
\usepackage{graphicx}
\usepackage[margin=1.2in]{geometry}
\usepackage[onehalfspacing]{setspace}
\usepackage{mathtools}

\usepackage{bbm} 

\DeclareMathOperator{\var}{var}

\usepackage{comment}
\usepackage[title]{appendix}
\usepackage{etoolbox} 
\AtBeginEnvironment{appendices}{\crefalias{section}{appendix}}

\usepackage[
bibencoding=utf8,
maxbibnames=3, 
minbibnames=1, 
backend=biber,%
sortlocale=en_US,%
style=apa,
uniquename=false,%
url=true,%
sortcites=false,
doi=true,%
eprint=true%
]{biblatex}
\usepackage{changepage}

\usepackage{xcolor}

\definecolor{yaleblue}{RGB}{0,53,107}
\definecolor{webbrown}{rgb}{.6,0,0}%
\usepackage{hyperref} %
\hypersetup{%
  breaklinks = true,%
  colorlinks = true,%
  anchorcolor = webbrown,%
  citecolor = webbrown,%
  filecolor = webbrown,%
  linkcolor = webbrown,%
  menucolor = webbrown,%
  urlcolor= webbrown,%
  citebordercolor= 1 0 0,%
  menubordercolor=1 0 0,%
  urlbordercolor=1 0 0,%
  runbordercolor=1 0 0,%
  pdftitle={Financial Event Studies},
  pdfauthor={Goldsmith-Pinkham and Lyu}
}

\usepackage{cleveref}
\crefname{appsec}{appendix}{appendices}
\crefname{appsubsec}{appendix}{appendices}
\crefname{assumption}{assumption}{assumptions}

\usepackage[nolist]{acronym}
\begin{acronym}
  \acro{CI}{confidence interval}%
  \acro{OLS}{ordinary least squares}%
  \acro{CLT}{central limit theorem}%
  \acro{IV}{instrumental variables}%
  \acro{ATE}{average treatment effect}%
  \acro{RCT}{randomized control trial}%
  \acro{VAM}{value-added model}%
  \acro{LAN}{locally asymptotically normal}%
  \acro{DiD}{difference-in-differences}%
  \acro{OVB}{omitted variables bias}
  \acro{FWL}{Frisch-Waugh-Lovell}
\end{acronym}

\usepackage{accents}

\newcolumntype{Y}{>{\centering\arraybackslash}X}

\usepackage{float}
\usepackage{ragged2e}

\usepackage{amsthm}
\usepackage{xcolor}
\usepackage{mdframed}

\newtheoremstyle{myplain}
  {6pt}   
  {6pt}   
  {\itshape} 
  {}      
  {\bfseries\color{yaleblue}} 
  {.}     
  {0.5em} 
  {\thmname{#1}\thmnumber{ #2}\thmnote{ (\textit{#3})}}

\newtheoremstyle{myremark}
  {6pt}
  {6pt}
  {}
  {}
  {\bfseries\color{yaleblue}}
  {.}
  {0.5em}
  {\thmname{#1}\thmnumber{ #2}\thmnote{ (\textit{#3})}}

\theoremstyle{myplain}
\newtheorem{theorem}{Theorem}[section]
\newtheorem{lemma}{Lemma}
\newtheorem{proposition}{Proposition}
\newtheorem{assumption}{Assumption}
\newtheorem{corollary}{Corollary}
\newtheorem{defN}{Definition}

\theoremstyle{myremark}
\newtheorem{remark}{Remark}

\newmdenv[
  topline=false,
  bottomline=false,
  rightline=false,
  leftline=true,
  linecolor=blue!70!black,
  linewidth=1.2pt,
  innerleftmargin=6pt,
  skipabove=\topsep,
  skipbelow=\topsep
]{thmframe}

\newmdenv[
  topline=false,
  bottomline=false,
  rightline=false,
  leftline=true,
  linecolor=blue!50!black,
  linewidth=0.9pt,
  innerleftmargin=6pt,
  skipabove=\topsep,
  skipbelow=\topsep
]{remarkframe}

\renewenvironment{theorem}{%
  \begin{thmframe}\begin{oldtheorem}%
}{%
  \end{oldtheorem}\end{thmframe}%
}

\renewenvironment{lemma}{%
  \begin{thmframe}\begin{oldlemma}%
}{%
  \end{oldlemma}\end{thmframe}%
}

\renewenvironment{proposition}{%
  \begin{thmframe}\begin{oldproposition}%
}{%
  \end{oldproposition}\end{thmframe}%
}

\renewenvironment{assumption}{%
  \begin{thmframe}\begin{oldassumption}%
}{%
  \end{oldassumption}\end{thmframe}%
}

\renewenvironment{corollary}{%
  \begin{thmframe}\begin{oldcorollary}%
}{%
  \end{oldcorollary}\end{thmframe}%
}

\renewenvironment{remark}{%
  \begin{remarkframe}\begin{oldremark}%
}{%
  \end{oldremark}\end{remarkframe}%
}

\begin{document}

\title{Causal Inference in Financial Event Studies\thanks{\noindent Contact: \href{mailto:paul.goldsmith-pinkham@yale.edu}{paul.goldsmith-pinkham@yale.edu} We thank Nick Barberis, Stefano Giglio and Will Goetzmann for helpful discussions, and audiences at the NBER Summer Institute Forecasting \& Empirical Methods session, SEA, Dallas Fed, Kellogg Finance department, and SMU Statistics department.}}
\author{Paul Goldsmith-Pinkham\\Yale University \& NBER \and Tianshu Lyu\\Yale University }%
\date{\today}

  \pagenumbering{Alph} 
\begin{titlepage}
  \maketitle
  \thispagestyle{empty}
\begin{adjustwidth*}{0.1cm}{0.1cm}
\begin{abstract} 
    Financial event studies, ubiquitous in finance research, typically use linear factor models with known factors to estimate abnormal returns and identify causal effects of information events. This paper demonstrates that when factor models are misspecified—an almost certain reality—traditional event study estimators produce inconsistent estimates of treatment effects. The bias is particularly severe during volatile periods, over long horizons, and when event timing correlates with market conditions. We derive precise conditions for identification and expressions for asymptotic bias. As an alternative, we propose synthetic control methods that construct replicating portfolios from control securities without imposing specific factor structures. Revisiting four empirical applications, we show that some established findings may reflect model misspecification rather than true treatment effects. While traditional methods remain reliable for short-horizon studies with random event timing, our results suggest caution when interpreting long-horizon or volatile-period event studies and highlight the importance of quasi-experimental designs when available.
\end{abstract}
\end{adjustwidth*}
\end{titlepage}
\pagenumbering{arabic}
\setstretch{1.25}


\section{Introduction}
Financial economists were practicing causal inference well before the credibility revolution \parencite{angrist2010credibility}. By examining how asset prices respond to information events---such as merger announcements, earnings releases, or regulatory changes---financial event studies compare the returns of treated assets to benchmark comparison asset returns.  The approach remains central: between 2010 and 2025, 305 articles in the Journal of Finance and the Review of Financial Studies reference event‐study methods (\Cref{fig:jf_rfs}). 

\begin{figure}[bh]
    \centering
    \caption{\textbf{Prevalence of financial event studies in finance journals:} This figure plots the share of articles by year that mention the word ``cumulative abnormal returns'' OR ``announcement returns'' in the Journal of Finance or Review of Financial Studies. Source: \url{https://paulgp.com/econlit-pipeline/search.html}}
    \label{fig:jf_rfs}
       \includegraphics[width=\linewidth]{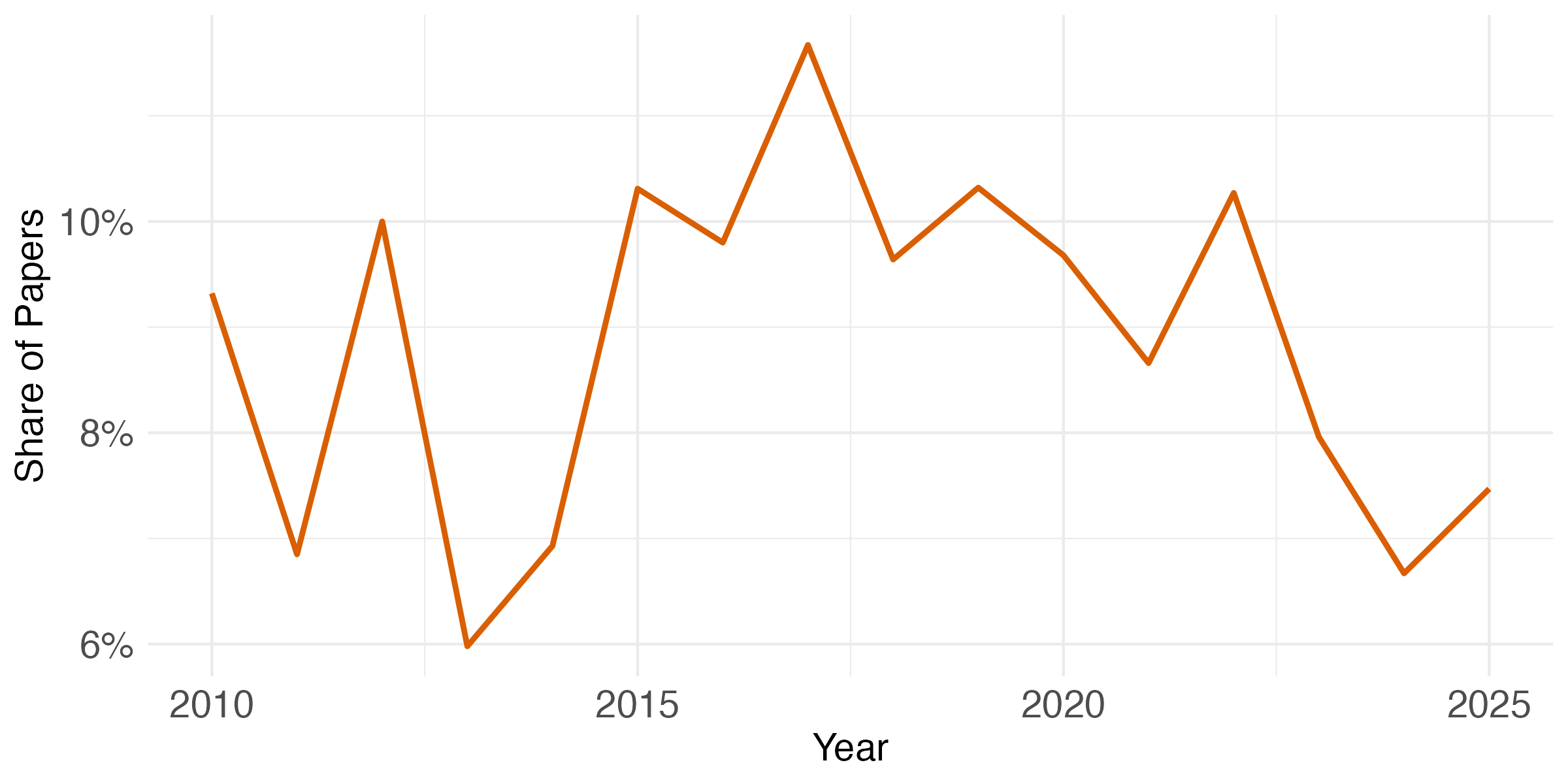}
\end{figure}

While financial event studies target causal effects as their estimands, the suite of estimators used in financial event studies are antiquated relative to the many tools available. The textbook approach, starting as early as \textcite{fama1969adjustment} 
and canonized in the \textcite{campbell1997econometrics} textbook, relies on linear factor models with known factors to construct counterfactual returns, i.e. what a security's return would have been absent the event. Researchers typically estimate a security's exposure to market factors during a pre-event window, then use these estimated loadings to predict what returns during the event window. The difference between actual and predicted returns constitutes the abnormal return.

This paper demonstrates that this standard approach faces a fundamental identification challenge. We show analytically that when the factor model is misspecified---which is almost certainly the case given the ongoing debates about the appropriate asset pricing model---abnormal return estimators are generally inconsistent estimators for causal effects. The problem is particularly severe in three empirically relevant scenarios. First, when events occur during periods of extreme market volatility, even small misspecification in factor loadings gets amplified by large factor realizations, potentially generating economically significant bias. Second, in long-horizon event studies that examine returns over months or years, misspecification bias accumulates over time, making the resulting estimates potentially more reflective of model error than treatment effects. Third, when event timing coincides with particular market conditions—for instance, if mergers cluster during market downturns—the standard estimators conflate selection effects with treatment effects.

We provide precise conditions under which traditional event study methods identify causal effects. Identification of the average treatment effect on the treated requires either correct specification of the factor model (unlikely given decades of asset pricing research showing the difficulty of this task), or random assignment of treatment across securities. When these conditions fail, we derive expressions for the asymptotic bias that provide guidance on when concerns should be most acute. 

Our results stand in contrast to folk wisdom that the structure of the factor model does not have significant impacts on the size of the estimated effects. For example, in footnote 5, \textcite{shleifer1986demand} states ``The [index inclusion] results were not materially different when returns were not corrected for market movements.'' We show that this irrelevance is due to two key features: (1) random timing of many events over time (in the \textcite{shleifer1986demand} case, index inclusions) and (2) very short-run estimates such that the treatment effect dominates any omitted risk premium. If these two cases do not hold, this irrelevance will disappear. 

We contrast three types of estimators that can be used for financial event studies, and compare their properties: (1) classic abnormal return estimators, based on specified factors, (2) difference-in-mean estimators, which construct control groups through decisions of the econometrician, and (3) synthetic estimators \parencite{abadie2003economic,abadie2010synthetic}, which use historical prices from control assets to construct either a replicating portfolio (synthetic control) or to construct a set of factors using PCA \parencite{xu2017generalized}.  The key insight is that rather than imposing a specific factor structure ex ante, the synthetic methods construct a portfolio of control securities that best match the pre-event return path of treated securities. If such a replicating portfolio exists, it should provide valid counterfactual returns in the post-event period without requiring correct specification of the underlying factor model. 

Our theoretical results hinge on the assumption that the expected return for a cohort of treated stocks follows an unknown time-invariant linear factor model. This assumption is not innocuous, and likely not true for all time periods. But, it is also a weaker assumption than the traditional abnormal return estimators. Any approach that uses a model to infer the counterfactual outcomes for the treated stocks will require some kind of model stability assumption (without additional structure like \textcite{kelly2019characteristics}). We view it as valuable future work to see if other more robust asset pricing models can be used to generate counterfactual returns, such as \textcite{giglio2025test} and \textcite{kelly2019characteristics}.

One key benefit of focusing carefully on the estimand of interest is that we are able to show that buy-and-hold abnormal return estimates are particularly challenging to estimate because they require the matching portfolio to not just match on expected returns, but also on volatility. If the control group's returns have different variance, then the differing volatility drag will lead to very different results. To make this concrete: imagine that there is \emph{no} treatment effect, but a diversified portfolio is used as a control group for a stock, both with equal expected returns. The lower variance for the diversified portfolio will lead to a \emph{negative} treatment effect from a buy-and-hold perspective, despite no actual treatment effect. This implies that doing buy-and-hold abnormal returns with an index can be seriously flawed.

We revisit four empirical settings that span the range of typical applications. First, we reexamine the \textcite{acemoglu2016value} study of political connections during the 2008 financial crisis, where the Treasury Secretary announcement coincided with extreme market volatility—daily returns exceeded 6\% on multiple event days. The original estimates using simple averaging suggest economically large effects of political connections. Even abnormal return models using the Fama-French 3 factor model suggest economically meaningful effects. However, the estimates disappear when using our proposed synthetic methods, suggesting that model misspecification with a single event can create spurious results when events coincide with volatile market conditions.

Next, we analyze S\&P 500 index inclusions, and show that since the index inclusion events appear random across time, the effect of short-run model misspecification is non-existent, echoing the folk wisdom above. However, we show that the substantial pre-announcement drift, often pointed to as a source of possible index inclusion front-running, disappears once we properly account for the unobserved factor exposures of included firms. This finding suggests that what appears to be anticipation or momentum may actually reflect model misspecification. 

Third, we examine the effect of acquistions in merger deals on acquiring firms with some studies finding large negative abnormal returns over several years.\parencite{loughran1997long, rau1998glamour}  We demonstrate that these long-run patterns are highly sensitive to model specification, consistent with our theoretical prediction that misspecification bias accumulates over longer horizons.

Our last empirical result applies a version of \textcite{lalonde1986evaluating} to our analysis by using quasi-experimental variation to provide a benchmark for our model-based approaches. The treatment and control groups for the baseline are found in close merger contests where multiple firms bid for the same target. Following \textcite{malmendier2018winning}, contest losers provide a natural counterfactual for winners since they are ex ante similar firms competing for identical targets. The results are not supportive of abnormal return models at all, but only weakly support synthetic methods. These results suggest that for long-run analyses, it is far better to construct counterfactuals based on quasi-experimental variation than using model-based approaches.

These empirical findings have important implications for the interpretation of the vast event study literature in finance. Many influential results---particularly those involving long horizons, volatile periods, or systematic event timing---may reflect factor model misspecification rather than true treatment effects. However, we emphasize that our results do not invalidate the entire enterprise. For short-horizon studies with plausibly random event timing, traditional methods remain reliable and our empirical work confirms they produce similar estimates to more sophisticated approaches. The key insight is recognizing when standard methods are likely to fail and having appropriate alternatives available.




Our work connects several distinct literatures. Methodologically, we build on the econometrics of event studies in finance \parencite{mackinlay1997event, kothari2007econometrics} while incorporating insights from the modern causal inference literature \parencite{imbens2015causal,abadie2021introduction}. We also contribute to the older debate about long-run event studies \parencite{mitchell2000managerial, barber1997detecting} by providing a formal framework for understanding when and why these studies are problematic.


\section{Estimands and estimators in financial event studies}
\label{sec:Identificaton}

This section formalizes the setup of financial events on stock market returns in the language of potential outcomes. We begin by introducing the basic notation (Section~\ref{sec:notation}), defining potential returns and treatment indicators for each security over time. We then specify the causal estimands of interest, clarifying what it means to identify a treatment effect in an event study context. Finally, we discuss how these causal quantities relate to traditional event study methods based on “abnormal returns” and factor model adjustments.

\subsection{Setup and notation}
\label{sec:notation}

We study the causal effects of corporate events on security returns using a potential outcomes framework. Consider a panel of $N$ securities indexed by $i = 1, 2, \ldots, N$ observed over $T$ time periods indexed by $t = 1, 2, \ldots, T$.

\subsubsection{Event Timing and Treatment Status}

For each security $i$, let $T_i$ denote the time when an event occurs:
\begin{equation}
T_i = \begin{cases}
s & \text{if security } i \text{ experiences the event at time } s \\
\infty & \text{if security } i \text{ never experiences the event}
\end{cases}
\end{equation}

We denote the set of event times as $\mathcal{S} \subseteq \{1, \ldots, T\}$ and the set of never-treated (control) securities as $\mathcal{C} = \{i: T_i = \infty\}$. Following standard practice in event studies, we assume events are irreversible---once an event occurs (e.g., a merger announcement or earnings release), it cannot be undone.

\subsubsection{Potential Outcomes Framework}

Now, we define the potential outcomes framework for our returns. Let \(R_{i,t}(s)\) be the potential return for security \(i\) at time \(t\) if it has the event occur in period $s$, and \(R_{i,t}(\infty)\) the potential return in the absence of any event. Because a security cannot be both treated and untreated, we only observe one of the potential returns for each \((i,t)\):

\begin{equation}
    R_{i,t} = R_{i,t}(\infty) + \sum_{s \in \mathcal{S}}(Y_{i,t}(s) - Y_{i,t}(\infty))1(T_{i} = s).
\end{equation}

\subsubsection{Treatment Effects}

We postulate that financial event studies are focused on identifying the difference between the \emph{realized} returns for a treated firm ($R_{it}(s)$) versus the returns in the \emph{absence} of the event. We define the difference in returns due to the event in period $s$ for firm $i$ in period $t$ as the \emph{individual treatment} or equivalently, the \emph{abnormal firm return}:
\begin{defN}[Individual treatment effect /  abnormal firm return]
  
\begin{equation}
    \tau_{i}(s,t) = \underbrace{R_{i,t}(s)}_{ \substack{\text{observed for}\\ 
               \text{ treated firm}}}  - \underbrace{R_{i,t}(\infty)}_{\substack{\text{unobserved}\\ 
               \text{counterfactual}}}.
\end{equation}
\end{defN}

For a firm that has the event occur in period $s$, $R_{it}(s)$ is observed, and hence is identified. But, $R_{it}(\infty)$ is not. Indeed, in asset pricing, the challenge of modeling the exact return for an individual asset is viewed as an near-impossible task, even with a structural model. Instead, a large number of asset pricing papers focus on the challenge of estimating the \emph{average} return for firms given a set of characteristics and/or risk factors \parencite[E.g.][]{chamberlain1983arbitrage,CONNOR198413,fama1993common, ross2013arbitrage, kelly2019characteristics, bryzgalova2025forest}. 

This focus on expected returns makes causal inference and asset pricing models happy bedfellows. The inability to known the exact counterfactual return is known as the \emph{fundamental problem of causal inference} and leads to a focus on other alternative estimators, often constructing \emph{average} counterfactual returns for a group of treated units.

A significant body of empirical work and legal scholarship focuses on identifying the effect of events on single firms' valuations, since these valuations are used in litigation to estimate damages \parencite{baker2020machine}. But our view is that in academic research studying financial event studies, a much more natural estimand to target is the \emph{average} treatment effect on the treated (ATT), using many treated firms to estimate an overall average effect, rather than the effect on a single firm. We view the estimated abnormal returns for single firm events as case studies of a much more stable design that focuses on the \emph{average} effect.

\begin{defN}[Cohort-Period Average Treatment Effect on the Treated (ATT)]
    Let the average treatment effect on returns in period $t$ for firms treated in period $s$ be
\begin{equation}
\tau(s, t)^{ATT} = E(\tau_{i}(s,t) \mid T_{i} = s) = E(R_{i,t}(s) \;-\; R_{i,t}(\infty) \mid T_{i} = s).
\end{equation}
\end{defN}

This cohort-period ATT describes the effect of a treatment happening in period $s$ during period $t$ \emph{for those firms who are experience the period $s$ event}. If these firms are special in some way, then this may not be the same effect for other firms (for example, if these firms are riskier, and the effect differs by risk profile). 

\subsubsection{Event-Time Analysis}

These cohort-period ATTs can be combined in a number of ways. Most crucially for our results, combining event cohorts to study effects relative to an event time will average across different event timings. The average treatment effect $\kappa$ periods after an event is:
\begin{equation}
\theta_\kappa^{ATT} = \sum_{s \in \mathcal{S}} w_s \cdot \tau^{ATT}(s, s+\kappa)
\end{equation}
where $w_s$ represents the weight on event cohort $s$. A natural choice is $w_s = N_s/\sum_{s'} N_{s'}$, where $N_s$ is the number of securities with $T_i = s$.

Many empirical papers studying these announcements are interested in cumulating the effects. The Cumulative Average Treatment Effect (CATT), analogous to cumulative abnormal returns (CAR), from event time 0 to $H$ is:
\begin{equation}
\theta_H^{CATT} = \sum_{\kappa=0}^H \theta_\kappa^{ATT}
\end{equation}

In this paper, we focus on these linear transformations of the ATT because they are well-behaved econometrically. However, an alternative approach to cumulative arithmetric returns is the buy-and-hold abnormal return, which we discuss briefly here to highlight its econometric challenges.

\subsubsection{Geometric Returns and Buy-and-Hold Abnormal Returns}
Announcement effects are often cumulated using \emph{buy-and-hold} returns, which correspond to geometric returns. The usual approach for defining abnormal buy-and-holds returns in the literature differences out the buy and hold return of a counterfactual portfolio or stock \parencite{savor2009stock, barber1997detecting} from a stock's buy-and-hold return. In our setting, this is analogous to 
\begin{equation}
    \prod_{\kappa=0}^{H} (1+R_{i,s+\kappa}(s)) - \prod_{\kappa=0}^{H} (1+R_{i,s+\kappa}(\infty)).
\end{equation}
This object is challenge to analyze analytically, and has many challenging statistical properties \parencite{barber1997detecting, mitchell2000managerial}. 

In our notation, this corresponds to the following geometric estimands. Let the cohort-horizon geometric ATT for cohort $s$ at horizon $H$ as
\begin{align*}
    \tau^{geo,ATT}(s,H) &= E(\log(\prod_{\kappa=0}^{H} (1+R_{i,s+\kappa}(s)) - E(\log(\prod_{\kappa=0}^{H} (1+R_{i,s+\kappa}(\infty))))\\ 
    &= \sum_{\kappa=0}^{H} E(\log(1+R_{i,s+\kappa}(s))) -E(\log(1+R_{i,s+\kappa}(\infty))).
\end{align*}
As with the arithmetic ATT, this can be averaged over the event timings:\footnote{Note that in the case of the geometric cumulative return, we first cumulate over the holding period, and then average across periods, since the non-linear structure makes the two non-interchangeable. }
\begin{equation*}
\theta^{geo,ATT}_{H} = \sum_{s} w_{s} \tau^{geo,ATT}(s,H)
\end{equation*}

Note that this estimand effectively studies the \emph{percentage} difference in gross cumulative returns, rather than \emph{level} difference in gross cumulative returns. For researchers interested in sign tests (e.g. positive or negative long-run returns), both objects work equally well.

We now present a result tying the arithmetic (abnormal return)  and geometric (buy and hold) ATT together.

\begin{lemma}
\label{lem:geo_vol}
    The following holds for $\tau^{geo,ATT}(s,t)$ under all models of $R_{it}(\infty)$:
    \begin{align}     
    \theta^{geo,ATT}_{H} &= \theta^{ATT}_{H} - \sum_{s}w_{s}\sum_{\kappa=0}^{H} \left[E(R_{i,s+\kappa}(\infty)\tau_{i}(s,s+\kappa) + \frac{1}{2}\tau_{i} (s,s+\kappa)^2 \mid  T_{i} = s)\right].
    \end{align}
    If treatment effects and control return are independent across cohort $s$, such that we can write $\mu = E(R_{i,s+\kappa}(\infty) | T_{i} = s)$ for all $\kappa$ and $s$, then this can be simplified to  
    \begin{align}     
    \theta^{geo,ATT}_{H} &= (1-\mu)\theta^{ATT}_{H} - \sum_{\kappa=0}^{H} \frac{1}{2}Var(\theta^{ATT}_\kappa) - \frac{1}{2}(\theta^{ATT}_{\kappa})^{2}
    \end{align}
\end{lemma}

This result shows that buy-and-hold returns incorporate both volatility drag and the interaction between base returns and treatment effects, making them more complex to analyze than arithmetic returns. An important implication of this is if the counterfactual return $\hat{R}$ chosen for $R_{it}(s)$ identifies $E(R_{it}(\infty) | T_{i} = s)$, it may be a \emph{bad} counterfactual for buy-and-hold returns because \emph{it does not match on volatility}. For example, a portfolio with identical returns to $E(R_{it}(\infty)|T_{i}=s)$ may have much lower variance (due to diversification). As a result, the volatility drag from the \emph{treated} observed units will bring down the geometric returns, even in the absence of any true effect!

As a result of Lemma \ref{lem:geo_vol}, we focus on estimating the arithmetic ATT, rather than approximating the buy-and-hold return.\footnote{The issues raised here are analogous to problems in difference-in-difference for log vs. level outcomes. If the parallel trends assumption holds for a level outcome, than it almost surely cannot equivalently hold for a log outcome, unless the treatment is randomly assigned. \parencite{roth2023parallel}} Geometric returns would require a counterfactual return portfolio that matches on \emph{both} level and variance, and since the variance of a portfolio does not have the same theoretical guidance for a model as expected returns, finding this counterfactual portfolio is quite hard. It also suggests that papers that use buy-and-hold abnormal returns may contaminate their results as a function of how many firms are included in the counterfactual return portfolio due to diversification differences.

\subsubsection{Factor Model Structure}

We now operationalize our model for $E(R_{it}(\infty) | T_{i} = s,)$,  based on a long literature in asset pricing \parencite{chamberlain1983arbitrage, CONNOR198413}.

\begin{assumption}[Linear Factor Model]
\label{assn:factor_model}
    In the absence of the event, the average return of the portfolio of assets exposed to the event in period $s$ follows a linear factor model with intercept $\alpha_{i}$, $K$ time varying factors $\mathbf{F}_{t}$ and factor weights $\mathbf{\beta}_{i}$, such that
    \begin{equation}
        E\left(R_{it}(\infty) \mid  T_{i} = s \right) = \alpha_{s} + \beta_{s}\mathbf{F}_{t},
    \end{equation}
    where $\alpha_{s} = E(\alpha_{i} | T_{i} = s), \beta_{s} = E(\beta_{i} | T_{i} = s)$.
\end{assumption}

Note that the linear factor assumption is quite strong. For example, it does not allow for changing factor loadings \parencite{barberis2005comovement}. It also does not allow for the market to \emph{anticipate} an event (rationally) in the future if the event does not eventually occur.\footnote{This issue is considered in a series of papers in the finance literature, e.g. \textcite{prabhala1997conditional}, that consider conditional events.} However, it nests generally almost all financial event study methods, such as using the market model, CAPM, or Fama-French factors to construct the counterfactual return \parencite{campbell1997econometrics}. It is also possible that this model could only hold for a short period of time, allowing for varying loadings over a longer period of time (as in \cite{kelly2019characteristics}).

\subsubsection{Identification Assumptions}

\begin{assumption}[Limited Anticipation]
\label{assn:limited_anticipation}
For some known $\delta \geq 0$,
\begin{equation}
R_{i,t}(s) = R_{i,t}(\infty) \quad \text{for all } t < s - \delta
\end{equation}
\end{assumption}

\begin{remark}
    Since we can write
    \begin{align}
         E(R_{i,t}(T_{i})\mid T_{i} = s, \boldsymbol{F}_{t}) &=  E(R_{i,t}(\infty)\mid T_{i} = s, \boldsymbol{F}_{t})  + \tau(t,s)^{ATT}\\
         &=  \alpha_{s} + \beta_{s}\boldsymbol{F}_{t} + \tau(s,t)^{ATT},
    \end{align}
    \Cref{assn:limited_anticipation} implies that $\tau(s,t)^{ATT} = 0$ for all $t < T_i - \delta$. This means that the event has no impact on the returns of the treated group prior to $\delta$ periods before the event. Setting $\delta > 0$ allows for some pre-event information leakage, while $\delta = 0$ assumes no anticipation.
\end{remark}

\begin{assumption}[Event Assignment]
\label{assn:event_assignment}
The probability that security $i$ experiences an event at time $t$ is:
\begin{equation}
p_t(\mathbf{X}_i, \mathbf{F}) = \Pr(T_i = t \mid \mathbf{X}_i, \mathbf{F})
\end{equation}
where $\mathbf{X}_i = (\alpha_i, \boldsymbol{\beta}_i)$ represents security characteristics and $\mathbf{F} = (\mathbf{F}_1, \ldots, \mathbf{F}_T)$ represents all factor realizations.
\end{assumption}

Two important special cases are:
\begin{itemize}
\item \textbf{Random assignment}: $p_t(\mathbf{X}_i, \mathbf{F}) = p_t(\mathbf{F})$ (event assignment independent of security characteristics)
\item \textbf{Random timing}: $p_t(\mathbf{X}_i, \mathbf{F}) = p_t(\mathbf{X}_i)$ (event timing independent of factor realizations)
\end{itemize}

These assumptions formalize when simple estimators will be unbiased and when more sophisticated methods are needed.

This assumption implies that the treated group cannot have an impact from the announcement for a sufficient window prior to the date of the release. There is obvious evidence in the finance literature of hidden information leaking out, with prices responding beforehand (e.g. \textcite{schwert1996markup}). Indeed, this is often pointed to evidence for the strong version of the efficient markets hypothesis. Hence, limited anticipation will be necessary to set a benchmark for when leakage has not yet occurred. This will  allow the researcher to identify the periods in which we can estimate the counterfactual returns. This is the assumption necessary to use the pre-event estimation window commonly used in financial event studies \parencite{campbell1997econometrics,kothari2007econometrics}. 

However, it is important to distinguish between selection into the treatment (e.g. $\{R_{it}(s)\}_{s\in\mathcal{S}}$ being correlated $T_{i}$) and anticipation of the treatment. The former is quite plausible, as we see in our analysis of the S\&P 500 index inclusion effect in \Cref{sec:index_inclusion} -- firms that are growing and having a large market cap are more likely to be selected into the S\&P. The latter will bias our estimates of the true treatment effect, and can be caused by market participants anticipating the event. 

\subsection{Estimators}

\label{sec:estimators}

We now present four sets of estimators and characterize the conditions under which they identify the ATT. In all cases, we assume returns are already adjusted for the risk-free rate. 

\subsubsection{The Abnormal Returns Approach}

Consider first the canonical abnormal returns model used in finance research \parencite{campbell1997econometrics, BROWN19853}. The researcher begins by selecting a set of observable factors $F^o_t$ and estimates factor loadings $(\hat{\alpha}_i, \hat{\beta}_i)$ using ordinary least squares on data prior to $T_i - \delta$:

\begin{equation}
R_{it} = \alpha_i + \beta_i F^o_t + \varepsilon_{it}, \quad t < T_i - \delta
\end{equation}

These estimates $\hat{\alpha}_i$ and $\hat{\beta}_i$ minimize squared prediction errors for stock $i$'s returns using the observed factors. The factors $F^o_t$ may include no factors, a single factor (the market return), or multiple factors (e.g., Fama-French factors).

\begin{defN}[Abnormal Returns Estimator]
\label{def:abnormal_returns}
Define the predicted return for stock $i$ at time $t$ as $\hat{R}_{it} = \hat{\alpha}_i + \hat{\beta}_i F^o_t$. The abnormal return is:
\begin{equation}
AR_{it} = R_{it} - \hat{R}_{it}
\end{equation}
The cohort-period abnormal return estimator is:
\begin{equation}
\tau^{AR}(s,t) = \mathbb{E}(AR_{i,t} | T_i = s) = \mathbb{E}(R_{i,t} | T_i = s) - \mathbb{E}(\hat{R}_{it} | T_i = s)
\end{equation}
\end{defN}

This approach attempts to remove the component of returns attributable to systematic factor exposure, leaving only the ``abnormal'' component. Under correct specification of the factor structure ($F^o_t = F_t$ for all relevant factors), this abnormal component should isolate the treatment effect. However, when factors are omitted or mismeasured, the estimated loadings $\hat{\beta}_i$ may fail to capture the true exposures $\beta_i$, leading to bias.

\subsubsection{Alternative Approaches}

We compare the abnormal returns approach to three alternatives estimation approaches.

\begin{defN}[Difference-in-Means Estimator]
\label{def:diff_in_means}
The difference-in-means estimator compares average returns of treated securities to a control group:
\begin{align}
\hat{\tau}^{cont}(s,t) &= \mathbb{E}(R_{i,t} | T_i = s) - \mathbb{E}(R_{i,t} | i \in C) \\
\hat{\theta}^{cont}_\kappa &= \sum_{s \in S} w_s \hat{\tau}^{cont}(s, s+\kappa)
\end{align}
\end{defN}

When the control group consists of all securities weighted by market capitalization, this estimator corresponds to the ``market-adjusted-return model'' of \textcite{campbell1997econometrics} and \textcite{BROWN19853}. Alternatively, the control group might consist of matched firms selected based on observable characteristics, as in \textcite{barber1997detecting} and \textcite{loughran1997long}. 

Second, we consider a synthetic control estimator \parencite{abadie2021introduction} that uses the pre-event data to construct a synthetic control group:

\begin{defN}[Synthetic Control Estimator]
\label{def:synth_control}
Let $R_{s,t} = \mathbb{E}(R_{it} | T_i = s)$ denote the average return of securities treated at time $s$. The synthetic control estimator constructs a weighted portfolio of control securities to match the pre-event return path of the treated portfolio:
\begin{align}
\hat{\tau}^{synth}(s,t) &= R_{s,t} - \sum_{j \in C} \hat{\omega}_j R_{j,t} \\
\hat{\theta}^{synth}_\kappa &= \sum_{s \in S} w_s \hat{\tau}^{synth}(s, s+\kappa)
\end{align}
where the weights $\hat{\omega}_j$ solve:
\begin{equation}
\hat{\omega} = \arg\min_\omega \sum_{t < s-\delta} \left( R_{s,t} - \sum_{j \in C} \omega_j R_{j,t} \right)^2
\end{equation}
and are subject to a non-negativity constraint: $\omega_j \geq 0$.
\end{defN}

The synthetic control method originated in \textcite{abadie2003economic, abadie2010synthetic} and has expanded and grown as a method over the last decade. Synthetic control directly constructs a counterfactual by matching the pre-event dynamics of treated securities using a portfolio of controls.
This synthetic control is then used as a counterfactual return following the event. 

The key distinction from abnormal returns is that synthetic control does not require the researcher to specify or estimate the underlying factor structure. Instead, it searches for portfolio weights that replicate the treated group's returns in the pre-period, effectively letting the data determine the appropriate factor exposures. If such a replicating portfolio exists, it should continue to provide valid counterfactual returns in the post-period (absent the treatment).

We could depart from the original synthetic control applications by allowing negative weights. Traditionally, synthetic control methods restrict $\omega_j \geq 0$ to ensure the counterfactual represents a convex combination of control units. However, this restriction is unnecessarily limiting in financial applications. Allowing negative weights permits short positions and significantly expands the set of achievable factor loadings, making it more likely that a replicating portfolio exists. This flexibility is natural in financial markets and consistent with standard long-short portfolio construction. However, absent this restriction, we are not able to prove our results on unbiasedness using results from \textcite{ferman2021properties}.

We focus on constructing a single synthetic control for the portfolio of treated securities ($R_{s,t}$) rather than constructing separate synthetic controls for each individual security. This choice reflects both practical and theoretical considerations. Empirically, individual stock returns contain substantial idiosyncratic noise that would make firm-by-firm matching challenging. Theoretically, our estimands target average treatment effects for groups of securities, not individual effects, making portfolio-level analysis natural. This approach follows very naturally the approach advocated in \textcite{ben2022synthetic} for staggered synthetic control.

In practice, perfect pre-period fit may not be achievable. Extensions by \textcite{abadie2021penalized, ben2021augmented, ben2022synthetic} allow for approximate rather than exact matching, trading off pre-period fit against overfitting concerns. However, most importantly for our analysis in financial event studies, \textcite{ferman2021properties} shows that if the data follows a linear factor structure, then with sufficient pre-event time periods and control units, the estimator is consistent. This is consistent with a wide-range of asset pricing work highlighting the importance of having assets that span risk factors \parencite{giglio2021asset, giglio2025test}.There is also a close connection to the mimicking-portfolio approach \parencite[E.g.][]{huberman1987mimicking}.


We also consider a third estimator, following \textcite{xu2017generalized}, which uses PCA regression with cross-validation to estimate a factor structure with unknown factors:
\begin{defN}
    \label{def:gsynth}
The Gsynth approach assumes that non-treated stocks follow an interactive fixed effects model:
\begin{equation}
    R_{it}(\infty) = \alpha_i + \boldsymbol{\lambda}_i' \mathbf{F}_t  + \varepsilon_{it}
\end{equation}
where $\mathbf{F}_t$ are $r$ unobserved common factors and $\boldsymbol{\lambda}_i$ are unit-specific factor loadings.

The estimation proceeds in three steps:

\textbf{Step 1: Initial Factor Estimation}
Using only control units, estimate factors via principal components:
\begin{equation}
    (\hat{\mathbf{F}}, \hat{\boldsymbol{\Lambda}}) = \arg\min_{\mathbf{F}, \boldsymbol{\Lambda}} \sum_{i \in \mathcal{C}} \sum_{t=1}^T (R_{it} - \alpha_i - \boldsymbol{\lambda}_i' \mathbf{F}_t)^2
\end{equation}
subject to normalization constraints $\mathbf{F}'\mathbf{F}/T = \mathbf{I}_r$ and $\boldsymbol{\Lambda}'\boldsymbol{\Lambda}$ diagonal.

\textbf{Step 2: Cross-Validation for Model Selection}
Select the number of factors $r$ via cross-validation:
\begin{equation}
    \hat{r} = \arg\min_{r \in \{1,\ldots,r_{max}\}} CV(r)
\end{equation}
where $CV(r)$ is the cross-validated mean squared prediction error using pre-treatment periods for the treatment group.

\textbf{Step 3: Counterfactual Construction}
For each treated unit $i$ with $T_i = s$:
\begin{enumerate}
    \item Estimate unit-specific loadings using pre-treatment data:
    \begin{equation}
        \hat{\boldsymbol{\lambda}}_i = \arg\min_{\boldsymbol{\lambda}} \sum_{t < s - \delta} (R_{it} - \alpha_i - \boldsymbol{\lambda}' \hat{\mathbf{F}}_t)^2
    \end{equation}
    \item Construct counterfactual for post-treatment periods:
    \begin{equation}
        \hat{R}_{it}^{GSC}(\infty) = \hat{\alpha}_i + \hat{\boldsymbol{\lambda}}_i' \hat{\mathbf{F}}_t
    \end{equation}
\end{enumerate}

The treatment effect estimate is:
\begin{equation}
    \hat{\tau}^{GS}(s,t) = \frac{1}{N_s}\sum_{i: T_i = s} \left(R_{it} - \hat{R}_{it}^{GS}(\infty)\right)
\end{equation}

\end{defN}

The Gsynth estimator more directly leans on the linear factor structure, but does not require knowing the true factors, and uses the set of control firms to construct the set of counterfactual returns. 

We focus on these two alternative estimators, but other alternative methods, such as IPCA \textcite{kelly2019characteristics} or the three-pass method in \textcite{giglio2021asset} may work as well or better. We leave it to future work to consider what approaches may work best.

\subsection{Theoretical Results}
\label{sec:theory}

We now establish conditions under which these estimators identify the ATT. For this proposition, it is convenient to see how these estimators differ from the target single event-period estimand:
\begin{align}
\tau^{AR}(s,t) - \tau^{ATT}(s,t) &= (\alpha_s - \hat{\alpha}_s) + (\beta_s F_t - \hat{\beta}_s F^o_t) + \varepsilon_{st} \label{eq:bias_ar} \\
\hat{\tau}^{cont}(s,t) - \tau^{ATT}(s,t) &= (\alpha_s - \alpha_\infty) + (\beta_s - \beta_\infty) F_t + (\varepsilon_{st} -\varepsilon_{\infty,t}) \label{eq:bias_cont} \\
\hat{\tau}^{alt}(s,t) - \tau^{ATT}(s,t) &= (\alpha_s - \hat{\alpha}^{alt}_s) + (\beta_s - \hat{\beta}^{alt}_s) F_t + \varepsilon_{st}\label{eq:bias_synth}
\end{align}
where $\alpha_s = \mathbb{E}(\alpha_i | T_i = s)$, $\beta_s = \mathbb{E}(\beta_i | T_i = s)$  are the average intercept and factor loadings for treated securities, $\alpha_\infty$ and $\beta_\infty$ are corresponding quantities for the control group, $\hat{\alpha}_s$ and $\hat{\beta}_s$ are the estimated loadings from the abnormal returns approach, and $\hat{\alpha}^{alt}_s$ and $\hat{\beta}^{alt}_s$ are the implied loadings from either the synthetic control or gsynth estimator. $\varepsilon_{st} = n^{-1}_{s}\sum \varepsilon_{it}$ is the average idiosyncratic noise for the $i$ cohort, and $\varepsilon_{\infty,t} = n^{-1}_{i \in \mathcal{C}}\sum varepsilon_{it}$ is the average noise for the control group.

\begin{proposition}[Single Event Finite Sample and Asymptotic Bias]
\label{prop:consistency}
Let Assumptions \ref{assn:factor_model}, and \ref{assn:limited_anticipation} hold. Then:

\begin{enumerate}
\item \upshape{Asymptotic properties}. As $n_s, n_c, T_{pre} \to \infty$:
\begin{align}
\tau^{AR}(s,t) - \tau^{ATT}(s,t) &\xrightarrow{p} (\alpha_s - \tilde{\alpha}_s) + (\beta_s F_t - \tilde{\beta}_s F^o_t) \label{eq:plim_ar} \\
\hat{\tau}^{cont}(s,t) - \tau^{ATT}(s,t) &\xrightarrow{p} (\alpha_s - \alpha_\infty) + (\beta_s - \beta_\infty) F_t \label{eq:plim_cont} \\
\hat{\tau}^{synth}(s,t) - \tau^{ATT}(s,t) &\xrightarrow{p} 0 \label{eq:plim_synth}
\end{align}
where $\tilde{\alpha}_s$ and $\tilde{\beta}_s$ are probability limits of the estimated parameters.

\item Under random assignment ($p_t(\mathbf{X}_i, \mathbf{F}) = p_t(\mathbf{F})$),as  $n_s, n_c \to \infty$, the difference-in-means estimator is consistent even with fixed $T_pre$:
\begin{equation}
\hat{\tau}^{cont}(s,t) - \tau^{ATT}(s,t) \xrightarrow{p} 0
\end{equation}

\item Under correct specification, ($F^o_t = F_t$ for all $t$), the abnormal returns estimator is consistent as $n_{s}, T_{pre} \to \infty$:
\begin{equation}
\tau^{AR}(s,t) - \tau^{ATT}(s,t) \xrightarrow{p} 0
\end{equation}
\end{enumerate}
\end{proposition}

All proofs are in \Cref{appsec:proofs}.

The most complex part of this proof, proof of asymptotic unbiasedness of the synthetic control estimator, follows directly from \textcite{ferman2021properties}, who show that the synthetic control estimator is asymptotically unbiased under the assumption of an unknown linear factor model and many control units. The results for Gsynth also follow directly from \textcite{xu2017generalized}. The other two results follow from the assumptions and the definition of the estimators.

\begin{remark}
Both the misspecified abnormal return estimator and the difference-in-means estimator in a given time period are inconsistent. Both converge to a random variable that is a linear combination of the two factors, but the linear combination varies depends on the factor loadings (and factor correlation). These inconsistencies are similar to the inconsistencies highlighted in Theorem 1 of \textcite{andrews2005cross}. In contrast, the synthetic control estimator is consistent and converges to the true effect. If the abnormal return estimator is correctly specified, then it is also consistent. If the treatment is randomly assigned, then the difference-in-means estimator is also consistent, since $\beta_{s} = \beta_{\infty}$ in the limit. This implies that these estimators are highly susceptible to coincident shocks at the same time, and the inference will be highly suspect (hence the need to cluster on event-timings in many financial event studies).
\end{remark}

Of course, intuitively, in many applications the factor loadings are often not too large, and the underlying risk premia are, on average, typically small relative to $\tau^{ATT}(s,t)$. For example, the one-day index inclusion effect is estimated to be somewhere between 1-4\%, depending on the time period. By comparison, the market return is, on average, 0.05\%, two orders of magnitude smaller than the treatment effects. 

However, there are many periods when the market return can be far larger, such as during periods of market volatility.  There is substantial variation in the size of these factors, with an interquartile range of 1\% and very large fat tails. Hence, the correlation of the factors with the timing of the event is very important. This will be  apparent in our first empirical example of \cite{acemoglu2016value}. As a result, this bias can be quite large.  Formally, we can write the following from \Cref{thm:bias}:

\begin{corollary}
Fix $\alpha_{s}, \beta_{s}, \hat{\alpha}_{s}, \hat{\beta}_{s}$ and $F_{t}^{o}$. Then, $|\tau^{AR}(s,s) - \tau^{ATT}(s,s)|$ is increasing in $|F_{t}|$ for entries where $\beta_{s}$ is non-zero.
\end{corollary}





We next consider how these results change if there are multiple event periods. 


\begin{theorem}[Bias with multiple events]
\label{thm:bias}
    Let \Cref{assn:factor_model} and \ref{assn:limited_anticipation} hold.
    \begin{enumerate}
        \item If $n_{s}, n_{c}, T_{pre} \to \infty$, then asymptotically, the synthetic control and gsynth estimators are unbiased, 
        \begin{align}
     \hat{\theta}_{\kappa}^{alt} -  \theta_{\kappa}^{ATT} &\rightarrow_{p} 0.
        \end{align}
        where $alt \in \{synth, gs\}$.
        \item If $|\mathcal{S}| > 0$ and $1 > p_{t}(\mathbf{X}_{i}, \mathbf{F}) > \epsilon > 0$, then if $n_{s}, n_{c}, T_{pre} \to \infty$, the other two estimators are biased and converge to a weighted combination of conditional expected risk premia across the event periods:
        \begin{align}
            \hat{\theta}^{ar}_{\kappa}-\theta_{\kappa}^{ATT}&=  E\left( (\alpha_{s} - \tilde{\alpha}_{s}) + (\beta_{s}\boldsymbol{F}_{s+\kappa} - \tilde{\beta}_{s}\boldsymbol{F}^{o}_{s+\kappa})\mid T_{i} \in \mathcal{S} \right)\\
            \hat{\theta}^{cont}_{\kappa} -\theta_{\kappa}^{ATT}  &=  E\left((\alpha_{s} - \alpha_{\infty}) + (\beta_{s}-\beta_{\infty})F_{t}\mid T_{i} \in \mathcal{S} \right)
        \end{align}
        \item If random assignment across firms holds,   then for $n_{s}, n_{c} \to \infty$, 
         \begin{align}
     \hat{\theta}_{\kappa}^{cont} -  \theta_{\kappa}^{ATT} &\rightarrow_{p} 0.
        \end{align}
        \item If random timing holds, for $n_{s}, T_{pre} \to \infty$,then 
        \begin{align}
            \hat{\theta}^{ar}_{\kappa}-\theta_{\kappa}^{ATT}&=  E\left(\alpha_{s} - \tilde{\alpha}_{s}\mid T_{i} \in \mathcal{S} \right)\\
            &+ E\left(\beta_{i}\mid T_{i} \in \mathcal{S} \right)E\left(\boldsymbol{F}_{t} \right) - E\left(\tilde{\beta}_{i}\mid T_{i} \in \mathcal{S} \right)E\left(\boldsymbol{F}^{o}_{s+\kappa} \right)
            \end{align}
            and  for $n_{s}, n_{c} \to \infty$, 
            \begin{align}
            \hat{\theta}^{cont}_{\kappa} -\theta_{\kappa}^{ATT}  &=  E\left(\alpha_{s} - \alpha_{\infty}\mid T_{i} \in \mathcal{S} \right) + E\left(\beta_{s}-\beta_{\infty}\mid T_{i} \in \mathcal{S} \right)E\left(\mathbf{F}_{t}\right)
        \end{align}
    \end{enumerate}
   
\end{theorem}

An implication of this is that the abnormal returns estimator is can be quite close to the true treatment effect, even when the factor model is misspecified. Moreover, this bias could be small even for a model that ignores factors, consistent with the simulation evidence  in \textcite{BROWN19853} that the form of the abnormal return estimator has limited effects on the estimates.\footnote{The simulations in \textcite{BROWN19853} are such that the event days are exactly randomly assigned across time: ``Each time a security is selected, a hypothetical event day is generated. Events are selected with replacement, and are assumed to occur with equal probability on each trading day from July 2, 1962, through December 31, 1979.''} In fact, a common phrase described in event studies is that the structure of the model used in $\tau^{AR}$ does not have significant impacts on the estimated effects. For example, in footnote 5, \textcite{shleifer1986demand} states ``The [index inclusion] results were not materially different when returns were not corrected for market movements. Similarly, combining the before and after estimation periods did not make much difference.'' Or in \textcite{edmans2012link} ``I use the standard short event-study window so that the calculation of abnormal returns is relatively insensitive to the benchmark asset pricing model used.''  


\subsection{Increasing cumulative bias in long-run event studies}

Researchers are often interested in the trends or cumulative impact of events on returns, as measured by cumulative abnormal returns or buy-and-hold abnormal returns. This gets mapped to different economic and behavioral theories about how the market processes information (e.g.  \textcite{daniel1998investor} is a theory to explain these effects from a behavioral perspective; \textcite{kwon2022extreme} consider 90 day post-announcement effects relative to announcement day effects). 

Some papers have pointed to flaws in studying these types of long-run perspectives -- for example, Mitchell and Stafford (2000) highlight the flaws in the inference around long-run abnormal return studies of firm activity. As we show in \Cref{lem:geo_vol}, the buy-and-hold abnormal return has additional challenges caused by variance considerations in the counterfactual portfolio. We now use our results in \Cref{prop:consistency} and \cref{thm:bias} to show that even estimating arithmetic cumulative abnormal returns in the long-run amplifies the misspecification bias. 

Following the analogy principle for the CATT, \Cref{thm:bias} tells us that the bias in for the abnormal return and difference-in-mean estimators is intimately related to the cumulative sum of the factor premium. Under random timing and many events, the bias for the CATT at horizon $H$ is  $H E(\beta_{i} \mid  T_{i} \in \mathcal{S})E(F_{t})$. The factors have a positive mean (since the risk of the factors leads to positive expected return), and thus the bias in the estimator will drift proportional to the expected value of the factors during the time period, scaled by the relative estimation error in $E(\beta_{i} \mid  T_{i} \in \mathcal{S})$.

Consider estimating the long-run impact of a merger on stock market prices. \textcite{RAGHAVENDRARAU1998223} find a three-year long run effect of -4\% for all mergers, while \textcite{savor2009stock} find a three-year long-run effect of -13.1\% for stock-financed mergers and 1.6\% for cash financed mergers. These results are well-motivated by \textcite{shleifer2003stock}, but their magnitude may reflect bias due to the errors in $\hat{\beta}$:
\begin{equation}
    (\beta_{s}-\hat{\beta}_{s})E(F_{t}) = \sum_{k=1}^{K}(\beta_{sk} - \hat{\beta}_{sk})E(F_{tk}).
\end{equation}

If $K=1$, for example, and was equal to the market, then our expected excess return is 6\%. If $\beta_{sk} - \hat{\beta}_{sk}$ was $-0.1$, then at the three year level, we might expect a bias of -1.8\%. This is of course an empirical question of which way the biases would go; is the constructed portfolio of firms too heavily loaded on risk factors? 

Note that these issues are not solved by using multiple event timings. This bias in factors cannot average out to zero, and so the only source by which we can achieve zero bias is through mean zero differences in the loadings.

It is also worth remarking how the results from Mitchell and Stafford (2000) can be seen analytically in our statistical terms. While the misspecification term $ (\beta_{s}-\hat{\beta}_{s})\sum_{\kappa=0}^{K^{0}}F_{s+\kappa}$ creates bias, it also creates cross-correlation in errors for every event-timing.\footnote{As they state: ``[M]ajor corporate events cluster through time by industry. This leads to positive cross-correlation of abnormal returns, making test statistics that assume independence severely overstated.'' }


\subsection{Key takeaways re: randomness}

Key takeaways for practitioners are four-fold:

\begin{enumerate}
    \item If treatment is randomly assigned across firms, then comparing returns to the average of the market is as good as any other approach. 
    \item If treatment is randomly assigned across periods, and there are multiple event timings, then the model used to estimate effects does not matter in the short-run.
    \item If treatment is randomly assigned across periods, but the model used to estimate effects is misspecified, then the estimates will be biased, even with many event timings.
    \item These results are identical whether there is a single treated firm or many treated firms. 
\end{enumerate}

\subsection{Individual estimates are noisy, but not necessarily biased}
We briefly discuss the case of a single firm being treated. To analyze this case, we need to allow for slightly more flexibility in our notation. 
\begin{assumption}
    Let $R_{it}(\infty) = \alpha_{i} + \beta_{i}\boldsymbol{F}_{t} + \varepsilon_{it}$, where $\varepsilon_{it}$ is i.i.d. across firms, and i.n.i.d. across time, and mean zero. 
\end{assumption}
\begin{remark}
    This assumption implies we can write $R_{it}(T_{i}) = R_{it}(\infty) + \tau_{i}(s,t) = \alpha_{i} + \beta_{i}\boldsymbol{F}_{t} + \tau_{i}(s,t) +  \varepsilon_{it}$.
\end{remark}

Then, consider the case of a single firm estimated in each estimator:
\begin{equation}
 \tau^{AR}_{i}(s,s) - \tau_{i}(s,s) = (\alpha_{s} - \hat{\alpha}_{s}) + (\beta_{s}\boldsymbol{F}_{s} - \hat{\beta}_{s}\boldsymbol{F}^{o}_{s}) + \varepsilon_{it}.
\end{equation}

Statistically, there are now three objects with randomness to worry about: the estimated parameters, the aggregate factors, and the idiosyncratic variance for the individual firm. Note that with several treated units, this last term disappears, but with a single unit, we have insurmountable noise. This is a common problem flagged in the event studies literature looking at securities litatigation \parencite{baker2020machine}. 

However, consider an approach that estimates many individual treatment effects in this manner (such as \textcite{kogan2017technological}). On, average, these estimates will be subject to the same results outlined above, but each one is quite noisy. This is equivalent to problems associated with estimating many treatment effects. One approach is to consider shrinkage estimators. Another would be to pool the firms based on characteristics of interest, and construct portfolios this way. This would remove $\varepsilon$. 

\section{Simulations}
\label{sec:simulations}
We highlight how the non-random timing and assignment, together with a misspecified factor model, could affect the bias with different estimators of treatment effects, using a simple simulation exercise. In the simulation, the returns follow a two-factor structure, with the second factor omitted in the estimation of abnormal returns. We compare the expected bias, root mean square error, and coverage with random vs. nonrandom assignment and timing.

\subsection{Simulation Design with 2 Factors and Selection}\label{sec:simul_2f_select_design}
We simulate a panel of stock returns with a linear factor structure:
\begin{equation}
    r_{it}=r_{f,t}+\beta_{i,mkt}(r_{mkt,t}-r_{f,t})+\beta_{i,smb}r_{smb,t}+\varepsilon_{i,t}, 
\end{equation}
where the return for each stock equals to the risk-free rate, plus the exposure times risk premium of a market factor and a size factor (small-minus-big), and a stock-level idiosyncratic component.

We assume that both factor loadings follow independent normal distributions: $\beta_{i,mkt},\beta_{i,smb}\sim \mathcal{N}(1,0.3^{2})$. We further assume that the idiosyncratic component of each stock is drawn i.i.d. from a Normal distribution: $\varepsilon_{i,t}\sim \mathcal{N}(0,0.1^{2})$. We choose a standard deviation of around 0.1 so that the residual variance constitutes approximately half of the total variance.

We simulate returns for 500 firms, with pre-treatment period of 239 days, 1 event day, and 10 post-treatment periods. Roughly 10\% of firms are treated, following one of two treatment assignment processes, discussed below. Treated firms get a true effect of 3\% on the treatment day, and nothing afterwards. The factor returns and the risk-free rate are randomly sampled from daily Fama-French returns from July 1926 to 2022 with block sampling to preserve the correlation structure between factors.

\paragraph{Treatment assignment process} We compare expected bias with different treatment assignment selection and timing selection. For firm assignment, we either completely randomly assign the treatment to 10\% of firms, or to instead relax this assumption, we model that the probability of a firm getting treated follows a logit function of the beta on the SMB factor
\begin{equation}
    p(treated)_i=\frac{\exp(\delta \beta_{i,smb})}{1+\exp(\delta\beta_{i,smb})},
\end{equation}
where $\delta=\frac{\log(0.1)}{E(\beta_{i,smb})}<0$ to achieve an average probability of 10\%. The lower the simulated SMB factor loading of the firm, the more likely to be treated. 

For treatment period selection, we similarly use two different assignment mechanisms.  The first is to randomly sample the 250 data periods, and always set the treatment period equal to $t=240$. This effectively makes the treatment period's factor draw uncorrelated with the treated firms' factor loadings. The secon approach with timing selection works as follows. First, we rank the SMB factor in 250 candidate treatment periods. We then use the rank of SMB returns as inputs to the selection function.\footnote{Raw factors returns have positive and negative values with mean close to 0, which will make the logit function highly sensitive.} The probability of any one of the candidate period being the treatment period is 
\begin{equation}
    p(selected)_t=\frac{\exp(\delta Rank_{2t})}{1+\exp(\delta Rank_{2t})},
\end{equation}
where $\delta=\frac{\log(1/250)}{E(Rank_t)}$. We then draw indicator variables for each candidate period from binomial distributions with respective treatment probability in each period. If multiple periods are drawn to be the event period, we use the one with the highest factor realization. Thus, if a period has a high factor realization of the omitted factor, it is more likely to become the treatment period.

\subsection{Simulation Results with 2 Factors and Selection}
In \Cref{tab:simul_2f_select}, we compare the performance of four different estimators across 50 simulations: mean difference between treated and control firms, average abnormal returns using the market factor (estimating the factor loading for each treated firm in the pre-period), average abnormal returns using the both factors (estimating the factor loadings for each treated firm in the pre-period), and average treatment effects from the generalized synthetic control method (Gsynth). Estimated bias is reported in percentage points. We also report the root mean square error (RMSE) and coverage of 95\% confidence intervals.

\begin{table}[H]
\caption{\textbf{Treatment Effect Bias and Coverage in Simulations: Two-Factor Structure} This table presents the bias and coverage of different estimators of treatment effects in financial returns. We simulate 500 firms with 10\% treated. The estimation period is 239 days and post-event period is 11 days. More details on the simulations is in Section \ref{sec:simul_2f_select_design}. Panel A reports simulation results with no selections, Panel B with only assignment selection, Panel C with only timing selection, and Panel D with both. We consider several estimators: difference in simple average, CAPM and 2-factor abnormal returns, and generalized synthetic methods. The expected biases and coverage are from 50 simulations.}\label{tab:simul_2f_select}
\centering
\scriptsize

\begin{tabular}[t]{lrrrrrrrrr}
\toprule
\multicolumn{7}{l}{Panel A: Random Assignment + Random Timing}\\
\midrule
\multicolumn{1}{c}{} & \multicolumn{3}{c}{All Periods} & \multicolumn{3}{c}{Treated Periods} & \multicolumn{3}{c}{Untreated Periods} \\
\cmidrule(l{3pt}r{3pt}){2-4} \cmidrule(l{3pt}r{3pt}){5-7} \cmidrule(l{3pt}r{3pt}){8-10}
Model & E(Bias) & MAD & RMSE & E(Bias) & MAD & Coverage & E(Bias) & MAD & Coverage\\
\midrule
Simple Means & 0.00 & 0.04 & 0.58 & 0.01 & 0.17 & 1 & 0.00 & 0.04 & 0.03\\
CAPM & -0.06 & 0.16 & 2.11 & -0.07 & 0.43 & 1 & -0.06 & 0.17 & 0.44\\
Correct Factor Structure & -0.01 & 0.04 & 0.54 & 0.00 & 0.16 & 1 & -0.01 & 0.04 & 0.04\\
Gsynth (PCA) & 0.00 & 0.04 & 0.56 & 0.02 & 0.17 & 1 & 0.00 & 0.04 & 0.03\\
\midrule
\multicolumn{7}{l}{Panel B: Assignment Selection + Random Timing}\\
\midrule
\multicolumn{1}{c}{} & \multicolumn{3}{c}{All Periods} & \multicolumn{3}{c}{Treated Periods} & \multicolumn{3}{c}{Untreated Periods} \\
\cmidrule(l{3pt}r{3pt}){2-4} \cmidrule(l{3pt}r{3pt}){5-7} \cmidrule(l{3pt}r{3pt}){8-10}
Model & E(Bias) & MAD & RMSE & E(Bias) & MAD & Coverage & E(Bias) & MAD & Coverage\\
\midrule
Simple Means & 0.02 & 0.05 & 0.71 & 0.04 & 0.18 & 1.00 & 0.02 & 0.05 & 0.11\\
CAPM & -0.05 & 0.13 & 1.78 & -0.04 & 0.35 & 0.98 & -0.05 & 0.14 & 0.40\\
Correct Factor Structure & -0.01 & 0.03 & 0.54 & 0.02 & 0.14 & 1.00 & -0.01 & 0.04 & 0.04\\
Gsynth (PCA) & 0.00 & 0.04 & 0.57 & 0.03 & 0.15 & 1.00 & 0.00 & 0.04 & 0.05\\
\midrule
\multicolumn{7}{l}{Panel C: Random Assignment + Timing Selection}\\
\midrule
\multicolumn{1}{c}{} & \multicolumn{3}{c}{All Periods} & \multicolumn{3}{c}{Treated Periods} & \multicolumn{3}{c}{Untreated Periods} \\
\cmidrule(l{3pt}r{3pt}){2-4} \cmidrule(l{3pt}r{3pt}){5-7} \cmidrule(l{3pt}r{3pt}){8-10}
Model & E(Bias) & MAD & RMSE & E(Bias) & MAD & Coverage & E(Bias) & MAD & Coverage\\
\midrule
Simple Means & -0.01 & 0.05 & 0.63 & 0.00 & 0.21 & 1 & -0.01 & 0.05 & 0.05\\
CAPM & 0.25 & 0.27 & 3.49 & 2.71 & 2.71 & 1 & 0.00 & 0.16 & 0.46\\
Correct Factor Structure & -0.02 & 0.04 & 0.54 & 0.00 & 0.12 & 1 & -0.02 & 0.04 & 0.04\\
Gsynth (PCA) & -0.01 & 0.04 & 0.57 & 0.01 & 0.13 & 1 & -0.01 & 0.04 & 0.04\\
\midrule
\multicolumn{7}{l}{Panel D: Assignment Selection + Timing Selection}\\
\toprule
\multicolumn{1}{c}{} & \multicolumn{3}{c}{All Periods} & \multicolumn{3}{c}{Treated Periods} & \multicolumn{3}{c}{Untreated Periods} \\
\cmidrule(l{3pt}r{3pt}){2-4} \cmidrule(l{3pt}r{3pt}){5-7} \cmidrule(l{3pt}r{3pt}){8-10}
Model & E(Bias) & MAD & RMSE & E(Bias) & MAD & Coverage & E(Bias) & MAD & Coverage\\
\midrule
Simple Means & -0.05 & 0.07 & 0.88 & -0.52 & 0.52 & 1 & -0.01 & 0.05 & 0.08\\
CAPM & 0.21 & 0.23 & 2.92 & 2.26 & 2.26 & 1 & 0.00 & 0.13 & 0.40\\
Correct Factor Structure & -0.02 & 0.04 & 0.52 & 0.01 & 0.12 & 1 & -0.02 & 0.04 & 0.05\\
Gsynth (PCA) & -0.01 & 0.04 & 0.56 & -0.01 & 0.14 & 1 & -0.01 & 0.04 & 0.04\\
\bottomrule
\end{tabular}

\end{table}

First, in Panel A, we see that the average bias is small even with the wrong factor structure, if the treatment is randomly assigned. Similarly, in Panel B, if we only have non-random assignment selection, the expected bias is also insignificant on average. However, this masks the variation across simulations - if a time period has a larger factor draw on the treatment day, that leads to much larger bias.

\begin{figure}[H]
\label{appfig:factor_draw_bias_corr}
    \caption{\small\textbf{Bias from CAPM Model on SMB Returns with Assignment Selection}
    This figure plots the biases from a CAPM estimator on the treatment period over realizations of the second factor across 50 simulations. We simulate 500 firms with 10\% of them getting treated. The estimation period is 239 days and post-event period is 11 days. More details on the simulations is in Section \ref{sec:simul_2f_select_design}. Panel A reports simulation results with no selections, Panel B with only assignment selection, Panel C with only timing selection, and Panel D with both. We consider several estimators: difference in simple average, CAPM and 2-factor abnormal returns, and generalized synthetic methods. The expected biases and coverage are from 50 simulations.}\label{fig:simul_2f_biascapm_smb_assign}
    \centering
    \includegraphics[width=0.7\linewidth]{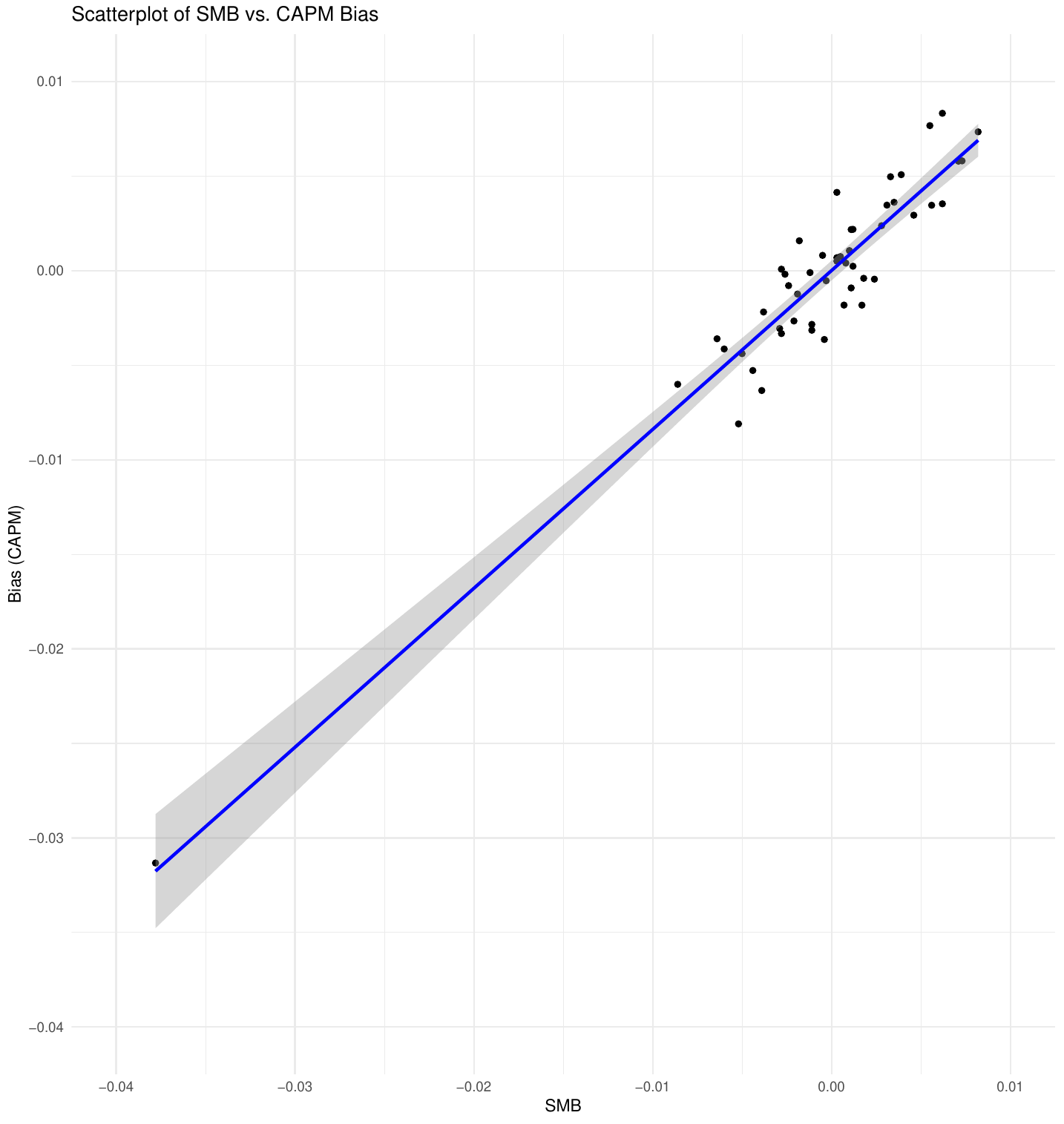}
\end{figure}

In Panel C, we consider random assignment of treatment to units, but non-random event timing.  As in Panel A,  the difference in means is unbiased thanks to the results in \Cref{thm:bias}. Since treatment is uncorrelated with factor loadings, there is no endogeneity and the simple means estimator is an unbiased estimator of the treatment effect. However, with non-random timing, the CAPM model is biased, because the abnormal return (as discussed in Section 2.2) will be the average $\beta$ for the omitted factor multiplied by the largest possible factor draw. In contrast, the difference in means is unbiased because while both treated and untreated firms are exposed to the high factor draw, they have identical factor exposures, which cancels out. For the correctly specified model, the estimated model correctly specifies the counterfactual, and so there is no bias. Finally, the Gsynth estimator is able to identify the correct underlying factor structure, and has limited bias as well. 

Once we have both types of selection in treatment in Panel D, we see that the simple difference in means is now biased. However, it is still less biased in absolute value than the misspecified CAPM model. This is because the \emph{gap} in the treatment and control factor loadings for the simple mean difference is still smaller than the level misspecification in the factor loadings in the CAPM estimation. Again, the Gsynth approach does quite well, with similar performance to the correctly specified factor model.


\section{Applications}
\subsection{Empirical Example 1: Geithner as Treasury Secretary} \label{sec:geithner}

We now turn to our first empirical example, examining the period when the announcement of Timothy Geithner as Treasury Secretary was leaked, following the setup of \textcite{acemoglu2016value}. This example highlights the results of Proposition~\ref{prop:consistency} in a simultaneous treatment setting. We demonstrate that the bias from an incorrect factor structure can be substantial in this setting, and that synthetic control methods help alleviate this bias. We argue that the bias arises from two sources: first, the event window coincides with turbulent market conditions characterized by large daily factor realizations; second, the counterfactual returns are constructed from control firms with substantially different factor exposures. We show that synthetic methods, which greatly reduce these biases, also match the factor loadings of treated firms for known factors such as size and value.

\paragraph{Empirical setup.} We examine the announcement of Timothy Geithner as nominee for Treasury Secretary on November 21, 2008. Following \textcite{acemoglu2016value}, we estimate average treatment effects over the 11-day window encompassing and following the announcement date, from November 21, 2008 (day 0) through December 8, 2008 (day 10).\footnote{November 24, 2008 corresponds to day 1 due to the weekend.} For treated and control bank returns, we use the data provided by the authors, who collected daily returns from Datastream.\footnote{We thank Amir Kermani for providing the replication code and data on his website.} For all trading days before and after the event, returns represent full trading day returns during regular trading hours. For the event day, returns are calculated from 3:00 p.m. (when the news leaked) until market close at 4:00 p.m. 

We consider two sets of control firms. First, we use the same set of financial firms listed on the NYSE or NASDAQ that are not connected to Geithner, as in \textcite{acemoglu2016value}. Second, we expand the control group to include all NYSE, AMEX, and NASDAQ (exchange codes 1--3) common stocks (share codes 10 or 11).

\begin{table}[thb]
    \caption{\small\textbf{ATT of Treasury Secretary Announcement}
    This table presents average treatment effects after the announcement of Timothy Geithner as Treasury Secretary. Event day 0 is November 21, 2008 from 3pm (when the news leaked) to market closing, consistent with \textcite{acemoglu2016value}. The average treatment effect is estimated using post periods from trading day 0 to day 10. We consider two control samples: banks or financial services firms trading on the NYSE or Nasdaq (Panel A), and all NYSE, AMEX, and NASDAQ common stocks (Panel B). We consider several estimators: difference in simple average, difference-in-differences, synthetic control, synthetic DinD, and generalized synthetic methods. Standard errors of simple average is from a two-sample t-test. Standard errors of DID, synthetic control, and synthetic DID are calculated using placebo inference following \textcite{arkhangelsky2021synthetic} with 100 repetitions. Standard errors of Gsynth is computed using parametric bootstrap with 1,000 samples. Standard errors in parentheses. * p<0.10, ** p<0.05, *** p<0.01 }\label{tab:geithner_att}
    \centering
{
\def\sym#1{\ifmmode^{#1}\else\(^{#1}\)\fi}
\makebox[\textwidth][c]{\begin{tabular}{l*{8}{c}}
\multicolumn{9}{l}{Panel A: Bank Controls}\\
\toprule
& (1) & (2) & (3) & (4) & (5) & (6) & (7) & (8) \\
& Average & DID & Market & CAPM & FF3F & SC & SDID & Gsynth \\
\midrule
Schedule connections & 0.026*** & 0.027*** & 0.024*** & 0.016*** & 0.014*** & 0.016*** & 0.018*** & 0.012** \\
& (0.007) & (0.005) & (0.007) & (0.007) & (0.006) & (0.005) & (0.005) & (0.006) \\
Personal connections & 0.029*** & 0.030*** & 0.027** & 0.016 & 0.013 & 0.004 & 0.009** & 0.008 \\
& (0.010) & (0.006) & (0.012) & (0.011) & (0.011) & (0.003) & (0.005) & (0.007) \\
New York connections & 0.019*** & 0.020*** & 0.017*** & 0.011*** & 0.009*** & 0.009*** & 0.012*** & 0.009** \\
& (0.005) & (0.004) & (0.004) & (0.004) & (0.004) & (0.003) & (0.003) & (0.004) \\
\midrule
Observations & 5,995 & 129,165 & 129,165 & 129,165 & 129,165 & 129,165 & 129,165 & 129,625 \\
\midrule
\multicolumn{9}{l}{Panel B: All Firm Controls}\\
\midrule
& Average & DID & Market & CAPM & FF3F & SC & SDID & Gsynth \\
\midrule
Schedule connections & 0.020** & 0.020*** & 0.024*** & 0.016*** & 0.014*** & 0.004 & 0.009 & 0.001 \\
& (0.008) & (0.007) & (0.007) & (0.007) & (0.006) & (0.007) & (0.006) & (0.008) \\
Personal connections & 0.020* & 0.021*** & 0.027** & 0.016 & 0.013 & -0.003 & 0.006 & 0.003 \\
& (0.010) & (0.006) & (0.012) & (0.011) & (0.011) & (0.005) & (0.005) & (0.008) \\
New York connections & 0.011** & 0.011*** & 0.017*** & 0.011*** & 0.009*** & 0.004 & 0.003 & 0.001 \\
& (0.005) & (0.004) & (0.004) & (0.004) & (0.004) & (0.004) & (0.003) & (0.004) \\
\midrule
Observations & 45,045 & 966,420 & 966,420 & 966,420 & 966,420 & 966,420 & 966,420 & 916,388\\
\bottomrule
\end{tabular}}}
\end{table}

\paragraph{Non-connected banks as controls.}
We first use public financial institutions without connections to Geithner as control firms. Panel A of Table~\ref{tab:geithner_att} reports the average treatment effects over the 11-day post-event window. Column 1 presents the difference in average returns between treated and control firms, implementing the counterfactual as a simple average of returns from non-connected firms—the same approach used in Table 2 of \textcite{acemoglu2016value}. Column 2 reports difference-in-differences estimates. Columns 3--5 present traditional factor model adjustments: the market model (Column 3), CAPM (Column 4), and Fama-French three-factor model (Column 5). Columns 6--8 employ synthetic control methods: standard synthetic control \parencite{abadie2010synthetic}, synthetic difference-in-differences \parencite{arkhangelsky2021synthetic}, and generalized synthetic control (Gsynth) from \textcite{xu2017generalized}. For all models requiring pre-event estimation, we use days $-256$ to $-31$, slightly shorter than the $-280$ to $-31$ window in the original paper to maintain a balanced panel. We report a graphical version of \Cref{tab:geithner_att} in \Cref{fig:geithner_att}.

\begin{figure}[thb]
    \caption{\small \textbf{Connections to Geithner and Returns after Treasury Secretary News.} 
    This figure plots the average treatment effects on the treated from \Cref{tab:geithner_att} after the announcement of Timothy Geithner as Treasury Secretary. Event day 0 is November 21, 2008 from 3pm (when the news leaked) to market closing, consistent with \cite{acemoglu2016value}. The average treatment effect is estimated using returns from trading day 0 to day 10. We consider two control samples: banks or financial services firms trading on the NYSE or Nasdaq (Panel A), and all NYSE, AMEX, and NASDAQ common stocks (Panel B). We consider several estimators: difference in average, difference-in-differences, synthetic control, synthetic DinD, and generalized synthetic methods. Standard errors of difference in average is from a two-sample t-test. Standard errors of DID, synthetic control, and synthetic DID are calculated using placebo inference following \cite{arkhangelsky2021synthetic} with 100 repetitions. Standard errors of Gsynth is computed using parametric bootstrap with 1,000 samples. }\label{fig:geithner_att}
    \begin{subfigure}[t]{0.46\textwidth}
        \centering
        \caption{Panel A: Bank Controls}
        \includegraphics[width=\textwidth]{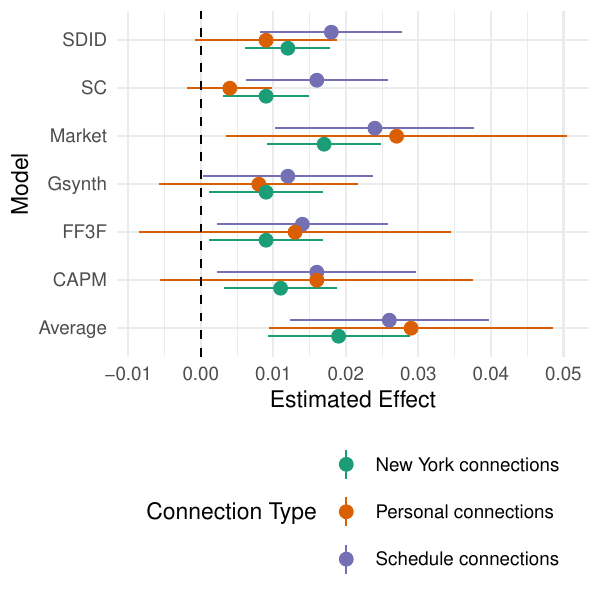}
    \end{subfigure}
        \vspace{1cm}
    \begin{subfigure}[t]{0.46\textwidth}
        \centering
        \caption{Panel B: Full CRSP Controls}
        \includegraphics[width=\textwidth]{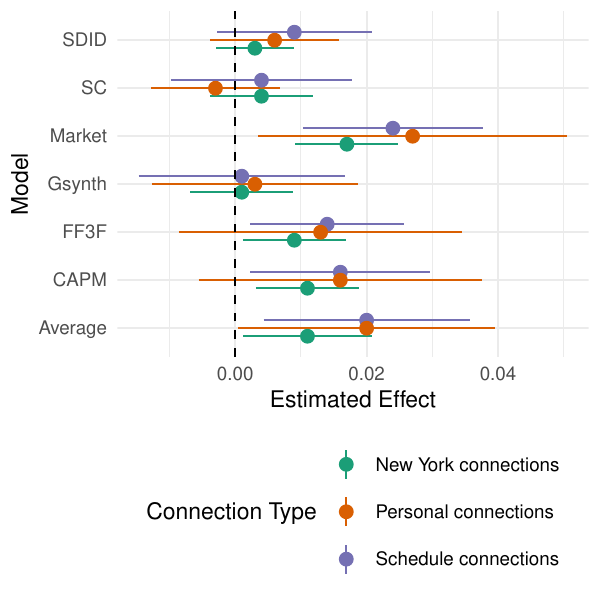}
    \end{subfigure}
\end{figure}

The results reveal a clear pattern. Simple averaging and difference-in-differences (Columns 1--2) show that firms with schedule connections experience 2.6--2.7\% higher cumulative returns, those with personal connections show 2.9--3.0\% higher returns, and firms with New York connections exhibit 1.9--2.0\% higher returns. The market model adjustment (Column 3) produces minimal changes. However, risk-adjusted returns using CAPM and Fama-French models (Columns 4--5) reduce these estimates by approximately 40--50\%, suggesting that connected firms have higher market betas. The synthetic control methods (Columns 6--8) produce even larger reductions, with standard synthetic control reducing schedule connection effects by 38\% and personal connection effects becoming statistically insignificant.\footnote{Our results contrast with \textcite{acemoglu2016value}, who employ synthetic control methods as robustness checks. Their approach was necessarily \textit{ad hoc} given the limited literature at the time on handling multiple treated units in synthetic control settings.}

\paragraph{All public firms as controls.}
We next expand the control group to include all common shares traded on NYSE, AMEX, and NASDAQ, with results reported in Panel B of Table~\ref{tab:geithner_att}. This expansion is motivated by the integration of equity markets: systematic factors should be well-identified using the universe of traded stocks. Restricting controls to financial firms alone may be suboptimal unless banking-specific factors exist that cannot be spanned by the broader market.

The expanded control group dramatically changes the results from synthetic control methods while leaving traditional methods largely unaffected. Simple averaging and difference-in-differences (Columns 1--2) continue to show significant effects of approximately 2\% for all connection types. The factor model adjustments (Columns 3--5) show a similar pattern to Panel A, with the market model producing minimal changes while CAPM and Fama-French adjustments reduce estimates by 30--50\%. 

Strikingly, the synthetic control methods now produce near-zero and statistically insignificant estimates. Standard synthetic control (Column 6) yields point estimates of 0.4\% for schedule connections and $-0.3\%$ for personal connections. Gsynth (Column 8) estimates are particularly close to zero: 0.1\% for schedule connections, 0.3\% for personal connections, and 0.1\% for New York connections—all statistically insignificant. This dramatic difference suggests that the broader control group allows synthetic methods to better match the factor exposures of treated firms, effectively eliminating the estimated treatment effects. The contrast between traditional factor adjustments (which still show significant effects) and synthetic methods (which do not) highlights the importance of allowing flexible, data-driven matching of factor exposures rather than imposing a specific factor structure.

\subsubsection{Market Returns around Event}\label{sec:mktret_geithner}

We now investigate the sources of bias in the original estimates. First, we examine the distribution of market returns during the event window. Figure~\ref{fig:sp500_geithner} displays the kernel density of daily S\&P 500 returns from 1962--2023, overlaid with the realized returns during the 11-day event window. The event period coincides with extraordinary market volatility, with returns falling in the extreme tails of the historical distribution. The market surged 6.6\% on November 21 (day 0) and 6.5\% on November 24 (day 1), while the largest decline of $-8.4\%$ occurred on December 1 (day 5).

These extreme factor realizations have important implications for identification. As demonstrated in Proposition~\ref{prop:consistency}, when treatment occurs simultaneously for all units, abnormal return estimators are particularly sensitive to factor model misspecification. The bias is proportional to both the magnitude of factor realizations and the difference in factor loadings between treated and control firms: $(\beta_s - \hat{\beta}_s)F_t$. Large factor realizations during the event window amplify any misspecification bias arising from imperfect matching of factor exposures.

This mechanism explains the substantial reduction in estimated treatment effects when using synthetic control methods rather than simple averaging. Proposition~\ref{prop:consistency} shows that both  synthetic control and gsynth estimators are asymptotically unbiased even with omitted factors, as they construct control portfolios that match the pre-event factor structure of treated firms without requiring explicit factor model specification. The extreme market conditions during the Geithner announcement thus reveal the importance of proper counterfactual construction in volatile periods.

\begin{figure}[thb]
    \caption{\small \textbf{S\&P 500 Returns around Treasury Secretary Announcement}
    This figure plots the daily returns of S\&P 500 index around the announcement of Timothy Geithner as Treasury Secretary. Event day 0 is November 21, 2008 from 3pm (when the news leaked) to market closing, consistent with \cite{acemoglu2016value}. The blue solid line plots the kernel density function of daily S\&P 500 returns from 1962 to 2023, and the sienna dashed vertical lines are the realization of daily returns in the post periods from trading day 0 to day 10. We label the dates with the largest outliers. The most positive realization is on event days November 21 and 24.}\label{fig:sp500_geithner}
    \includegraphics[width=\linewidth]{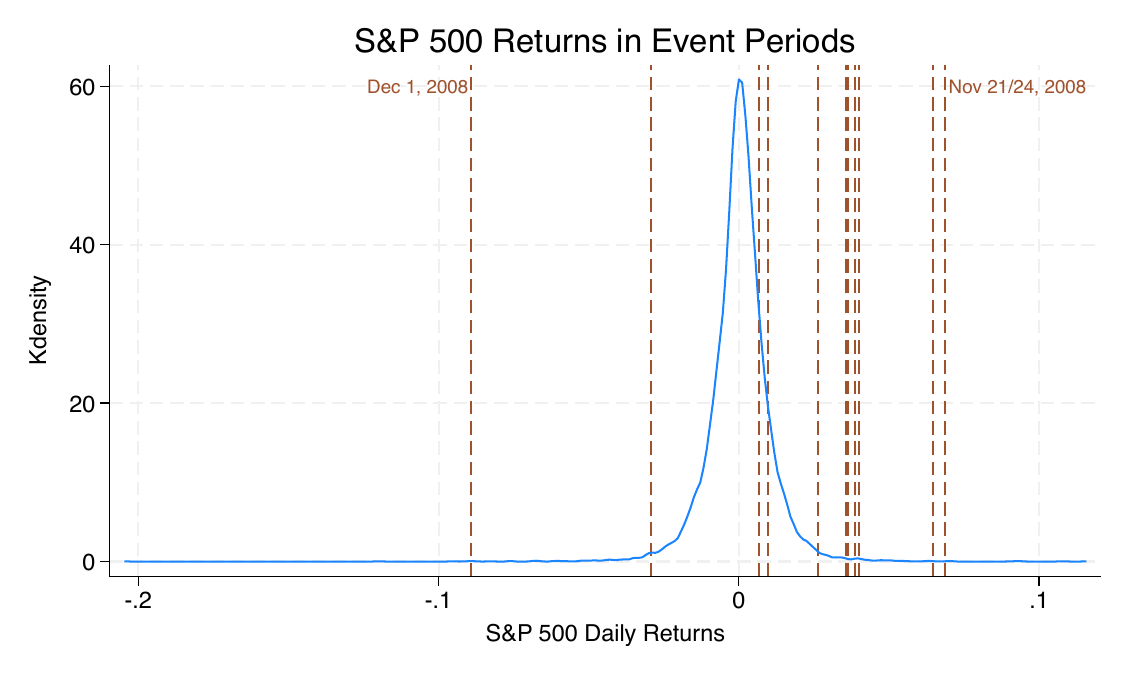}
\end{figure}

\subsubsection{Factor loadings of treated units match synthetic control factor loadings}

We now provide direct evidence on the factor exposure differences between treated and control firms. We estimate market betas using daily returns from day $-280$ to day $-31$ before the event, running firm-level time-series regressions on the S\&P 500 index return for CAPM betas and on the Fama-French three factors for multifactor betas.\footnote{Fama-French factor returns are obtained from Kenneth French's data library: \url{https://mba.tuck.dartmouth.edu/pages/faculty/ken.french/data_library.html}.}

Table~\ref{tab:beta_geithner} reports the weighted average betas for treated and control portfolios. Panel A presents equal-weighted averages for treated firms and two control groups: financial institutions only and all public firms. The results reveal substantial factor exposure mismatches. Treated firms have an average CAPM beta of 1.43, compared to 0.83 for financial controls—a difference of 0.60. The Fama-French three-factor model confirms this pattern: treated firms exhibit a market beta of 1.28 versus 0.66 for controls, with similar disparities in SMB (0.23 vs.\ 0.75) and HML (0.61 vs.\ 0.72) exposures.

These factor loading differences, combined with the extreme market realizations documented in Section~\ref{sec:mktret_geithner}, generate substantial bias in simple difference estimators. During the event window, the 0.60 difference in market beta translates to a bias of approximately $0.60 \times 6.9\% = 4.1\%$ on November 21 alone. This mechanical bias explains much of the estimated effect found using naive averaging methods.

We report the (weighted) average of betas of treated and control firms in Table \ref{tab:beta_geithner}. First, in Panel A, we first show the average CAPM and Fama-French three-factor betas of the treated firms and equal-weighted averages of financial firm controls and all public firm controls. The average CAPM beta of the treated firms is 1.43, much higher than 0.83 from the control firms. Expanding to a three-factor model, we still see a higher market beta in treated firms. Given these mismatches of treated and control betas, together with turmoil market returns, as shown in Section \ref{sec:mktret_geithner}, could lead to large biases in average treatment effects by comparing treated versus control firms.

\begin{table}[thb]
    \centering
    \caption{\small\textbf{Treated and Control Betas in Geithner as Treasury Secretary}
    This table presents the average CAPM and Fama-French three-factor betas for the treated and control firms. We first estimate firm-level betas using daily stock returns from 280 to 30 days before the announcement of Timothy Geithner as Treasury secretary on Nov 21, 2008. We then average the betas within the treated firms and two control samples: banks or financial services firms trading on the NYSE or Nasdaq, and all NYSE, AMEX, and NASDAQ common stocks. In Panel A, we show the simple average of treated firms and two control firms, and in Panel B, we calculate weighted average beta using weights from various synthetic methods: synthetic control and synthetic DinD.}\label{tab:beta_geithner}
    \begin{tabular}{lccc}
    \multicolumn{4}{l}{Panel A: Simple Averages}\\
\toprule
                 & Treated & Control & Control (All CRSP) \\
                 \midrule
CAPM Beta        & 1.427   & 0.825   & 0.832              \\
FF3F Market Beta & 1.275   & 0.659   & 0.857              \\
FF3F Size Beta   & 0.233   & 0.748   & 0.553              \\
FF3F Value Beta  & 0.607   & 0.720   & 0.144             \\
\bottomrule
\end{tabular}

\begin{tabular}{lcccc}
\multicolumn{4}{l}{Panel B: Weighted Averages with Synthetic Methods}\\
\toprule
                 & \multicolumn{2}{c}{Bank Controls} & \multicolumn{2}{c}{All CRSP Controls} \\
                 & SC              & SDID            & SC                & SDID              \\
                 \midrule
CAPM Beta        & 1.331           & 1.111           & 1.383             & 1.281             \\
FF3F Market Beta & 1.148           & 0.905           & 1.220             & 1.165             \\
FF3F Size Beta   & 0.480           & 0.819           & 0.377             & 0.627             \\
FF3F Value Beta  & 0.750           & 0.872           & 0.674             & 0.593            
\\
\bottomrule
\end{tabular}
\end{table}

In Panel B, we compute the weighted average betas of control firms using synthetic control weights, with both standard synthetic control and synthetic difference in differences. First, we see that synthetic methods match the beta in the treated firms  well. For example, the synthetic control gives a weighted average beta of 1.33, much closer to the treated beta of 1.43 than the equal-weighted average. Fama-French three-factor betas of the treated firms are 1.28 on the market, 0.23 on SMB, and 0.61 on HML, and synthetic control weights give a market beta of 1.15, SMB beta of 0.48, 0.75 ( closer than 0.66, 0.75, and 0.72 with simple average). Second, if we extend the set of possible control firms from financial firms in \textcite{acemoglu2016value} to all public firms in CRSP, we obtain better matches across all synthetic methods. For synthetic control specifically, controlled firms give an average beta of 1.38, closer to 1.43 in the treated firm. There is also a significant improvement in matching the Fama-French three-factor betas, synthetic control betas are 1.22, 0.38, and 0.67 (compared to treated betas of 1.28, 0.23, and 0.61). Finally, standard synthetic control methods give slightly better weights than synthetic difference-in-differences, who is more directly related to a mimicking portfolio approach.

Overall, synthetic methods matches the beta of treated firms well, which results in a lower bias in the average treatment effects.


\subsection{Empirical Example 2: Index Inclusion}
\label{sec:index_inclusion}

We next examine S\&P 500 index inclusion announcements, analyzing both immediate announcement returns and pre-announcement price dynamics to test our theoretical predictions regarding identification in staggered event settings.

We first demonstrate that in staggered event settings, announcement-day bias is negligible because factor returns on event days average close to zero, particularly when compared to the large treatment effects of 3--4\%. However, consistent with \textcite{Greenwood2025Index}, we document substantial pre-announcement drift. Synthetic control methods that match on pre-event returns nearly eliminate this drift. This pattern is consistent with selection on unobserved factors: firms added to the index differ systematically from control firms along dimensions not captured by observable factors, generating apparent pre-event "drift" that actually reflects factor model misspecification.

\paragraph{Empirical setting.}
Following \textcite{Greenwood2025Index}, we obtain index inclusion dates from Siblis Research and match tickers to CRSP PERMNOs using header information. Siblis provides announcement dates for S\&P 500 additions. For the period September 1976 through September 1989, when announcement dates are missing, we exploit the institutional detail that index changes were announced after Wednesday market close and became effective the following day, allowing us to infer announcement dates.\footnote{During this period, S\&P followed a predictable schedule of announcing changes after Wednesday close for Thursday implementation.} We measure returns on the announcement date when it falls on a trading day; otherwise, we use the most recent prior trading day.

To assess whether event timing can be treated as random, we examine the distribution of factor returns on announcement days.  Appendix \Cref{fig:factor_balance_index}, Panel A shows that the distribution of daily market returns on S\&P 500 index inclusion announcement days is virtually indistinguishable from the distribution on non-announcement days. This pattern holds consistently across our entire sample period, from 1980--1989 through 2010--2020. The small-minus-big (SMB) factor exhibits similar distributional stability (Panel B). These results support treating announcement timing as conditionally random with respect to factor realizations, satisfying a key identification assumption for our short-horizon analysis.

Table~\ref{tab:beta_index} reports CAPM and Fama-French three-factor betas for firms added to the S\&P 500 index, estimated using daily returns from days $-250$ to $-100$ relative to announcement.\footnote{We exclude the immediate pre-announcement period to avoid contamination from potential information leakage.} We present results separately by decade from 1980 through 2020 to examine temporal variation in the characteristics of included firms.

\begin{table}[tbhp]
    \centering
    \caption{\small \textbf{Beta Distributions of Included Firms across Decades} 
    This table presents the average CAPM and Fama-French three-factor betas for firms included in S\&P 500, compared with a random set of control firms of the same sample size. For each treated firm and inclusion date, we randomly pick a non-treat firm in CRSP sample with common share in NYSE, NASDAQ, or AMEX, which at least 250 trading days of returns before the announcement date. We then estimate firm-level betas using daily stock returns from 250 to 100 days before the announcement of inclusions into S\&P 500 index. We provide the summary statistics for the distribution of betas of included firms, separately for each decade.} \label{tab:beta_index}     
\small 
\begin{tabular}{lcccc}
\toprule
               & \multicolumn{2}{c}{Treated} & \multicolumn{2}{c}{Random Control} \\
               & Mean          & Std         & Mean             & Std             \\
   \midrule
\multicolumn{5}{l}{\textbf{Panel A: 1980-1989}} \\
CAPM Beta      & 0.961         & 0.523       & 0.582            & 0.551           \\
FF3F Mkt Beta  & 1.108         & 0.539       & 0.854            & 0.784           \\
FF3F SMB Beta  & 0.558         & 0.604       & 0.815            & 1.044           \\
FF3F HML Beta  & -0.148        & 0.987       & 0.021            & 1.188           \\
\multicolumn{5}{l}{\textbf{Panel B: 1990-1999}} \\
CAPM Beta      & 1.025         & 0.660       & 0.651            & 0.754           \\
FF3F Mkt Beta  & 1.171         & 0.660       & 0.873            & 0.911           \\
FF3F SMB Beta  & 0.489         & 0.661       & 0.805            & 1.215           \\
FF3F HML Beta  & -0.015        & 1.242       & 0.022            & 1.475           \\
\multicolumn{5}{l}{\textbf{Panel B: 2000-2009}} \\
CAPM Beta      & 1.087         & 0.697       & 0.824            & 0.985           \\
FF3F Mkt Beta  & 1.079         & 0.560       & 0.820            & 0.688           \\
FF3F SMB Beta  & 0.271         & 0.674       & 0.667            & 0.929           \\
FF3F HML Beta  & -0.002        & 1.227       & 0.075            & 1.482           \\
\multicolumn{5}{l}{\textbf{Panel D: 2010-2020}} \\
CAPM Beta      & 1.060         & 0.388       & 0.973            & 0.997           \\
FF3F Mkt Beta  & 1.026         & 0.343       & 0.872            & 0.614           \\
FF3F SMB Beta  & 0.225         & 0.520       & 0.628            & 1.201           \\
FF3F HML Beta  & -0.273        & 0.590       & 0.311            & 1.272          \\
\bottomrule
\end{tabular}
\end{table}

Across all decades, the average market beta of included firms is approximately one.  When treated firms have market betas near unity, the simple market-adjusted return (which implicitly assumes $\beta = 1$) yields similar results to the more sophisticated CAPM adjustment that estimates firm-specific betas. This convergence occurs because the bias term $(1 - \beta_i) \times r_{m,t}$ approaches zero when $\beta_i \approx 1$, consistent with the theoretical predictions in \Cref{thm:bias}.

The combination of two empirical regularities—random event timing with respect to factor realizations and limited selection on factor loadings—suggests that short-horizon abnormal return estimates should exhibit minimal bias regardless of the specific factor model employed. This prediction from Theorem~\ref{thm:bias} finds strong empirical support in Table~\ref{tab:index_include_att}, where announcement-day treatment effects are remarkably stable across estimation methods. The difference between simple market adjustment and sophisticated synthetic control methods is less than 0.2 percentage points in most decades, confirming that model specification has negligible impact on short-horizon estimates when the conditions of Theorem~\ref{thm:bias} are satisfied.

\begin{table}[thb]
    \centering
    \caption{\small \textbf{Announcement-Day Treatment Effects of Index Inclusion}
    This table presents average treatment effects on the announcement days of index inclusion, averaged across inclusions for each decade. We consider several estimators: difference in simple average, CAPM, Fama-French 3-factor, and gsynth. The estimation window of factor loadings are from -250 to -101 before the announcement dates. }\label{tab:index_include_att}

\begin{tabular}{lccccc}
\toprule
          & Diff-in-Means & Market & CAPM   & FF3F   & Gsynth \\
          \midrule
1980-1989 & 3.27\%        & 3.25\% & 3.15\% & 3.05\% & 3.06\% \\
1990-1999 & 4.61\%        & 4.62\% & 4.69\% & 4.71\% & 4.79\% \\
2000-2009 & 3.42\%        & 3.43\% & 3.33\% & 3.22\% & 3.41\% \\
2010-2020 & 1.14\%        & 0.94\% & 0.85\% & 0.85\% & 0.93\% \\
\bottomrule
\end{tabular}
\end{table}

\subsubsection{Pre-inclusion Drift}

While \Cref{thm:bias} predicts negligible bias in short-horizon studies, it also implies that long-horizon estimates may suffer from substantial bias unless factor exposures are correctly specified. We now examine the "pre-announcement drift" documented by \textcite{Greenwood2025Index}, analyzing it decade by decade as a manifestation of potential long-horizon bias.

Interpreting pre-announcement price movements requires careful consideration of \Cref{assn:limited_anticipation}, our limited anticipation assumption. This assumption is particularly tenuous in the index inclusion setting for two reasons. First, market participants have incentives to anticipate market index changes. Second, as \textcite{Greenwood2025Index} document, inclusion is partially predictable: firms with market capitalizations just below the S\&P 500 cutoff face substantially higher inclusion probabilities than other firms. This predictability complicates the identification of treatment effects, as observed pre-announcement returns may reflect either genuine anticipation (violating \Cref{assn:limited_anticipation}) or selection on unobserved characteristics that drive both inclusion probability and returns.

To disentangle these effects, we pursue a two-pronged empirical strategy. First, we implement propensity score matching based on observable firm characteristics to account for selection on observables. Second, we employ synthetic control methods that match on pre-event returns, effectively controlling for unobserved factors that drive both selection and returns. The difference between these two approaches helps identify whether pre-announcement drift reflects anticipation or factor model misspecification.

Index inclusion predictability operates along two dimensions: the timing of additions (when inclusions occur) and the cross-section of selections (which firms are added). While ideally we would model both, we focus on cross-sectional predictability by estimating inclusion propensities based on observable firm characteristics. Specifically, we estimate annual logistic regressions:
\begin{equation}
    \mathbf{1}(\text{Added})_{i,y,m} = \alpha_{y} + \beta_{y} \cdot \text{MktCapRank}_{i,y,m-1} + \varepsilon_{i,y}
\end{equation}
where $\text{MktCapRank}_{i,y,m-1}$ is firm $i$'s market capitalization rank at the end of month $m-1$, and inclusion occurs in month $m$ of year $y$. Consistent with \textcite{Greenwood2025Index}, we find increasing predictability over time, with recent decades showing stronger relationships between lagged size and inclusion probability.

Using these propensity scores, we construct matched control groups via nearest-neighbor matching, creating portfolios of "pseudo-included" firms with similar inclusion probabilities but no actual inclusion. Under the assumption that selection between observationally equivalent firms is quasi-random, differences between included and pseudo-included firms should primarily reflect the causal effect of inclusion rather than selection bias.

To address selection on unobservables, we additionally implement the generalized synthetic control method of \textcite{xu2022gsynth}. For each announcement date, we estimate factor loadings using returns from days $-250$ to $-101$, deliberately excluding the immediate pre-announcement period where anticipation effects may contaminate estimation. We then construct synthetic control portfolios that match the pre-event return dynamics of included firms, examining the period from day $-100$ to $-15$.

This dual approach yields three distinct counterfactuals for cumulative abnormal returns (CARs): (i) simple market adjustment as in \textcite{Greenwood2025Index} (we also do CAPM and FF3F adjustments for completeness, but do not subtract $\alpha$ for reasons that will be clear shortly) (ii) propensity score-matched pseudo-included firms that control for selection on observables, and (iii) synthetic controls that account for selection on unobserved factors. Comparing these counterfactuals from day $-100$ through the announcement date allows us to decompose pre-announcement drift into components attributable to observable characteristics versus unobserved factor exposures. If drift persists after propensity score matching but disappears with synthetic controls, this would suggest that unobserved factors—rather than anticipation based on observables—drive the pre-announcement returns.

\begin{figure}[thb]
    \caption{\small\textbf{Cumulative abnormal pre-addition returns}
    This figure plots the average cumulative abnormal returns following index inclusion announcements in event time, averaged across inclusions for each decade. We use several definitions of abnormal returns with different counterfactual returns. Solid lines plot abnormal returns with S\&P 500 market returns, dashed lines plot abnormal returns with a propensity-score-matched counterfactual firm on lagged market cap rank, and dotted lines plot abnormal returns with synthetic portfolios from the generalized synthetic method \parencite{xu2022gsynth}. The returns are normalized to start at zero, 100-trading days before the announcement.}\label{fig:ii_psm_gsynth}
    \begin{center}  \includegraphics[width=\textwidth]{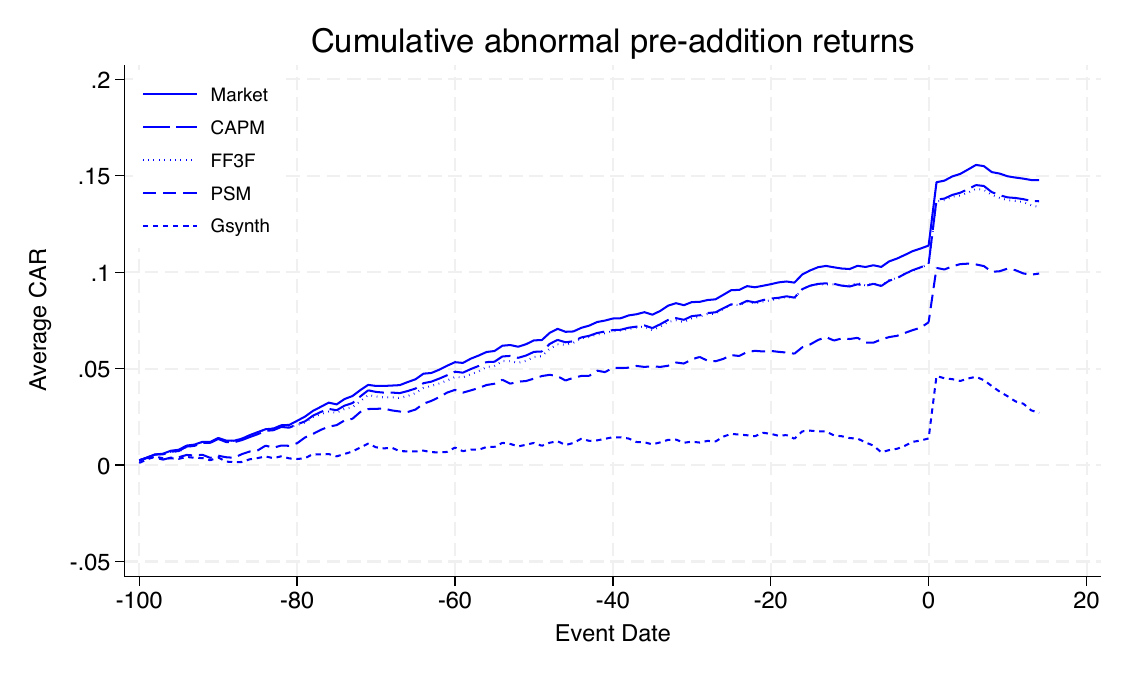}\end{center}
\end{figure}

First, we find the pre-announcement drift as estimated by either the propensity score matched difference, or by the Gsynth approach drops significantly when compared to the market adjusted method. The effectiveness of Gsynth is quite striking in this setting, and suggests that longer-run cumulative effects can be substantially biased. What can explain the differences identified between these estimated methods? In Appendix \Cref{fig:preindexcumulret}, we show that there is a substantial drift in our known factors across most decades. Considering the positive loadings in \Cref{tab:beta_index}, this suggests that the counterfactual return needs to sufficiently account for any and all potential unobserved factors driving the expected returns to avoid this bias highlighted in \Cref{thm:bias}.

\begin{figure}[H]
    \centering
    \caption{Per-Period and Cumulative ATT with factor models, gsynth, and synthetic methods}
    \label{fig:pretrend}
    \begin{tabular}{cc}
    Panel A: Per-Period ATT & Panel B: Cumulative ATT\\
    \includegraphics[width=0.49\linewidth]{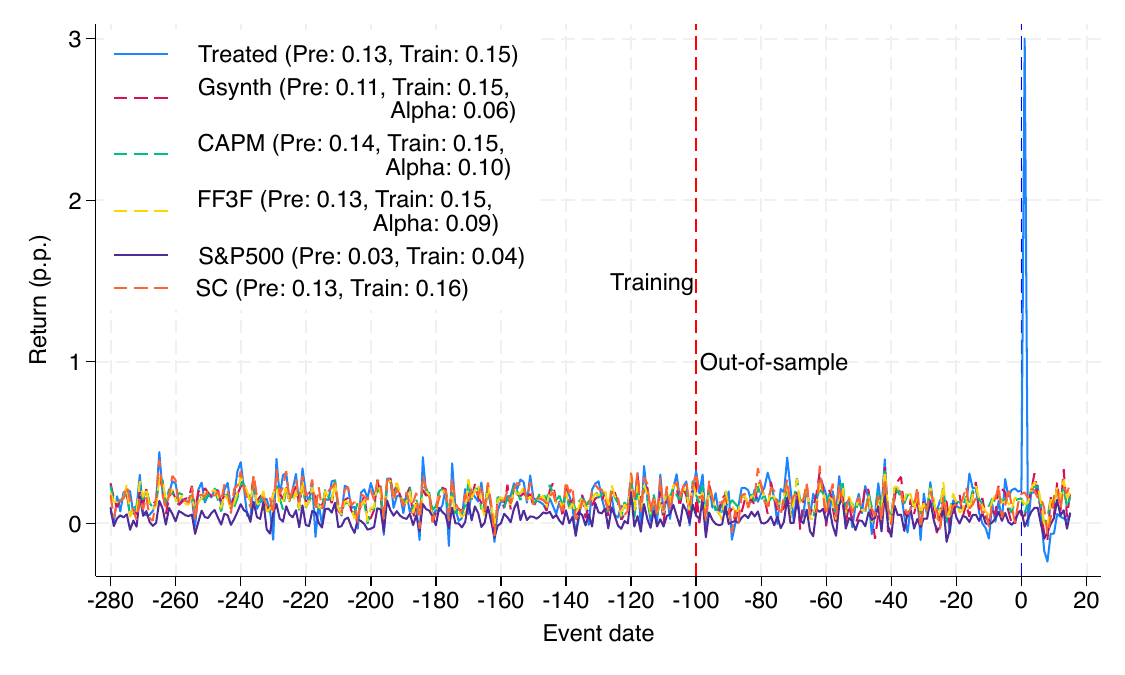} &   \includegraphics[width=0.49\linewidth]{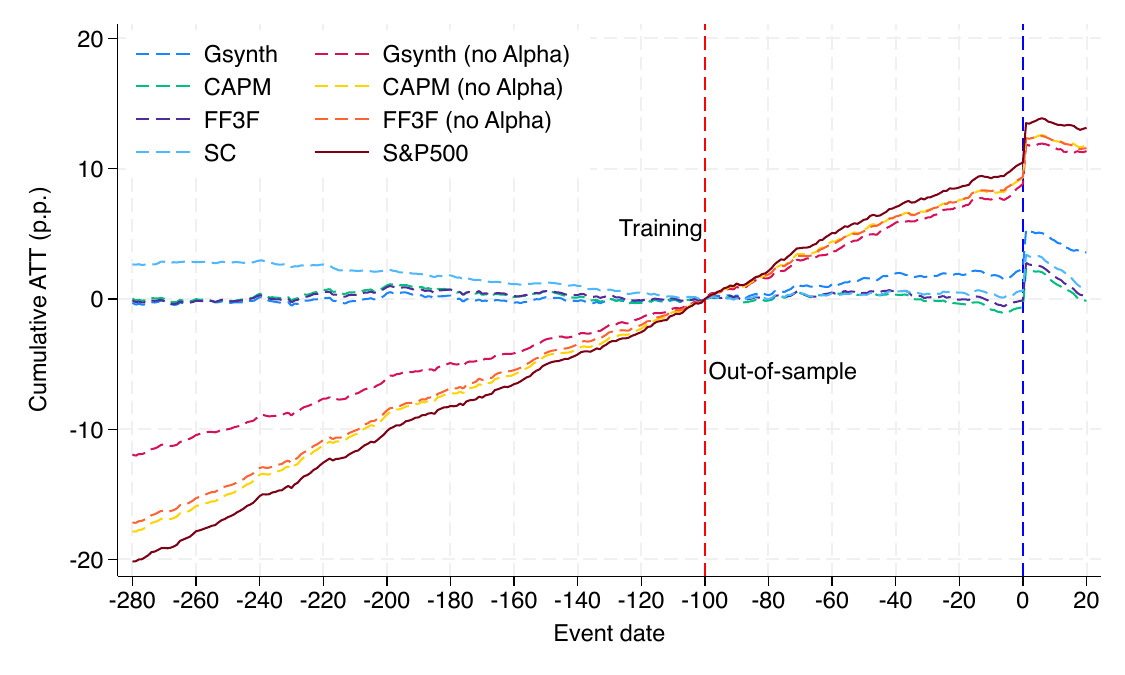}
    \end{tabular}
\end{figure}

Capturing all potential unobserved factors is not easy, however. In \Cref{fig:pretrend}, Panel A, we plot the average return in event time for the treated group, and then our five counterfactuals. We report the average daily return for each group, and note a 0.13 p.p. daily return for the included firms (\emph{prior} to inclusion), an unusually high daily return. In contrast, the S\&P500 has a daily return of only 0.03 p.p. during this period. This suggests that the included firms are quite unusual. Recall that if our factor choice in a linear model sufficiently spans the risk factors, we should estimate an alpha of zero, even in the presence of positive average returns.  Our counterfactual models do an excellent job of matching the average return in the training and pre-periods. However, for the CAPM and FF3F, more than half of the average predicted return comes from just the intercept, $\alpha$. For gsynth, the estimated $\alpha$ is less than half, but still 0.06 p.p. Strikingly, the synthetic control counterfactual, which only takes a positive weighted average of control firms and does not include a constant, matches the pre-period return closely.  

How should we interpret the estimated alpha in this linear factor models? There is presumably two components in an estimated factor model's alpha, true alpha, and model error:
\begin{equation}
    \hat{\alpha} = \underbrace{\alpha}_{\text{true }\alpha} + \underbrace{\beta^{unobs} E(F_{t}^{unobs} | t \in \text {estimation window})}_{\text{Misspecifaction}},
\end{equation}
where the model misspecification captures the return premium over this period that is not included in our model. In this setting, we view true alpha as zero, especially 280 days prior to the inclusion event. As a result, the positive alpha likely suggests model misspecification. The implications of this misspecification depend on the stability of this misspecification term. In Panel B of \cref{fig:pretrend}, we see that the inclusion of alpha ensures that the various linear factor models do as well the synthetic control method in removing almost all pre-inclusion drift. This suggests that the trend beforehand is not due to front-running, but instead differential return profiles for included stocks.  However, failing to include alpha for the CAPM, FF3F  and gsynth fail to remove the pre-inclusion drift.


\subsection{Empirical Example 3: Mergers and Acquisitions}

\subsubsection{Empirical Setting}

We examine acquirer returns around merger announcements using deal data from SDC Platinum. Following \textcite{malmendier2018behavioral} and \textcite{savor2009stock}, we implement several sample restrictions to ensure clean identification. We require targets to be classified as ``Public,'' ``Private,'' or ``Subsidiary'' and restrict to completed deals with all-cash or all-stock payment structures, as mixed consideration complicates the interpretation of market-timing effects. To ensure economic materiality, we require the target's pre-announcement market value to exceed 5\% of the acquirer's market capitalization. We exclude repurchases, self-tenders, and minority stake purchases by requiring deal types to be ``Disclosed Dollar Value'' or ``Undisclosed Dollar Value,'' and mandating that acquirers hold less than 50\% of the target six months before announcement.

We match acquirers to CRSP using six-digit CUSIPs, restricting to U.S. common shares (share codes 10 or 11) traded on NYSE, NASDAQ, or AMEX. Our event window spans days $-280$ to $+250$ relative to announcement. For the 20\% of deals announced on non-trading days, we define $t = 0$ as the next trading day. Control firms comprise all CRSP-listed firms without contemporaneous merger announcements that have complete returns data over the event window. Our final sample contains 14,847 merger events across 6,625 unique dates, providing substantial variation in event timing for identification.

\subsubsection{Short-Term Announcement Returns}

To assess whether event timing can be treated as random, we examine the distribution of factor returns on announcement days.  Appendix \Cref{fig:factor_balance_index} shows that the distribution of daily market returns on S\&P 500 index inclusion announcement days is virtually indistinguishable from the distribution on non-announcement days. The small-minus-big (SMB) factor exhibits similar distributional stability (unreported) These results support treating announcement timing as conditionally random with respect to factor realizations, satisfying a key identification assumption for our short-horizon analysis.

Table~\ref{tab:beta_ma_sl_m} reports CAPM and Fama-French three-factor betas for firms with merger announcements, estimated using daily returns from days $-250$ to $-100$ relative to announcement.\footnote{We exclude the immediate pre-announcement period to avoid contamination from potential information leakage.} We also examine how the betas change following the announcement as well, and show that there are statistically significant changes after announcement, but they are small economically.

We first examine three-day announcement returns $[-1, +1]$ to test whether short-horizon estimates are robust to model specification, as predicted by \Cref{thm:bias}. We compare two approaches: market-adjusted returns using the CRSP value-weighted index (including distributions) following \textcite{malmendier2018behavioral}, and gsynth estimates using the generalized synthetic control method of \textcite{xu2017generalized}.

For the synthetic control approach, we estimate a separate model for each of the 6,625 event dates, treating all firms announcing mergers on that date (typically one or two firms) as the treatment group. We construct factor loadings using returns from days $-280$ to $-31$, excluding the immediate pre-announcement period to avoid contamination. Control firms consist of all CRSP securities without merger announcements that satisfy our data requirements.

\Cref{tab:merger_car} reports cumulative abnormal returns by target type and payment method. Consistent with our theoretical predictions, the difference between market-adjusted and synthetic control estimates is economically negligible—less than 10 basis points in most specifications. This robustness to model choice confirms that short-horizon merger announcement effects are identified regardless of the specific factor adjustment employed.

\begin{table}[H]\centering
\caption{Average three day cumulative p.p.}
\label{tab:merger_car}
\footnotesize
\makebox[\textwidth][c]{
\begin{tabular}{lcccccc}
\toprule
 & Full sample & Public targets & Private targets & Other targets & Cash merger & Stock merger \\
\midrule
Market mean  & 0.8 & -1.2 & 1.3 & 1.6 & 1.1 & 0.4 \\
Gsynth mean  & 0.7 & -1.3 & 1.0 & 1.5 & 1.0 & 0.2 \\
\midrule
Count & 14,847 & 3,297 & 7,030 & 4,520 & 9,261 & 5,592 \\
\bottomrule
\end{tabular}}
\end{table}

In \Cref{fig:att_ma}, we plot the daily returns for the treated firms and the various control returns. Similar to the preannoucnment drift for the S\&P index inclusion, the treated firm has significant daily returns prior in the year prior to the announcement, with 0.14 p.p. daily average returns (contrast with the market having 0.05 p.p. daily average returns). Again, the CAPM and FF3F models have significant alpha (0.08 p.p.), while gsynth has remarkably small alpha (0.02 p.p.). However, both gsynth and synthetic control do a poorer job matching the overall pre-period return, with 0.12 p.p. and 0.19p.p. returns respectively, relative to 0.14 for the treated group. 

\begin{figure}[th]
    \centering
    \caption{Per-Period Treated and Counterfactual Returns by event date. Event date$=[-120,250]$}
    \label{fig:att_ma}
    \includegraphics[width=0.8\linewidth]{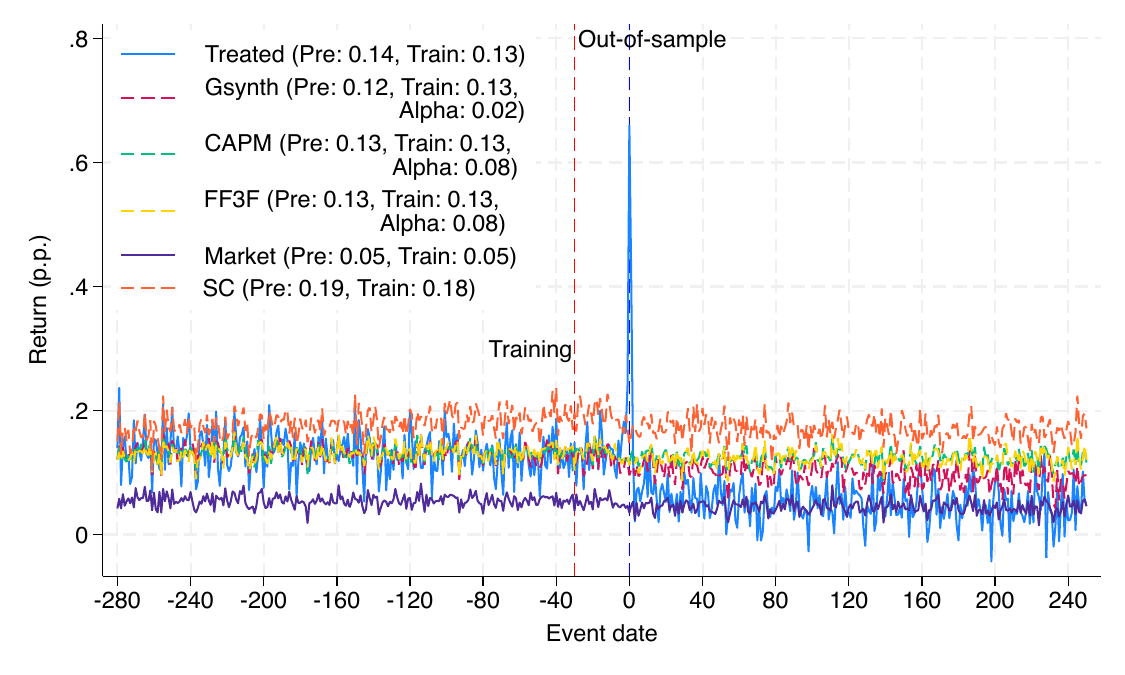}
\end{figure}

The path of the daily return line for the treated group spikes significantly on the event date, and then declines precipitously to a new steady state. It is worth remarking that within 30 days the event announcement, the treated firms' returns appear to line up almost exactly with the market returns. This is suggestive that there is a structural shift in the underlying return performance to these acquiring firms, perhaps due to change in true alpha, or perhaps due to factor loadings. 

\begin{figure}[th]
    \centering
    \caption{Cumulative ATT by event date. Event date$=[-120,250]$}
    \label{fig:cumulative_att_ma}
    \includegraphics[width=0.8\linewidth]{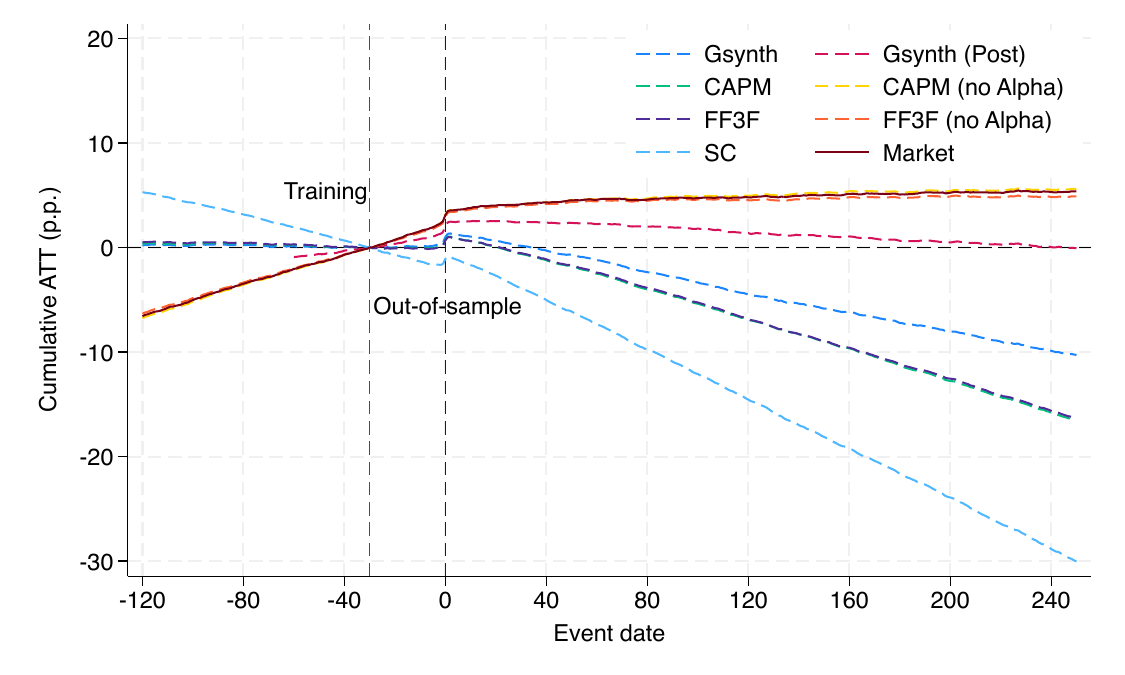}
\end{figure}

In \Cref{fig:cumulative_att_ma}, we see the long-run implications of these alphas in in that the difference counterfactual predictions have wildly different long-run cumulative ATT a year after the event. In the literature, the presence of a negative post-acquisition event is often pointed to as evidence in favor of \textcite{shleifer2003stock}, but the modeling assumptions to make these types of assessments seem quite strong. This type of analysis also has implications for papers studying over- and under-reaction in the stock market. 

What can we do to deal with this specification error? A crucial alternative is to examine the effect in a setting where treatment is as-if randomly assigned. 

\subsection{Empirical Example 4: Close Merger Contests as a Quasi-Experimental Benchmark}

Our previous empirical examples rely on model-based identification strategies that require correct specification of the factor structure. To validate these approaches, we now examine a setting with quasi-experimental variation: close merger contests where assignment to treatment (winning the contest) is plausibly random conditional on observables. \textcite{malmendier2018winning} show that in protracted bidding contests, winners and losers are ex ante similar firms competing for the same target, making losers natural counterfactuals for winners.

This setting provides a unique opportunity to assess the performance of different estimators. Since losing bidders offer a design-based counterfactual, we can compare our model-based estimates (market adjustment, factor models, synthetic controls) against this quasi-experimental benchmark. Agreement between model-based and design-based estimates would validate our econometric approaches; divergence would suggest specification problems in the model-based methods.

\subsubsection{Empirical Setting}

Following \textcite{malmendier2018winning}, we analyze close merger contests defined as those with above-median duration.\footnote{We thank the authors for providing data on winning and losing bidders, announcement and completion dates, and contest duration.} Protracted contests involving multiple rounds of bids and counterbids suggest that participants had similar ex ante winning probabilities, supporting the identifying assumption that contest outcomes are quasi-random.

We construct event-time at the monthly frequency, with $t = 0$ marking the month-end before the initial bid announcement. The pre-contest period spans months $t = -35$ to $t = 0$. The contest period ($t = 1$) encompasses all months from initial bid through completion, averaging 361 days in our sample. The post-merger period runs from $t = 2$ to $t = 36$. This structure accommodates contests of varying duration while maintaining a consistent event-time framework.

We match contest participants to CRSP monthly returns, filling missing observations with market returns following \textcite{malmendier2018winning}. For synthetic control estimation, we augment the sample with all CRSP common shares (share codes 10 or 11) traded on NYSE, NASDAQ, or AMEX that have complete returns over the event window. This expanded control group allows the synthetic control algorithm to construct appropriate counterfactuals even when losing bidders may themselves be poor matches due to contest-specific shocks affecting all participants.

To check how well the as-if random counterfactual losing bidders match to winners, In \Cref{tab:beta_close_contest}, we compare the risk exposures on different factors of winners and losers. We see that these two groups are relatively similar, although the losers have slightly higher HML beta than winners.

\begin{table}[thbp]
  \caption{\small \textbf{Beta Distributions of Winners and Losers in Close Merger Contests} 
    This table presents the average CAPM and Fama-French three-factor betas for winner and losers in close contest mergers. We estimate firm-level betas using daily stock returns using event month -35 to -1 before the start of the contest. We provide the mean and median of CAPM market beta and betas in Fama-French three-factor model. We test if the betas of treated and control firms are statistically different using a two-sided t-test with \textcite{welch1947generalization} approximation.} \label{tab:beta_close_contest}  
      \centering
    \begin{tabular}{lccccc}
\cmidrule{1-6}          & \multicolumn{2}{c}{Winner} & \multicolumn{2}{c}{Loser} & \multicolumn{1}{l}{Mean t-test} \\
          & \multicolumn{1}{l}{Mean} & \multicolumn{1}{l}{Median} & \multicolumn{1}{l}{Mean} & \multicolumn{1}{l}{Median} & \multicolumn{1}{l}{Loser - Winner} \\
    \midrule
    CAPM Beta & 1.148 & 0.922 & 1.114 & 0.968 & -0.034 \\
    FF3F Mkt Beta & 1.157 & 0.945 & 1.172 & 1.010 & 0.015 \\
    FF3F SMB Beta & 0.369 & 0.291 & 0.232 & 0.168 & -0.138 \\
    FF3F HML Beta & -0.055 & -0.171 & 0.313 & 0.393 & 0.368* \\
    \bottomrule
    \end{tabular}%
\end{table}%

In \Cref{fig:cum_att_ma_random}, we plot the cumulative ATTs for the winners relative to our various controls. Our benchmark is the ``Loser'' control, in solid blue. We see that in the pre-period, there is reasonable balance between the two groups, and then a small but significant decline following the announcement, suggesting a negative effect. Notably, this counterfactual has the smallest and least trending of the different counterfactuals. The only alternative model-based portfolio that is meaningfully close is the Gsynth control.  

\begin{figure}[H]
    \centering
    \caption{Cumulative ATT by event date. Event month$=[-35,36]$}
    \label{fig:cum_att_ma_random}
    \includegraphics[width=0.9\linewidth]{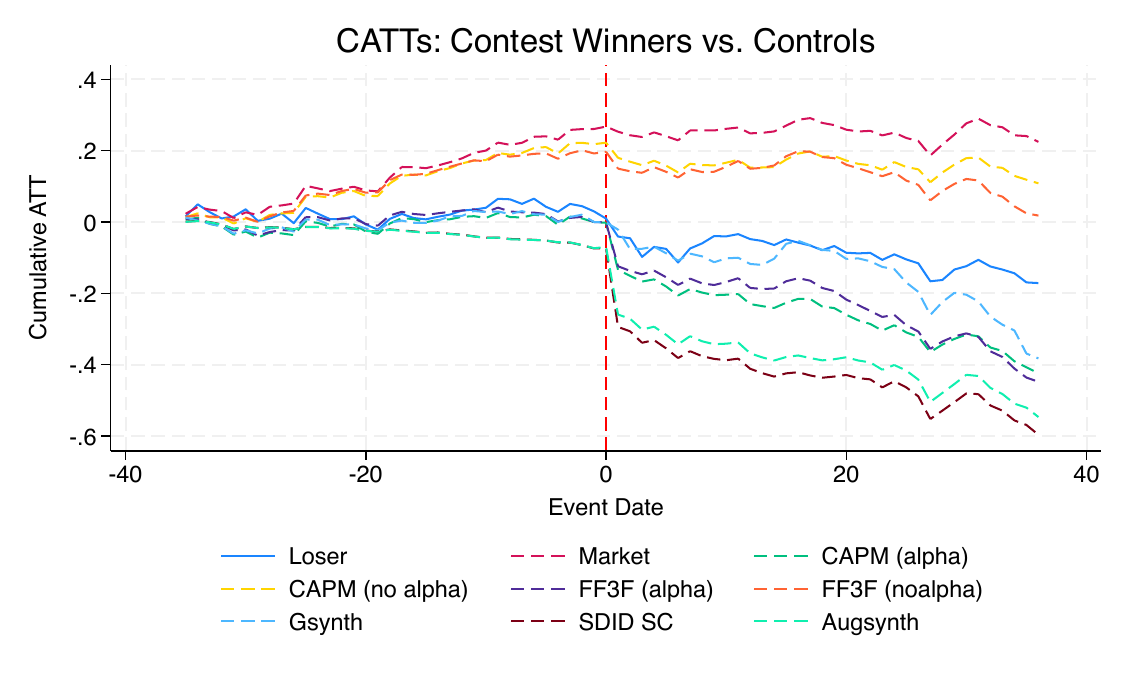}
\end{figure}

It is worth noting that the synthetic control methods have much larger declines during the merger battle, and afterwards as well. The abnormal return models fit well in the pre-period, but then predict significant and continuing declines in cumulative returns. Overall, this evidence suggests that most of the models (especially abnormal return) do poorly in the longer run, although gsynth is the exception.




\section{Conclusion}
This paper brings modern causal inference techniques to financial event studies, highlighting important limitations in standard approaches while providing constructive solutions. We demonstrate that traditional abnormal return estimators face inconsistency problems due to factor model misspecification -- a concern that becomes particularly severe in long-horizon analyses where small daily biases accumulate substantially over time.

While staggered event timing helps mitigate these issues in short-horizon studies by averaging out factor realizations, this solution proves inadequate for long-horizon analyses. The key insight is that misspecification bias compounds over longer horizons, regardless of how events are distributed across time.

Synthetic control methods offer a promising alternative by directly modeling counterfactual security paths without requiring correct specification of the underlying factor structure. Our empirical applications to political connections during market turbulence and S\&P 500 index inclusions convincingly demonstrate the practical value of these methods.

Our findings suggest that many influential results based on long-horizon event studies may reflect factor model misspecification rather than genuine causal effects. We recommend that researchers employ synthetic control methods as a robust complement to traditional approaches, particularly when studying extended price responses or when events occur during periods of high market volatility.

\newpage
\printbibliography

\clearpage
\clearpage


\clearpage
\appendix
\counterwithin{figure}{section}
\counterwithin{table}{section}
\section{Proofs}
\label{appsec:proofs}
\begin{proof}[Proof of Proposition~\ref{prop:consistency}]
Let $n_s = \#\{i : T_i = s\}$ and $n_c = \#\mathcal{C}$ denote the number of treated and control securities, respectively.  
Throughout the proof we maintain Assumptions~\ref{assn:factor_model} and~\ref{assn:limited_anticipation} and impose the following regularity conditions:

\begin{itemize}
    \item[(i)] For each $i,t$, the idiosyncratic component $\varepsilon_{it}$ satisfies
    \[
    \mathbb{E}(\varepsilon_{it} \mid \mathbf{F}_t, \mathbf{X}_i, T_i) = 0, 
    \quad 
    \mathbb{E}(\varepsilon_{it}^2) < \infty.
    \]
    \item[(ii)] Within each cohort $s$ and the control group $\mathcal{C}$, $\{\varepsilon_{it}\}$ are independent across $i$ (for fixed $t$) with uniformly bounded second moments, so that the law of large numbers applies to cross-sectional averages.
    \item[(iii)] The usual OLS regularity conditions hold for the pre-treatment regressions used to construct the abnormal-returns estimator (e.g.\ fixed $K$, non-singular regressor covariance matrix, etc.).
\end{itemize}

Define the cohort- and control-group average idiosyncratic shocks
\[
\varepsilon_{st} \equiv \frac{1}{n_s} \sum_{i:T_i = s} \varepsilon_{it},
\qquad
\varepsilon_{\infty t} \equiv \frac{1}{n_c} \sum_{i \in \mathcal{C}} \varepsilon_{it}.
\]
By (i)–(ii), for each fixed $t$,
\[
\varepsilon_{st} \xrightarrow{p} 0
\quad\text{and}\quad
\varepsilon_{\infty t} \xrightarrow{p} 0
\quad\text{as } n_s, n_c \to \infty.
\]

Recall the algebraic decompositions in equations~\eqref{eq:bias_ar}–\eqref{eq:bias_synth}:
\begin{align*}
\tau^{AR}(s,t) - \tau^{ATT}(s,t) 
    &= (\alpha_s - \hat{\alpha}_s) + (\beta_s \mathbf{F}_t - \hat{\beta}_s \mathbf{F}^o_t) + \varepsilon_{st}, \\
\hat{\tau}^{cont}(s,t) - \tau^{ATT}(s,t) 
    &= (\alpha_s - \alpha_\infty) + (\beta_s - \beta_\infty)\mathbf{F}_t + (\varepsilon_{st} -\varepsilon_{\infty,t}), \\
\hat{\tau}^{alt}(s,t) - \tau^{ATT}(s,t) 
    &= (\alpha_s - \hat{\alpha}^{alt}_s) + (\beta_s - \hat{\beta}^{alt}_s)\mathbf{F}_t + \varepsilon_{st},
\end{align*}
where $\hat{\alpha}_s, \hat{\beta}_s$ are the cohort averages of the OLS estimates from the abnormal-returns model and $(\hat{\alpha}^{alt}_s, \hat{\beta}^{alt}_s)$ denote the implied factor loadings from an alternative estimator (synthetic control or GSC) at the cohort level.

\paragraph{(1) Probability limits of the three estimators.}

\smallskip\noindent
\emph{Abnormal returns estimator.}
By definition, $\tilde{\alpha}_s$ and $\tilde{\beta}_s$ are the probability limits of the cohort-average OLS coefficients:
\[
\hat{\alpha}_s \xrightarrow{p} \tilde{\alpha}_s, 
\qquad 
\hat{\beta}_s \xrightarrow{p} \tilde{\beta}_s
\quad\text{as } T_{pre} \to \infty,
\]
where the limit is the linear projection of $R_{it}(\infty)$ onto $\mathbf{F}^o_t$ in the pre-treatment window $\{t < s - \delta\}$.\footnote{Limited anticipation (Assumption~\ref{assn:limited_anticipation}) guarantees that $R_{it} = R_{it}(\infty)$ for $t < T_i - \delta$, so pre-event returns identify the no-event process.} Combining this with the fact that $\varepsilon_{st} \xrightarrow{p} 0$ as $n_s \to \infty$ yields, from~\eqref{eq:bias_ar},
\[
\tau^{AR}(s,t) - \tau^{ATT}(s,t) 
    \xrightarrow{p} (\alpha_s - \tilde{\alpha}_s) + (\beta_s \mathbf{F}_t - \tilde{\beta}_s \mathbf{F}^o_t),
\]
which is equation~\eqref{eq:plim_ar}.

\smallskip\noindent
\emph{Difference-in-means estimator.}
The difference-in-means estimator does not involve any pre-event estimation, so $T_{pre}$ is irrelevant here. Using \eqref{eq:bias_cont} and the fact that $\varepsilon_{st} - \varepsilon_{\infty,t} \xrightarrow{p} 0$ as $(n_s, n_c) \to \infty$, we obtain
\[
\hat{\tau}^{cont}(s,t) - \tau^{ATT}(s,t)
    \xrightarrow{p} (\alpha_s - \alpha_\infty) + (\beta_s - \beta_\infty)\mathbf{F}_t,
\]
which is equation~\eqref{eq:plim_cont}.

\smallskip\noindent
\emph{Synthetic control estimator.}
For the synthetic control estimator, we specialize $\hat{\tau}^{alt}(s,t)$ in~\eqref{eq:bias_synth} to $\hat{\tau}^{synth}(s,t)$ and denote the implied loadings by $(\hat{\alpha}^{synth}_s,\hat{\beta}^{synth}_s)$.

Under Assumption~\ref{assn:factor_model} together with a standard interactive fixed-effects structure for the untreated potential outcomes,
\[
R_{it}(\infty) = \alpha_i + \boldsymbol{\beta}_i' \mathbf{F}_t + \varepsilon_{it},
\]
the average untreated return for cohort $s$ can be written as
\[
\mathbb{E}\!\left(R_{it}(\infty) \mid T_i = s\right)
    = \alpha_s + \boldsymbol{\beta}_s' \mathbf{F}_t,
\]
with an analogous representation for each control unit $j \in \mathcal{C}$.

Ferman~\parencite{ferman2021properties} shows that, under such a factor structure, with sufficiently many pre-treatment periods and control units, the synthetic control weights constructed by minimizing the pre-treatment mean squared error recover the factor loadings of the treated unit (here, the treated cohort) in probability. Formally, applying their Theorem~1 to the cohort-level treated unit $R_{s,t}$, we obtain
\[
\hat{\alpha}^{synth}_s \xrightarrow{p} \alpha_s,
\qquad
\hat{\beta}^{synth}_s \xrightarrow{p} \beta_s
\quad\text{as } n_c, T_{pre} \to \infty.
\]
Combining this with $\varepsilon_{st} \xrightarrow{p} 0$ and~\eqref{eq:bias_synth} yields
\[
\hat{\tau}^{synth}(s,t) - \tau^{ATT}(s,t)
    \xrightarrow{p} 0,
\]
which is equation~\eqref{eq:plim_synth}.

\medskip
\noindent
\textit{Remark (Gsynth).} The same logic applies to the generalized synthetic control estimator in Definition~\ref{def:gsynth}. Under the interactive fixed-effects model 
\[
R_{it}(\infty) = \alpha_i + \boldsymbol{\lambda}_i' \mathbf{F}_t + \varepsilon_{it}
\]
and the regularity conditions in \textcite{xu2017generalized}, Xu shows that the counterfactual returns $\hat{R}^{GS}_{it}(\infty)$ are consistent for $R_{it}(\infty)$, uniformly over post-treatment periods. Aggregating over $i$ within cohort $s$ then implies that the cohort-period ATT estimated by Gsynth is also consistent:
\[
\hat{\tau}^{GS}(s,t) - \tau^{ATT}(s,t) \xrightarrow{p} 0.
\]

\paragraph{(2) Consistency of the difference-in-means estimator under random assignment.}

Assume now the random assignment condition in Assumption~\ref{assn:event_assignment}:
\[
p_t(\mathbf{X}_i,\mathbf{F}) = p_t(\mathbf{F}),
\]
so that $T_i$ is independent of $\mathbf{X}_i = (\alpha_i,\boldsymbol{\beta}_i)$. Then the distribution of $\mathbf{X}_i$ is the same in every treatment cohort $s$ and in the never-treated group $\mathcal{C}$. In particular,
\[
\alpha_s = \mathbb{E}(\alpha_i \mid T_i = s) = \mathbb{E}(\alpha_i \mid i \in \mathcal{C}) = \alpha_\infty,
\]
and similarly $\beta_s = \beta_\infty$.

Substituting these equalities into the probability limit in~\eqref{eq:plim_cont} gives
\[
\hat{\tau}^{cont}(s,t) - \tau^{ATT}(s,t) 
    \xrightarrow{p} 0.
\]
Note that this result only relies on large $n_s, n_c$; it does not require $T_{pre} \to \infty$ because the difference-in-means estimator does not use pre-event estimation.

\paragraph{(3) Consistency of the abnormal returns estimator under correct specification.}

Finally, suppose that the factor model is correctly specified in the abnormal-returns regression, i.e.\ $\mathbf{F}^o_t = \mathbf{F}_t$ for all $t$. For each security $i$,
\[
R_{it}(\infty) = \alpha_i + \boldsymbol{\beta}_i'\mathbf{F}_t + \varepsilon_{it},
\quad t < T_i - \delta.
\]
By Assumption~\ref{assn:limited_anticipation}, these pre-event observations coincide with the no-event potential outcome, and standard OLS consistency arguments imply that, as $T_{pre} \to \infty$,
\[
\hat{\alpha}_i \xrightarrow{p} \alpha_i,
\qquad
\hat{\boldsymbol{\beta}}_i \xrightarrow{p} \boldsymbol{\beta}_i.
\]
Averaging within cohort $s$ and applying the law of large numbers as $n_s \to \infty$ gives
\[
\hat{\alpha}_s \equiv \frac{1}{n_s} \sum_{i:T_i=s} \hat{\alpha}_i \xrightarrow{p} 
\frac{1}{n_s} \sum_{i:T_i=s} \alpha_i \xrightarrow{p} \alpha_s,
\]
and analogously $\hat{\beta}_s \xrightarrow{p} \beta_s$. Hence, under correct specification,
\[
\tilde{\alpha}_s = \alpha_s, \qquad \tilde{\beta}_s = \beta_s.
\]
Substituting these equalities into~\eqref{eq:plim_ar} yields
\[
\tau^{AR}(s,t) - \tau^{ATT}(s,t) \xrightarrow{p} 0,
\]
so the abnormal returns estimator is consistent for $\tau^{ATT}(s,t)$ when the factor structure is correctly specified.

This completes the proof.
\end{proof}

\begin{proof}[Proof of Theorem \ref{thm:bias}]
Throughout, maintain \Cref{assn:factor_model,assn:limited_anticipation} and the auxiliary regularity conditions used in the proof of Proposition~\ref{prop:consistency} (mean-zero idiosyncratic shocks with a law of large numbers across $i$, and standard OLS regularity for the abnormal-returns regressions, plus the spanning/interactive fixed effects conditions for synthetic control and gsynth).

Recall that for any estimator $\star \in \{AR,cont,synth,GS\}$ and horizon $\kappa \ge 0$,
\[
\theta_\kappa^{ATT}
    = \sum_{s \in \mathcal{S}} w_s \tau^{ATT}(s,s+\kappa),
    \qquad
\hat{\theta}_\kappa^{\star}
    = \sum_{s \in \mathcal{S}} w_s \hat{\tau}^{\star}(s,s+\kappa),
\]
with weights $w_s = N_s / \sum_{s' \in \mathcal{S}} N_{s'}$. Hence
\begin{equation}
\label{eq:theta_diff}
\hat{\theta}_\kappa^{\star} - \theta_\kappa^{ATT}
    = \sum_{s \in \mathcal{S}} w_s \bigl(\hat{\tau}^{\star}(s,s+\kappa) - \tau^{ATT}(s,s+\kappa)\bigr).
\end{equation}

\paragraph{(1) Unbiasedness of synthetic control and gsynth.}

From Proposition~\ref{prop:consistency}, for each fixed cohort $s$ and period $t$,
\[
\hat{\tau}^{alt}(s,t) - \tau^{ATT}(s,t) \xrightarrow{p} 0
\quad\text{for } alt \in \{synth,GS\}
\]
as $n_s, n_c, T_{pre} \to \infty$ under the conditions of \textcite{ferman2021properties} (for synthetic control) and \textcite{xu2017generalized} (for gsynth). Setting $t = s+\kappa$ and substituting into \eqref{eq:theta_diff} yields
\[
\hat{\theta}_\kappa^{alt} - \theta_\kappa^{ATT}
    = \sum_{s \in \mathcal{S}} w_s \bigl(\hat{\tau}^{alt}(s,s+\kappa) - \tau^{ATT}(s,s+\kappa)\bigr).
\]
Since $\mathcal{S} \subseteq \{1,\dots,T\}$ is finite and the weights satisfy $0 \le w_s \le 1$ and $\sum_s w_s = 1$, a finite linear combination of terms that converge in probability to zero also converges to zero. Thus,
\[
\hat{\theta}_\kappa^{alt} - \theta_\kappa^{ATT} \xrightarrow{p} 0,
\]
which proves part~(1).

\paragraph{(2) Bias of abnormal-returns and difference-in-means estimators.}

Assume $|\mathcal{S}| > 0$ and $1 > p_t(\mathbf{X}_i,\mathbf{F}) > \epsilon > 0$. The lower bound $\epsilon$ guarantees that each event time in $\mathcal{S}$ occurs with positive probability in the population, so $N_s$ and $\sum_{s'} N_{s'}$ both diverge with $N$ and the cohort weights converge:
\[
w_s 
= \frac{N_s}{\sum_{s' \in \mathcal{S}} N_{s'}} 
\xrightarrow{p} 
\pi_s \equiv \Pr(T_i = s \mid T_i \in \mathcal{S}).
\]

From Proposition~\ref{prop:consistency}, for each fixed $s$ and $t$,
\begin{align*}
\tau^{AR}(s,t) - \tau^{ATT}(s,t)
    &\xrightarrow{p} (\alpha_s - \tilde{\alpha}_s) + (\beta_s \mathbf{F}_t - \tilde{\beta}_s \mathbf{F}^o_t), \\
\hat{\tau}^{cont}(s,t) - \tau^{ATT}(s,t)
    &\xrightarrow{p} (\alpha_s - \alpha_\infty) + (\beta_s - \beta_\infty)\mathbf{F}_t.
\end{align*}
Evaluating at $t = s+\kappa$ and plugging into \eqref{eq:theta_diff}, we obtain
\begin{align*}
\hat{\theta}_\kappa^{ar} - \theta_\kappa^{ATT}
    &\xrightarrow{p} 
      \sum_{s \in \mathcal{S}} \pi_s \Bigl[(\alpha_s - \tilde{\alpha}_s) + \bigl(\beta_s \mathbf{F}_{s+\kappa} - \tilde{\beta}_s \mathbf{F}^o_{s+\kappa}\bigr)\Bigr], \\
\hat{\theta}_\kappa^{cont} - \theta_\kappa^{ATT}
    &\xrightarrow{p} 
      \sum_{s \in \mathcal{S}} \pi_s \Bigl[(\alpha_s - \alpha_\infty) + (\beta_s - \beta_\infty)\mathbf{F}_{s+\kappa}\Bigr].
\end{align*}

Define a random event time $S$ with $\Pr(S = s \mid T_i \in \mathcal{S}) = \pi_s$. Then the limits above can be written compactly as conditional expectations over treated cohorts:
\begin{align*}
\hat{\theta}^{ar}_{\kappa} - \theta_{\kappa}^{ATT}
    &\xrightarrow{p}
    \mathbb{E}\left[ (\alpha_{S} - \tilde{\alpha}_{S}) 
        + \bigl(\beta_{S}\mathbf{F}_{S+\kappa} - \tilde{\beta}_{S}\mathbf{F}^{o}_{S+\kappa}\bigr)
        \,\Big|\, T_{i} \in \mathcal{S} \right], \\
\hat{\theta}^{cont}_{\kappa} - \theta_{\kappa}^{ATT}
    &\xrightarrow{p}
    \mathbb{E}\left[ (\alpha_{S} - \alpha_{\infty}) 
        + (\beta_{S}-\beta_{\infty})\mathbf{F}_{S+\kappa}
        \,\Big|\, T_{i} \in \mathcal{S} \right].
\end{align*}
Relabeling $S$ as $s$ inside the expectation gives the expressions stated in part~(2).

\paragraph{(3) Random assignment across firms.}

Under random assignment across firms, 
\[
p_t(\mathbf{X}_i,\mathbf{F}) = p_t(\mathbf{F}),
\]
so event assignment is independent of $\mathbf{X}_i = (\alpha_i,\beta_i)$. As in the proof of Proposition~\ref{prop:consistency}, this implies
\[
\alpha_s = \alpha_\infty, \qquad \beta_s = \beta_\infty
\quad\text{for all } s \in \mathcal{S}.
\]
Substituting these equalities into the limit for $\hat{\theta}^{cont}_\kappa - \theta_\kappa^{ATT}$ in part~(2) gives
\[
\hat{\theta}^{cont}_\kappa - \theta_\kappa^{ATT} \xrightarrow{p} 0
\]
as $n_s, n_c \to \infty$, even for fixed $T_{pre}$. This proves part~(3).

\paragraph{(4) Random timing.}

Now assume random timing,
\[
p_t(\mathbf{X}_i,\mathbf{F}) = p_t(\mathbf{X}_i),
\]
so that event timing is independent of the factor path $\mathbf{F}$, and adopt the standard assumption that firm characteristics $\mathbf{X}_i$ (hence $\alpha_i,\beta_i,\tilde{\alpha}_i,\tilde{\beta}_i$) are independent of $\mathbf{F}$. In addition, assume that the factors have constant mean over time:
\[
\mathbb{E}(\mathbf{F}_t) = \mathbb{E}(\mathbf{F}_{t'}) \equiv \mathbb{E}(\mathbf{F}_t)
\quad\text{for all } t,t',
\]
and similarly for $\mathbf{F}^o_t$.

Start from the general bias expressions in part~(2):
\begin{align*}
\hat{\theta}^{ar}_{\kappa} - \theta_{\kappa}^{ATT}
    &\xrightarrow{p}
    \mathbb{E}\left[ (\alpha_{s} - \tilde{\alpha}_{s}) 
        + \bigl(\beta_{s}\mathbf{F}_{s+\kappa} - \tilde{\beta}_{s}\mathbf{F}^{o}_{s+\kappa}\bigr)
        \,\Big|\, T_{i} \in \mathcal{S} \right], \\
\hat{\theta}^{cont}_{\kappa} - \theta_{\kappa}^{ATT}
    &\xrightarrow{p}
    \mathbb{E}\left[ (\alpha_{s} - \alpha_{\infty}) 
        + (\beta_{s}-\beta_{\infty})\mathbf{F}_{s+\kappa}
        \,\Big|\, T_{i} \in \mathcal{S} \right].
\end{align*}

Under random timing and independence between $(\alpha_s,\beta_s,\tilde{\alpha}_s,\tilde{\beta}_s,s)$ and the factor process $\mathbf{F}$, we can factor the cross term:
\begin{align*}
\mathbb{E}\bigl(\beta_s \mathbf{F}_{s+\kappa} \mid T_i \in \mathcal{S}\bigr)
    &= \mathbb{E}\bigl(\beta_s \mid T_i \in \mathcal{S}\bigr)\,
       \mathbb{E}(\mathbf{F}_{s+\kappa})
     = \mathbb{E}\bigl(\beta_i \mid T_i \in \mathcal{S}\bigr)\,
       \mathbb{E}(\mathbf{F}_t), \\
\mathbb{E}\bigl(\tilde{\beta}_s \mathbf{F}^o_{s+\kappa} \mid T_i \in \mathcal{S}\bigr)
    &= \mathbb{E}\bigl(\tilde{\beta}_s \mid T_i \in \mathcal{S}\bigr)\,
       \mathbb{E}(\mathbf{F}^o_{s+\kappa})
     = \mathbb{E}\bigl(\tilde{\beta}_i \mid T_i \in \mathcal{S}\bigr)\,
       \mathbb{E}(\mathbf{F}^o_{s+\kappa}),
\end{align*}
where the last equalities use that the distribution of $s$ within $\mathcal{S}$ is the same as the distribution of $T_i$ conditional on $T_i \in \mathcal{S}$, and that factor means are time-invariant.

Thus the asymptotic bias of the abnormal-returns estimator becomes
\begin{align*}
\hat{\theta}^{ar}_{\kappa}-\theta_{\kappa}^{ATT}
    &\xrightarrow{p}
      \mathbb{E}\left(\alpha_{s} - \tilde{\alpha}_{s}\mid T_{i} \in \mathcal{S} \right)
      + \mathbb{E}\left(\beta_{i}\mid T_{i} \in \mathcal{S} \right)\mathbb{E}\left(\mathbf{F}_{t} \right) \\
    &\phantom{\xrightarrow{p}=}\;
      - \mathbb{E}\left(\tilde{\beta}_{i}\mid T_{i} \in \mathcal{S} \right)\mathbb{E}\left(\mathbf{F}^{o}_{s+\kappa} \right),
\end{align*}
which is the expression stated in part~(4) for $\hat{\theta}^{ar}_\kappa$.

Similarly, for the difference-in-means estimator,
\begin{align*}
\mathbb{E}\bigl((\beta_s - \beta_\infty)\mathbf{F}_{s+\kappa} \mid T_i \in \mathcal{S}\bigr)
    &= \mathbb{E}(\beta_s - \beta_\infty \mid T_i \in \mathcal{S}) \, \mathbb{E}(\mathbf{F}_{s+\kappa}) \\
    &= \mathbb{E}(\beta_s - \beta_\infty \mid T_i \in \mathcal{S}) \, \mathbb{E}(\mathbf{F}_t),
\end{align*}
so the asymptotic bias simplifies to
\[
\hat{\theta}^{cont}_{\kappa} -\theta_{\kappa}^{ATT}  
    \xrightarrow{p}
    \mathbb{E}\left(\alpha_{s} - \alpha_{\infty}\mid T_{i} \in \mathcal{S} \right)
    + \mathbb{E}\left(\beta_{s}-\beta_{\infty}\mid T_{i} \in \mathcal{S} \right)\mathbb{E}\left(\mathbf{F}_{t}\right),
\]
as claimed.

This establishes all four parts of the theorem.
\end{proof}
\begin{proof}[Proof of Lemma \ref{lem:geo_vol}]
Throughout, we work with the second–order Taylor expansion of $\log(1+x)$ around $x=0$,
\[
\log(1+x) = x - \tfrac{1}{2}x^2 + r(x),
\quad\text{with}\quad r(x) = O(x^3),
\]
and omit the remainder term $r(x)$ for notational simplicity. All equalities below should be read as holding up to these higher–order terms in returns.

\paragraph{Step 1: Period--by--period geometric ATT.}

Fix an event cohort $s$ and calendar period $t$. By definition,
\[
\tau^{geo,ATT}(s,t)
    = E\big(\log(1+R_{it}(s)) - \log(1+R_{it}(\infty)) \mid T_i = s\big).
\]
Using the second–order expansion,
\begin{align*}
\log(1+R_{it}(s))
    &\approx R_{it}(s) - \tfrac{1}{2}R_{it}(s)^2, \\
\log(1+R_{it}(\infty))
    &\approx R_{it}(\infty) - \tfrac{1}{2}R_{it}(\infty)^2,
\end{align*}
so that
\begin{align*}
\tau^{geo,ATT}(s,t)
    &\approx E\left[
        \big(R_{it}(s) - R_{it}(\infty)\big)
        - \tfrac{1}{2}\big(R_{it}(s)^2 - R_{it}(\infty)^2\big)
        \,\Big|\, T_i = s\right].
\end{align*}

Let the individual treatment effect be
\[
\tau_i(s,t) = R_{it}(s) - R_{it}(\infty),
\]
so that $R_{it}(s) = R_{it}(\infty) + \tau_i(s,t)$. Then
\[
R_{it}(s)^2 - R_{it}(\infty)^2
    = \big(R_{it}(\infty) + \tau_i(s,t)\big)^2 - R_{it}(\infty)^2
    = 2 R_{it}(\infty)\tau_i(s,t) + \tau_i(s,t)^2.
\]
Substituting this into the expression above,
\begin{align*}
\tau^{geo,ATT}(s,t)
    &\approx E\left[
        \tau_i(s,t)
        - \tfrac{1}{2}\big(2 R_{it}(\infty)\tau_i(s,t) + \tau_i(s,t)^2\big)
        \,\Big|\, T_i = s\right] \\
    &= E\big(\tau_i(s,t) \mid T_i = s\big)
       - E\left( R_{it}(\infty)\tau_i(s,t) + \tfrac{1}{2}\tau_i(s,t)^2
            \,\Big|\, T_i = s\right).
\end{align*}
By definition of the arithmetic cohort--period ATT,
\[
\tau^{ATT}(s,t) = E\big(\tau_i(s,t) \mid T_i = s\big),
\]
so we have the key period--by--period relationship
\begin{equation}
\label{eq:geo_vs_ar_per_period}
\tau^{geo,ATT}(s,t)
    \approx \tau^{ATT}(s,t)
      - E\left( R_{it}(\infty)\tau_i(s,t)
                  + \tfrac{1}{2}\tau_i(s,t)^2
            \,\Big|\, T_i = s\right).
\end{equation}

\paragraph{Step 2: From period ATT to horizon $H$.}

For cohort $s$, the geometric ATT over horizon $H$ is
\[
\tau^{geo,ATT}(s,H)
    = \sum_{\kappa=0}^H \tau^{geo,ATT}(s,s+\kappa),
\]
and the corresponding arithmetic CATT is
\[
\tau^{CATT}(s,H)
    = \sum_{\kappa=0}^H \tau^{ATT}(s,s+\kappa).
\]
Summing \eqref{eq:geo_vs_ar_per_period} over $\kappa=0,\dots,H$ gives
\begin{align*}
\tau^{geo,ATT}(s,H)
    &\approx \sum_{\kappa=0}^H \tau^{ATT}(s,s+\kappa)
      - \sum_{\kappa=0}^{H}
        E\left( R_{i,s+\kappa}(\infty)\tau_i(s,s+\kappa)
                  + \tfrac{1}{2}\tau_i(s,s+\kappa)^2
            \,\Big|\, T_i = s\right).
\end{align*}

Now average across event cohorts with weights $w_s$:
\[
\theta_H^{geo,ATT}
    = \sum_{s} w_s \tau^{geo,ATT}(s,H),
    \qquad
\theta_H^{ATT} = \sum_{\kappa=0}^H \theta_\kappa^{ATT}
                 = \sum_s w_s \sum_{\kappa=0}^H \tau^{ATT}(s,s+\kappa).
\]
Thus,
\begin{align*}
\theta_H^{geo,ATT}
    &\approx \sum_s w_s \sum_{\kappa=0}^H \tau^{ATT}(s,s+\kappa) \\
    &\quad - \sum_s w_s \sum_{\kappa=0}^{H}
        E\left( R_{i,s+\kappa}(\infty)\tau_i(s,s+\kappa)
                  + \tfrac{1}{2}\tau_i(s,s+\kappa)^2
            \,\Big|\, T_i = s\right) \\
    &= \theta_H^{ATT}
       - \sum_s w_s \sum_{\kappa=0}^{H}
        E\left( R_{i,s+\kappa}(\infty)\tau_i(s,s+\kappa)
                  + \tfrac{1}{2}\tau_i(s,s+\kappa)^2
            \,\Big|\, T_i = s\right),
\end{align*}
which is the first expression in Lemma~\ref{lem:geo_vol}.

\paragraph{Step 3: Independence and simplification.}

Now impose the additional assumption stated in the lemma: for all $s$ and $\kappa$,
\begin{itemize}
    \item $R_{i,s+\kappa}(\infty)$ and $\tau_i(s,s+\kappa)$ are independent conditional on $T_i=s$; and
    \item the conditional mean of the no–event return is constant,
    \[
        \mu = E\big(R_{i,s+\kappa}(\infty) \mid T_i = s\big)
    \]
    does not depend on $s$ or $\kappa$.
\end{itemize}
Then
\[
E\big(R_{i,s+\kappa}(\infty)\tau_i(s,s+\kappa) \mid T_i=s\big)
   = \mu\,E\big(\tau_i(s,s+\kappa) \mid T_i=s\big)
   = \mu\,\tau^{ATT}(s,s+\kappa),
\]
and the expression from Step~2 becomes
\begin{align*}
\theta_H^{geo,ATT}
    &\approx \theta_H^{ATT}
        - \sum_s w_s \sum_{\kappa=0}^{H}
           \Big[\mu\,\tau^{ATT}(s,s+\kappa)
                + \tfrac{1}{2}E\big(\tau_i(s,s+\kappa)^2 \mid T_i=s\big)\Big] \\
    &= \theta_H^{ATT}
        - \mu \sum_{\kappa=0}^H \sum_s w_s \tau^{ATT}(s,s+\kappa)
        - \tfrac{1}{2}\sum_{\kappa=0}^H \sum_s w_s
            E\big(\tau_i(s,s+\kappa)^2 \mid T_i=s\big).
\end{align*}
Using $\sum_s w_s \tau^{ATT}(s,s+\kappa) = \theta_\kappa^{ATT}$ and $\sum_{\kappa=0}^H \theta_\kappa^{ATT} = \theta_H^{ATT}$, we get
\[
\theta_H^{geo,ATT}
    \approx (1-\mu)\theta_H^{ATT}
      - \tfrac{1}{2}\sum_{\kappa=0}^H \sum_s w_s
         E\big(\tau_i(s,s+\kappa)^2 \mid T_i=s\big).
\]

To rewrite the last term in terms of variances, consider a randomly drawn treated security $i$ and define the individual treatment effect at event time $\kappa$ as
\[
\Delta_{i,\kappa} \equiv \tau_i(s,s+\kappa) \quad\text{for the (random) cohort } s=T_i.
\]
Under the cohort weights $w_s$, the distribution of $s$ among treated units satisfies
\[
\Pr(s = r \mid T_i\in\mathcal S) = w_r,
\]
so
\begin{align*}
E(\Delta_{i,\kappa} \mid T_i\in\mathcal S)
    &= \sum_s w_s E\big(\tau_i(s,s+\kappa) \mid T_i=s\big)
     = \sum_s w_s \tau^{ATT}(s,s+\kappa)
     = \theta^{ATT}_\kappa, \\
E(\Delta_{i,\kappa}^2 \mid T_i\in\mathcal S)
    &= \sum_s w_s E\big(\tau_i(s,s+\kappa)^2 \mid T_i=s\big).
\end{align*}
Denote the cross–sectional variance of individual treatment effects at event time $\kappa$ by
\[
\var(\theta^{ATT}_\kappa)
    \equiv \var\big(\Delta_{i,\kappa} \mid T_i\in\mathcal S\big).
\]
Then
\[
E(\Delta_{i,\kappa}^2 \mid T_i\in\mathcal S)
    = \var(\theta^{ATT}_\kappa) + \big(\theta^{ATT}_\kappa\big)^2,
\]
so that
\[
\sum_s w_s E\big(\tau_i(s,s+\kappa)^2 \mid T_i=s\big)
    = \var(\theta^{ATT}_\kappa) + \big(\theta^{ATT}_\kappa\big)^2.
\]
Substituting into the expression for $\theta_H^{geo,ATT}$ yields
\begin{align*}
\theta_H^{geo,ATT}
    &\approx (1-\mu)\theta_H^{ATT}
      - \tfrac{1}{2}\sum_{\kappa=0}^{H}
         \Big[\var(\theta^{ATT}_\kappa) + \big(\theta^{ATT}_\kappa\big)^2\Big],
\end{align*}
which is the second expression in Lemma~\ref{lem:geo_vol}.

This completes the proof.
\end{proof}

\include{appendix}

\end{document}

%% file: appendix.tex
\section{Additional Simulation Results}
For the simulation sample where treatment is selected based on loading to the second factor and random timing, We plot the bias from difference in mean, CAPM, and Gsynth estimators, across simulation samples.

\begin{figure}[H]
    \centering
    \caption{Bias from Difference-in-Mean Model on SMB Returns with Assignment Selection}\label{fig:simul_2f_biascapm_smb_assign_diff}
    \begin{justify}
    {\small \vspace*{0.05in}
    This figure plots the biases from a difference-in-mean estimator on the treatment period over realizations of the second factor across 50 simulations. We simulate 500 firms with 10\% of them getting treated. The estimation period is 239 days and post-event period is 11 days. More details on the simulations is in Section \ref{sec:simul_2f_select_design}. Panel A reports simulation results with no selections, Panel B with only assignment selection, Panel C with only timing selection, and Panel D with both. We consider several estimators: difference in simple average, CAPM and 2-factor abnormal returns, and generalized synthetic methods. The expected biases and coverage are from 50 simulations.}
    \end{justify}
    \includegraphics[width=0.7\linewidth]{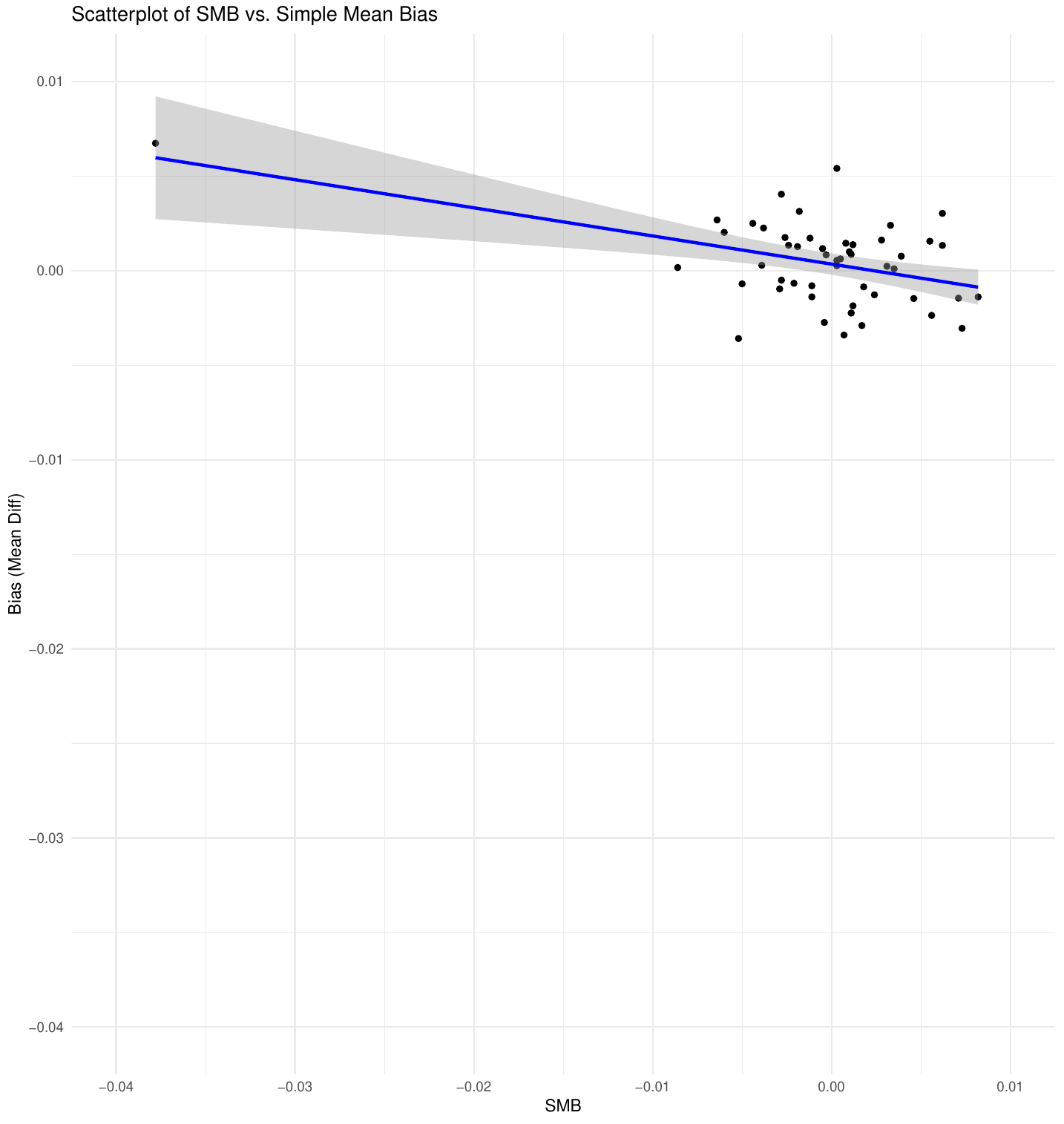}
\end{figure}

\begin{figure}[H]
    \centering
    \caption{Bias from Gsynth Model on SMB Returns with Assignment Selection}\label{fig:simul_2f_biascapm_smb_assign_gsynth}
    \begin{justify}
    {\small \vspace*{0.05in}
    This figure plots the biases from a Gsynth estimator on the treatment period over realizations of the second factor across 50 simulations. We simulate 500 firms with 10\% of them getting treated. The estimation period is 239 days and post-event period is 11 days. More details on the simulations is in Section \ref{sec:simul_2f_select_design}. Panel A reports simulation results with no selections, Panel B with only assignment selection, Panel C with only timing selection, and Panel D with both. We consider several estimators: difference in simple average, CAPM and 2-factor abnormal returns, and generalized synthetic methods. The expected biases and coverage are from 50 simulations.}
    \end{justify}
    \includegraphics[width=0.7\linewidth]{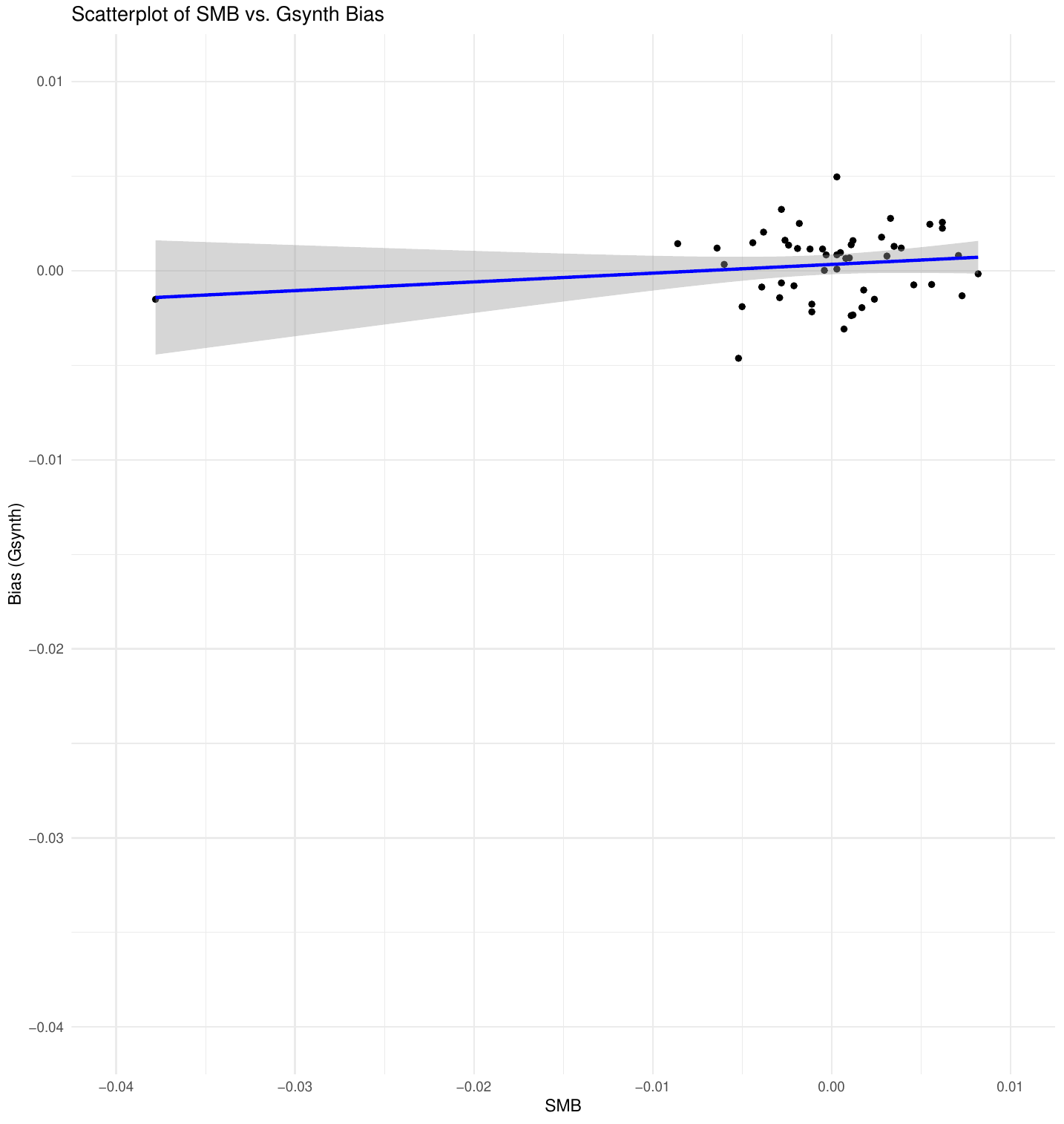}
\end{figure}

\section{Additional Results for Geithner}
\label{appsec:geithner}

This section presents addition results for the \textcite{acemoglu2016value} empirical example.

\subsection{Period-by-Period ATT}
In this section, we compare how different counterfactual affects the daily ATT in the post-event period. The `Average' column computes the difference in the simple mean of treated versus control firms, as reported in Panel A of Table 2 in the original paper. The `Synthetic Control' column computes the weighted average daily return with synthetic control weights, and the `Synthetic Diff-in-Diff' column uses the synthetic diff-in-diff weights instead.

For standard errors, in the `Average' column, we use the same approach as the original paper and adjust the standard errors for pre-event correlation between firms. In the `Synthetic Control' column, we report bootstrap standard errors estimated separately for each period. Since synthetic control weights will not change with the post-period, this method gives the correct standard errors period-by-period. We cannot use the same methodology for synthetic difference-in-differences because the estimated unit weight also depends on the data from the post-period.

We see that with synthetic control weights, the estimated ATT is much smaller compared to the simple mean.
\begin{table}[H]
  \centering
  \caption{Period-by-Period ATT to Geithner Announcement (Schedule connections)}
  \resizebox{\textwidth}{!}{%
    \begin{tabular}{rrrrrrrrrrrrr}
    \toprule
          &       & \multicolumn{3}{c}{Average} &       & \multicolumn{3}{c}{Synthetic Control} &       & \multicolumn{3}{c}{Synthetic Diff-in-Diff} \\
\cmidrule{3-5}\cmidrule{7-9}\cmidrule{11-13}    \multicolumn{1}{l}{Event day} & \multicolumn{1}{l}{Date} & \multicolumn{1}{l}{Conn.} & \multicolumn{1}{l}{Non-conn.} & \multicolumn{1}{l}{Difference} &       & \multicolumn{1}{l}{Conn.} & \multicolumn{1}{l}{Non-conn.} & \multicolumn{1}{l}{Difference} &       & \multicolumn{1}{l}{Conn.} & \multicolumn{1}{l}{Non-conn.} & \multicolumn{1}{l}{Difference} \\
    \midrule
    0     & 11/21/08 & 0.086 & 0.042 & 0.043*** &       & 0.086 & 0.066 & 0.019* &       & 0.086 & 0.058 & 0.028 \\
    1     & 11/24/08 & 0.130 & 0.046 & 0.084*** &       & 0.130 & 0.080 & 0.050** &       & 0.130 & 0.063 & 0.067 \\
    2     & 11/25/08 & 0.026 & 0.015 & 0.011 &       & 0.026 & 0.045 & -0.019 &       & 0.026 & 0.018 & 0.008 \\
    3     & 11/26/08 & 0.112 & 0.041 & 0.071*** &       & 0.112 & 0.070 & 0.042 &       & 0.112 & 0.055 & 0.057 \\
    4     & 11/28/08 & 0.056 & 0.018 & 0.038** &       & 0.056 & 0.028 & 0.027 &       & 0.056 & 0.025 & 0.030 \\
    5     & 12/1/08 & -0.131 & -0.076 & -0.056*** &       & -0.131 & -0.119 & -0.013 &       & -0.131 & -0.102 & -0.030 \\
    6     & 12/2/08 & 0.046 & 0.043 & 0.003 &       & 0.046 & 0.039 & 0.007 &       & 0.046 & 0.056 & -0.010 \\
    7     & 12/3/08 & 0.034 & 0.018 & 0.016 &       & 0.034 & 0.035 & -0.001 &       & 0.034 & 0.024 & 0.011 \\
    8     & 12/4/08 & -0.009 & -0.013 & 0.005 &       & -0.009 & -0.028 & 0.019 &       & -0.009 & -0.016 & 0.008 \\
    9     & 12/5/08 & 0.063 & 0.024 & 0.038** &       & 0.063 & 0.034 & 0.028** &       & 0.063 & 0.031 & 0.031 \\
    10    & 12/8/08 & 0.064 & 0.027 & 0.037** &       & 0.064 & 0.047 & 0.017 &       & 0.064 & 0.033 & 0.031 \\
    \bottomrule
    \end{tabular}%
    }
\end{table}%
\begin{table}[H]
  \centering
  \caption{Period-by-Period ATT to Geithner Announcement (Personal connections)}
  \resizebox{\textwidth}{!}{%
    \begin{tabular}{rrrrrrrrrrrrr}
    \toprule
          &       & \multicolumn{3}{c}{Average} &       & \multicolumn{3}{c}{Synthetic Control} &       & \multicolumn{3}{c}{Synthetic Diff-in-Diff} \\
\cmidrule{3-5}\cmidrule{7-9}\cmidrule{11-13}    \multicolumn{1}{l}{Event day} & \multicolumn{1}{l}{Date} & \multicolumn{1}{l}{Conn.} & \multicolumn{1}{l}{Non-conn.} & \multicolumn{1}{l}{Difference} &       & \multicolumn{1}{l}{Conn.} & \multicolumn{1}{l}{Non-conn.} & \multicolumn{1}{l}{Difference} &       & \multicolumn{1}{l}{Conn.} & \multicolumn{1}{l}{Non-conn.} & \multicolumn{1}{l}{Difference} \\
    \midrule
    0     & 11/21/08 & 0.075 & 0.043 & 0.033 &       & 0.075 & 0.073 & 0.003 &       & 0.075 & 0.069 & 0.007 \\
    1     & 11/24/08 & 0.143 & 0.047 & 0.096*** &       & 0.143 & 0.106 & 0.037 &       & 0.143 & 0.074 & 0.069 \\
    2     & 11/25/08 & 0.057 & 0.014 & 0.043* &       & 0.057 & 0.059 & -0.002 &       & 0.057 & 0.023 & 0.034 \\
    3     & 11/26/08 & 0.112 & 0.042 & 0.071*** &       & 0.112 & 0.113 & 0.000 &       & 0.112 & 0.070 & 0.042 \\
    4     & 11/28/08 & 0.085 & 0.018 & 0.067*** &       & 0.085 & 0.077 & 0.008 &       & 0.085 & 0.031 & 0.054 \\
    5     & 12/1/08 & -0.144 & -0.076 & -0.067*** &       & -0.144 & -0.140 & -0.004 &       & -0.144 & -0.121 & -0.023 \\
    6     & 12/2/08 & 0.044 & 0.043 & 0.001 &       & 0.044 & 0.063 & -0.019 &       & 0.044 & 0.066 & -0.022 \\
    7     & 12/3/08 & 0.043 & 0.018 & 0.024 &       & 0.043 & 0.033 & 0.010 &       & 0.043 & 0.025 & 0.017 \\
    8     & 12/4/08 & 0.005 & -0.014 & 0.019 &       & 0.005 & -0.024 & 0.029 &       & 0.005 & -0.015 & 0.020 \\
    9     & 12/5/08 & 0.042 & 0.025 & 0.017 &       & 0.042 & 0.046 & -0.004 &       & 0.042 & 0.039 & 0.003 \\
    10    & 12/8/08 & 0.043 & 0.028 & 0.015 &       & 0.043 & 0.055 & -0.012 &       & 0.043 & 0.042 & 0.002 \\
    \bottomrule
    \end{tabular}%
    }
\end{table}%
\begin{table}[H]
  \centering
  \caption{Period-by-Period ATT to Geithner Announcement (New York connections)}
  \resizebox{\textwidth}{!}{
    \begin{tabular}{rrrrrrrrrrrrr}
    \toprule
          &       & \multicolumn{3}{c}{Average} &       & \multicolumn{3}{c}{Synthetic Control} &       & \multicolumn{3}{c}{Synthetic Diff-in-Diff} \\
\cmidrule{3-5}\cmidrule{7-9}\cmidrule{11-13}    \multicolumn{1}{l}{Event day} & \multicolumn{1}{l}{Date} & \multicolumn{1}{l}{Conn.} & \multicolumn{1}{l}{Non-conn.} & \multicolumn{1}{l}{Difference} &       & \multicolumn{1}{l}{Conn.} & \multicolumn{1}{l}{Non-conn.} & \multicolumn{1}{l}{Difference} &       & \multicolumn{1}{l}{Conn.} & \multicolumn{1}{l}{Non-conn.} & \multicolumn{1}{l}{Difference} \\
    \midrule
    0     & 11/21/08 & 0.085 & 0.040 & 0.044*** &       & 0.085 & 0.069 & 0.016* &       & 0.085 & 0.051 & 0.033 \\
    1     & 11/24/08 & 0.078 & 0.046 & 0.031*** &       & 0.078 & 0.082 & -0.004 &       & 0.078 & 0.058 & 0.020 \\
    2     & 11/25/08 & 0.032 & 0.014 & 0.018 &       & 0.032 & 0.011 & 0.021* &       & 0.032 & 0.016 & 0.016 \\
    3     & 11/26/08 & 0.087 & 0.040 & 0.048*** &       & 0.087 & 0.065 & 0.022 &       & 0.087 & 0.048 & 0.040 \\
    4     & 11/28/08 & 0.016 & 0.019 & -0.003 &       & 0.016 & 0.023 & -0.006 &       & 0.016 & 0.022 & -0.005 \\
    5     & 12/1/08 & -0.105 & -0.075 & -0.030*** &       & -0.105 & -0.106 & 0.001 &       & -0.105 & -0.093 & -0.012 \\
    6     & 12/2/08 & 0.090 & 0.040 & 0.050*** &       & 0.090 & 0.052 & 0.037*** &       & 0.090 & 0.050 & 0.039 \\
    7     & 12/3/08 & 0.031 & 0.018 & 0.013 &       & 0.031 & 0.025 & 0.005 &       & 0.031 & 0.021 & 0.009 \\
    8     & 12/4/08 & -0.020 & -0.013 & -0.008 &       & -0.020 & -0.031 & 0.010 &       & -0.020 & -0.014 & -0.006 \\
    9     & 12/5/08 & 0.050 & 0.024 & 0.026** &       & 0.050 & 0.046 & 0.004 &       & 0.050 & 0.029 & 0.021 \\
    10    & 12/8/08 & 0.050 & 0.027 & 0.023** &       & 0.050 & 0.055 & -0.006 &       & 0.050 & 0.031 & 0.018 \\
    \bottomrule
    \end{tabular}%
  }
\end{table}%

\subsection{Placebo Period ATT}
\begin{table}[H]
    \centering
    \caption{Placebo Period ATT to Geithner Announcement (Schedule connections)}
{
\def\sym#1{\ifmmode^{#1}\else\(^{#1}\)\fi}
\begin{tabular}{l*{4}{c}}
\hline\hline
            &\multicolumn{1}{c}{(1)}&\multicolumn{1}{c}{(2)}&\multicolumn{1}{c}{(3)}&\multicolumn{1}{c}{(4)}\\
            &\multicolumn{1}{c}{Average}&\multicolumn{1}{c}{DID}&\multicolumn{1}{c}{SC}&\multicolumn{1}{c}{SDID}\\
\hline
Treated     &   -0.006* &      -0.006** &      -0.004   &      -0.003   \\
            &   (0.004) &     (0.003)   &     (0.003)   &     (0.003)   \\
\hline
Observations&   16,350 &     139,520   &     139,520   &     139,520   \\
\hline\hline
\multicolumn{4}{l}{\footnotesize Standard errors in parentheses}\\
\multicolumn{4}{l}{\footnotesize * p<0.10, ** p<0.05, *** p<0.01}\\
\end{tabular}
}
\end{table}
\begin{table}[H]
    \centering
    \caption{Placebo Period ATT to Geithner Announcement (Personal connections)}
{
\def\sym#1{\ifmmode^{#1}\else\(^{#1}\)\fi}
\begin{tabular}{l*{4}{c}}
\hline\hline
            &\multicolumn{1}{c}{(1)}&\multicolumn{1}{c}{(2)}&\multicolumn{1}{c}{(3)}&\multicolumn{1}{c}{(4)}\\
            &\multicolumn{1}{c}{Average}&\multicolumn{1}{c}{DID}&\multicolumn{1}{c}{SC}&\multicolumn{1}{c}{SDID}\\
\hline
Treated     &   -0.007 &      -0.006** &       0.001   &      -0.002   \\
            &   (0.005) &     (0.002)   &     (0.003)   &     (0.003)   \\
\hline
Observations&   16,350 &     139,520   &     139,520   &     139,520   \\
\hline\hline
\multicolumn{4}{l}{\footnotesize Standard errors in parentheses}\\
\multicolumn{4}{l}{\footnotesize * p<0.10, ** p<0.05, *** p<0.01}\\
\end{tabular}
}
\end{table}
\begin{table}[H]
    \centering
    \caption{Placebo Period ATT to Geithner Announcement (New York connections)}
{
\def\sym#1{\ifmmode^{#1}\else\(^{#1}\)\fi}
\begin{tabular}{l*{4}{c}}
\hline\hline
            &\multicolumn{1}{c}{(1)}&\multicolumn{1}{c}{(2)}&\multicolumn{1}{c}{(3)}&\multicolumn{1}{c}{(4)}\\
            &\multicolumn{1}{c}{Average}&\multicolumn{1}{c}{DID}&\multicolumn{1}{c}{SC}&\multicolumn{1}{c}{SDID}\\
\hline
Treated     &    -0.003 &      -0.002   &      -0.000   &      -0.000   \\
            &   (0.002) &     (0.001)   &     (0.001)   &     (0.001)   \\
\hline
Observations&   16,350 &     139,520   &     139,520   &     139,520   \\
\hline\hline
\multicolumn{4}{l}{\footnotesize Standard errors in parentheses}\\
\multicolumn{4}{l}{\footnotesize * p<0.10, ** p<0.05, *** p<0.01}\\
\end{tabular}
}
\end{table}

\subsection{Placebo Period}
In this section, we test how synthetic methods perform in a placebo period before the event. The placebo period is day -30 to day -1, which is not used in estimation but also the event is not yet happening. If we assume that synthetic methods perform well in capturing the underlying factor structure and the factor loadings stay stable before the event, we would expect that the ATT in the placebo period is close to 0.

Figure \ref{fig:sc_att_sep} plots the average treatment effect of raw returns on the left and the average treatment effect of abnormal returns (relative to a CAPM model with beta estimated using daily returns from day -280 to -31). In Figure \ref{fig:sc_att_pool}, we plot all the ATT on one graph for better comparison.

We see that synthetic control does the best job in the placebo period, but also has the least treatment effect post-period. By comparing the treatment effect of raw returns using synthetic controls with the treatment effect of abnormal returns with a simple average, we see that they are relatively close, which suggests that synthetic control does a good job matching the underlying market beta exposure of treatment firms.

\begin{figure}[H]
    \centering
    \caption{Period-by-Period ATT in Placebo and Post Period (Schedule connections)}\label{fig:sc_att_sep}
    \begin{minipage}[b]{0.48\textwidth}
    \includegraphics[scale=0.53]{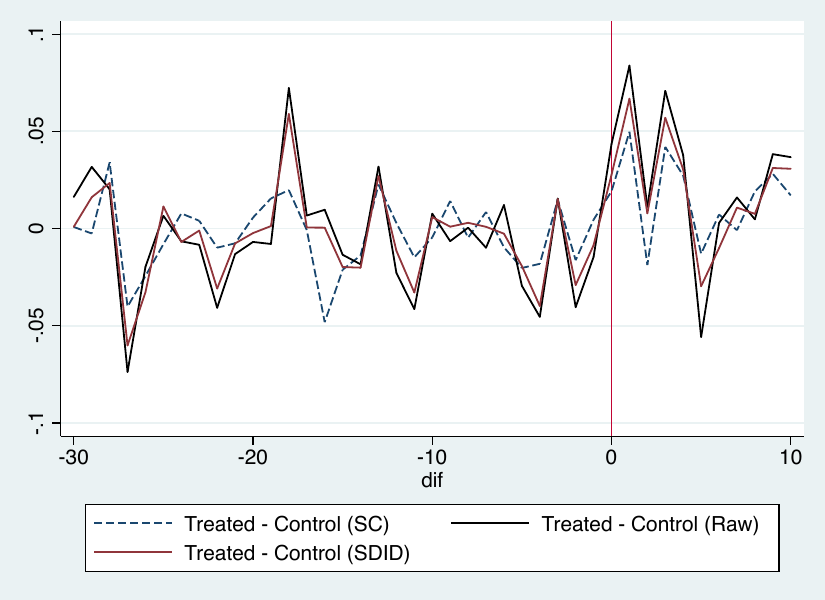}
    \end{minipage}
    \quad
    \begin{minipage}[b]{0.48\textwidth}
    \includegraphics[scale=0.53]{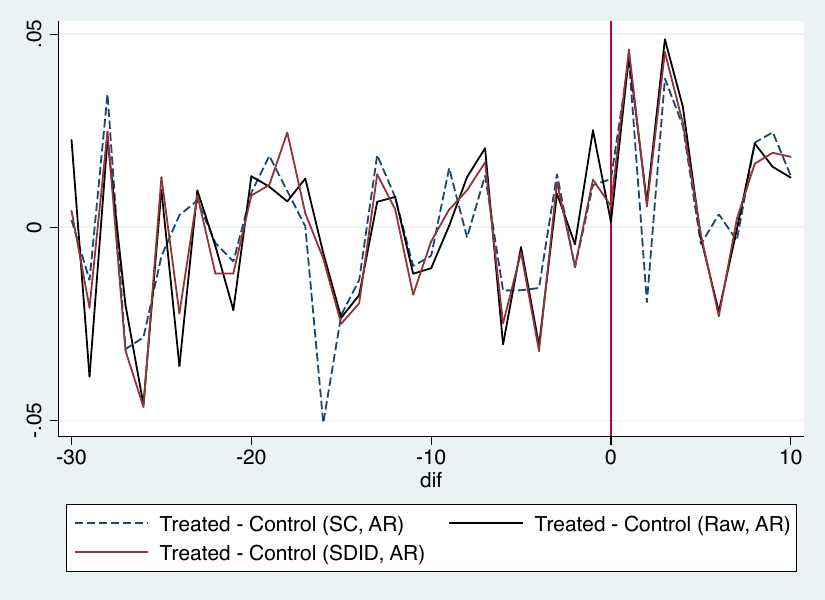}
    \end{minipage}
\end{figure}
\begin{figure}[H]
    \centering
    \caption{Period-by-Period ATT in Placebo and Post Period (Schedule connections)}\label{fig:sc_att_pool}
    \includegraphics{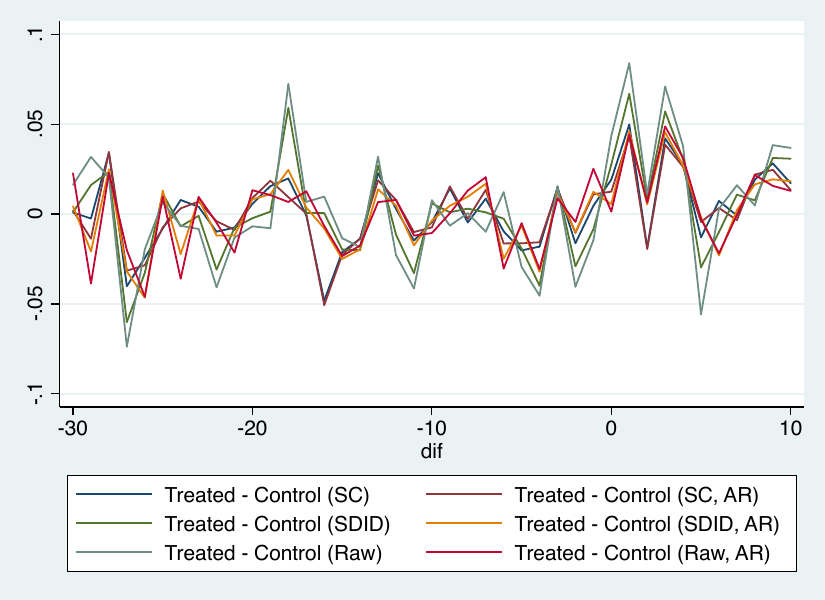}
\end{figure}

\subsection{Pre- versus Post-Event Beta and Weights}
In this section, we investigate how control beta is compared to treatment beta with different control weights. We also compare beta estimated pre-event with beta estimated post-event to see if the event also has a treatment effect on beta loadings. The pre-event beta is estimated over daily returns from day -280 to day -31, and the post-event beta is estimated over daily returns from day 31 to day 65. We exclude the immediate post-period because the returns can be confounded by the event effect. We also compare the synthetic weights estimated with pre- and post-period by comparing the treatment effect with pre- and post-weights.

First, we see that indeed synthetic control weights match control beta to treatment beta the best, compared to a simple average and synthetic diff-in-diff weights. For the pre-event, we see a control beta of 1.33 with synthetic control, compared to a treatment beta of 1.43. For the post-event, we have a control beta of 1.71, which is very close to a treatment beta of 1.73. The same conclusion can be drawn with Fama-French three-factor betas. Synthetic control weights give the closest control betas to treatment betas for market, size, and value factors.

Second, we see that post-betas are on average higher than pre-betas, suggesting that the event does have an effect on the underlying factor loadings of treatment firms. CAPM market beta increases from 1.43 to 1.73, a 21\% increase. In the three-factor model, we see the largest increase in size and value betas. Size beta increases from 0.23 to 0.41 (78\%), and value beta increases from 0.61 to 1.00 (64\%).

Third, Figure \ref{fig:sc_att_prepost} show the daily ATT with synthetic control weights for the placebo period (day -30 to -1), post-event period (day 0 to 30), and post-event-estimation period (day 31 to 65). We see that using post-event synthetic control weights gives us a larger event treatment effect, but it also gives a more positive ATT in the placebo period.

\begin{table}[H]
  \centering
  \caption{Pre-/Post-Event Market Beta from CAPM}
    \begin{tabular}{lrr}
    \toprule
    \multicolumn{3}{l}{Panel A: Pre Beta, Pre Weights} \\
    \midrule
          & \multicolumn{2}{c}{Market} \\
\cmidrule{2-3}          & \multicolumn{1}{c}{Treated} & \multicolumn{1}{c}{Control} \\
    \midrule
    Average & 1.4269 & 0.8251 \\
    SDID  & 1.4269 & 1.1111 \\
    SC    & 1.4269 & 1.3309 \\
    \midrule
    \multicolumn{3}{l}{Panel B: Post Beta, Post Weights} \\
    \midrule
          & \multicolumn{2}{c}{Market} \\
\cmidrule{2-3}          & \multicolumn{1}{c}{Treated} & \multicolumn{1}{c}{Control} \\
    \midrule
    Average & 1.7304 & 0.9377 \\
    SDID  & 1.7304 & 1.4076 \\
    SC    & 1.7304 & 1.7083 \\
    \midrule
    Panel C: Pre Beta, Post Weights &       &  \\
    \midrule
          & \multicolumn{2}{c}{Market} \\
\cmidrule{2-3}          & \multicolumn{1}{c}{Treated} & \multicolumn{1}{c}{Control} \\
    \midrule
    Average & 1.4269 & 0.8251 \\
    SDID  & 1.4269 & 1.0954 \\
    SC    & 1.4269 & 1.0751 \\
    \midrule
    Panel D: Post Beta, Pre Weights &       &  \\
    \midrule
          & \multicolumn{2}{c}{Market} \\
\cmidrule{2-3}          & \multicolumn{1}{c}{Treated} & \multicolumn{1}{c}{Control} \\
    \midrule
    Average & 1.7304 & 0.9377 \\
    SDID  & 1.7304 & 1.2105 \\
    SC    & 1.7304 & 1.3664 \\
    \bottomrule
    \end{tabular}%
\end{table}%

\begin{table}[H]
  \centering
  \caption{Pre-/Post-Event Beta from Fama-French Three Factors}
    \begin{tabular}{lrrrrrr}
    \toprule
    \multicolumn{7}{l}{Panel A: Pre Beta, Pre Weights} \\
    \midrule
          & \multicolumn{2}{c}{Market} & \multicolumn{2}{c}{SMB} & \multicolumn{2}{c}{HML} \\
\cmidrule{2-7}          & \multicolumn{1}{l}{Treated} & \multicolumn{1}{l}{Control} & \multicolumn{1}{l}{Treated} & \multicolumn{1}{l}{Control} & \multicolumn{1}{l}{Treated} & \multicolumn{1}{l}{Control} \\
    \midrule
    Original & 1.2748 & 0.6592 & 0.2330 & 0.7484 & 0.6068 & 0.7196 \\
    SDID  & 1.2748 & 0.9051 & 0.2330 & 0.8187 & 0.6068 & 0.8724 \\
    SC    & 1.2748 & 1.1477 & 0.2330 & 0.4796 & 0.6068 & 0.7495 \\
    \midrule
    \multicolumn{7}{l}{Panel B: Post Beta, Post Weights} \\
    \midrule
          & \multicolumn{2}{c}{Market} & \multicolumn{2}{c}{SMB} & \multicolumn{2}{c}{HML} \\
\cmidrule{2-7}          & \multicolumn{1}{l}{Treated} & \multicolumn{1}{l}{Control} & \multicolumn{1}{l}{Treated} & \multicolumn{1}{l}{Control} & \multicolumn{1}{l}{Treated} & \multicolumn{1}{l}{Control} \\
    \midrule
    Original & 1.2454 & 0.6265 & 0.4139 & 0.5633 & 0.9991 & 0.6898 \\
    SDID  & 1.2454 & 0.9544 & 0.4139 & 0.6791 & 0.9991 & 0.9785 \\
    SC    & 1.2454 & 1.2130 & 0.4139 & 0.4697 & 0.9991 & 1.0273 \\
    \bottomrule
    \end{tabular}%
\end{table}%
\begin{figure}[H]
    \centering
    \caption{Period-by-Period ATT with Pre \& Post SC Weights}\label{fig:sc_att_prepost}
    \includegraphics{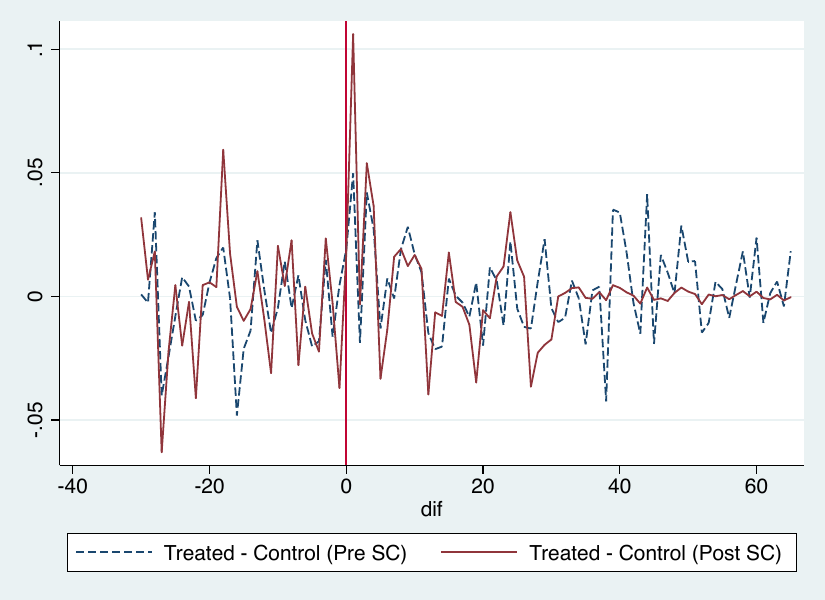}
\end{figure}

\subsection{Beta: All Public Firms as Control}
\begin{table}[H]
  \centering
  \caption{Pre-Event Market Beta from CAPM}
    \begin{tabular}{lrr}
    \toprule
    \multicolumn{3}{l}{Panel A: Pre Beta, Pre Weights} \\
    \midrule
          & \multicolumn{2}{c}{Market} \\
\cmidrule{2-3}          & \multicolumn{1}{c}{Treated} & \multicolumn{1}{c}{Control} \\
    \midrule
    Average & 1.4269 & 0.8324 \\
    SDID  & 1.4269 & 1.2814 \\
    SC    & 1.4269 & 1.3830 \\
    \bottomrule
    \end{tabular}%
  \label{tab:geithner_beta_all}%
\end{table}%

\begin{table}[H]
  \centering
  \caption{Pre-Event Beta from Fama-French Three Factors}
    \begin{tabular}{lrrrrrr}
    \toprule
    \multicolumn{7}{l}{Panel A: Pre Beta, Pre Weights} \\
    \midrule
          & \multicolumn{2}{c}{Market} & \multicolumn{2}{c}{SMB} & \multicolumn{2}{c}{HML} \\
\cmidrule{2-7}          & \multicolumn{1}{l}{Treated} & \multicolumn{1}{l}{Control} & \multicolumn{1}{l}{Treated} & \multicolumn{1}{l}{Control} & \multicolumn{1}{l}{Treated} & \multicolumn{1}{l}{Control} \\
    \midrule
    Original & 1.2748 & 0.8569 & 0.2330 & 0.5526 & 0.6068 & 0.1436 \\
    SDID  & 1.2748 & 1.1654 & 0.2330 & 0.6273 & 0.6068 & 0.5934 \\
    SC    & 1.2748 & 1.2201 & 0.2330 & 0.3774 & 0.6068 & 0.6743 \\
    \bottomrule
    \end{tabular}%
\end{table}%

\subsection{Placebo Period ATT: All Public Firms as Control}

\begin{table}[H]
    \centering
    \caption{Placebo Period ATT to Geithner Announcement (Schedule connections)}
{
\def\sym#1{\ifmmode^{#1}\else\(^{#1}\)\fi}
\begin{tabular}{l*{4}{c}}
\hline\hline
            &\multicolumn{1}{c}{(1)}&\multicolumn{1}{c}{(2)}&\multicolumn{1}{c}{(3)}&\multicolumn{1}{c}{(4)}\\
            &\multicolumn{1}{c}{Average}&\multicolumn{1}{c}{DID}&\multicolumn{1}{c}{SC}&\multicolumn{1}{c}{SDID}\\
\hline
Treated    & -0.003 &      -0.003   &      -0.004   &      -0.002   \\
           & (0.004) &     (0.002)   &     (0.003)   &     (0.002)   \\
\hline
Observations & 122,850 &   1,044,225   &   1,044,225   &   1,044,225   \\
\hline\hline
\multicolumn{4}{l}{\footnotesize Standard errors in parentheses}\\
\multicolumn{4}{l}{\footnotesize * p<0.10, ** p<0.05, *** p<0.01}\\
\end{tabular}
}
\end{table}
\begin{table}[H]
    \centering
    \caption{Placebo Period ATT to Geithner Announcement (Personal connections)}
{
\def\sym#1{\ifmmode^{#1}\else\(^{#1}\)\fi}
\begin{tabular}{l*{4}{c}}
\hline\hline
            &\multicolumn{1}{c}{(1)}&\multicolumn{1}{c}{(2)}&\multicolumn{1}{c}{(3)}&\multicolumn{1}{c}{(4)}\\
            &\multicolumn{1}{c}{Average}&\multicolumn{1}{c}{DID}&\multicolumn{1}{c}{SC}&\multicolumn{1}{c}{SDID}\\
\hline
Treated     & -0.004 &      -0.003   &      -0.001   &      -0.002   \\
            & (0.006) &     (0.004)   &     (0.003)   &     (0.003)   \\
\hline
Observations& 122,850&   1,044,225   &   1,044,225   &   1,044,225   \\
\hline\hline
\multicolumn{4}{l}{\footnotesize Standard errors in parentheses}\\
\multicolumn{4}{l}{\footnotesize * p<0.10, ** p<0.05, *** p<0.01}\\
\end{tabular}
}
\end{table}
\begin{table}[H]
    \centering
    \caption{Placebo Period ATT to Geithner Announcement (New York connections)}
{
\def\sym#1{\ifmmode^{#1}\else\(^{#1}\)\fi}
\begin{tabular}{l*{4}{c}}
\hline\hline
            &\multicolumn{1}{c}{(1)}&\multicolumn{1}{c}{(2)}&\multicolumn{1}{c}{(3)}&\multicolumn{1}{c}{(4)}\\
            &\multicolumn{1}{c}{Average}&\multicolumn{1}{c}{DID}&\multicolumn{1}{c}{SC}&\multicolumn{1}{c}{SDID}\\
\hline
Treated     & -0.000 &       0.000   &      -0.002   &       0.001   \\
            & (0.003) &     (0.002)   &     (0.002)   &     (0.002)   \\
\hline
Observations& 122,850 &   1,044,225   &   1,044,225   &   1,044,225   \\
\hline\hline
\multicolumn{4}{l}{\footnotesize Standard errors in parentheses}\\
\multicolumn{4}{l}{\footnotesize * p<0.10, ** p<0.05, *** p<0.01}\\
\end{tabular}
}
\end{table}

\clearpage

\section{Additional Results for Index Inclusion}
\label{appsec:index}
This section presents addition results for the S\&P index inclusion empirical example.

\clearpage

\begin{figure}[thb]
    \caption{\small\textbf{Cumulative Distributions of Factor Returns by Announcement Status}
    This figure plots the daily returns of the S\&P 500 index and Small-minus-Big (SMB) factor on the dates when there are index inclusion announcements versus the dates without. The blue line plots the overall cumulative distribution function from 1962 to 2023, and the red lines plot the cumulative distribution function of daily returns on the days when there is an index inclusion event. }
\label{fig:factor_balance_index}
    \begin{subfigure}[t]{\textwidth}
        \centering
        \caption{Panel A: S\&P 500 Daily Returns}
        \includegraphics[width=\textwidth]{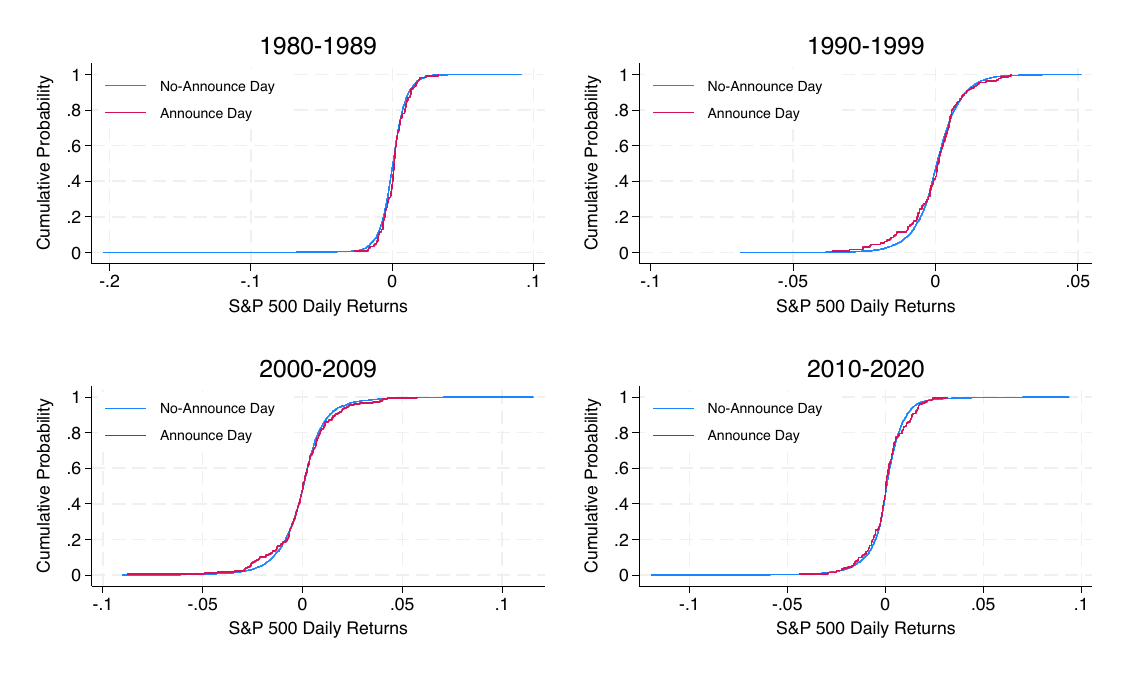}
    \end{subfigure}
        \vspace{1cm}
    \begin{subfigure}[t]{\textwidth}
        \centering
        \caption{Panel B: SMB Factor Daily Returns}
        \includegraphics[width=\textwidth]{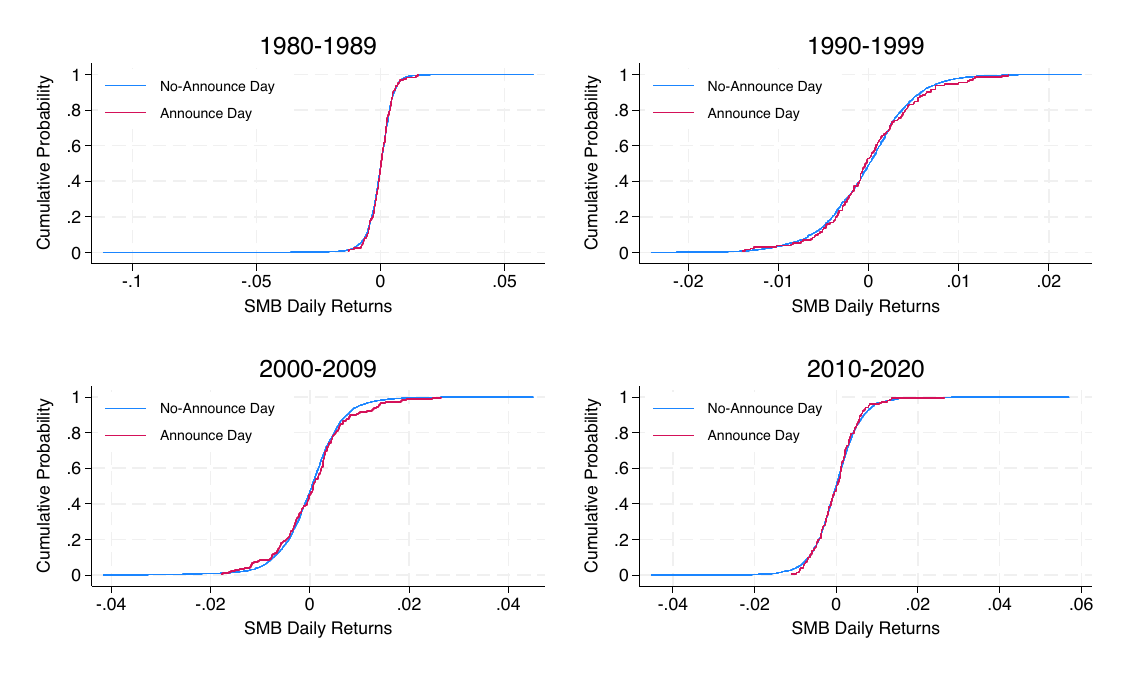}
    \end{subfigure}
\end{figure}

\clearpage

\begin{figure}[H]
    \caption{\small\textbf{Pre-addition Cumulative Market Factor Returns (Inclusion vs. Randomized No-Inclusion Days)}\label{fig:preindexcumulret}
    This figure plots the average cumulative returns on the market and the SMB factor following index inclusion announcements in event time, averaged across inclusions for each decade. We also plot the average cumulative returns on the market following randomized no-inclusion days. For each inclusion date, we pick a random date on no-inclusion dates. The returns are normalized to start at zero, 100-trading days before the announcement.}
    \begin{subfigure}[t]{0.5\textwidth}
            \caption{Panel A: S\&P 500 Cumulative Returns}
    \includegraphics[width=\linewidth]{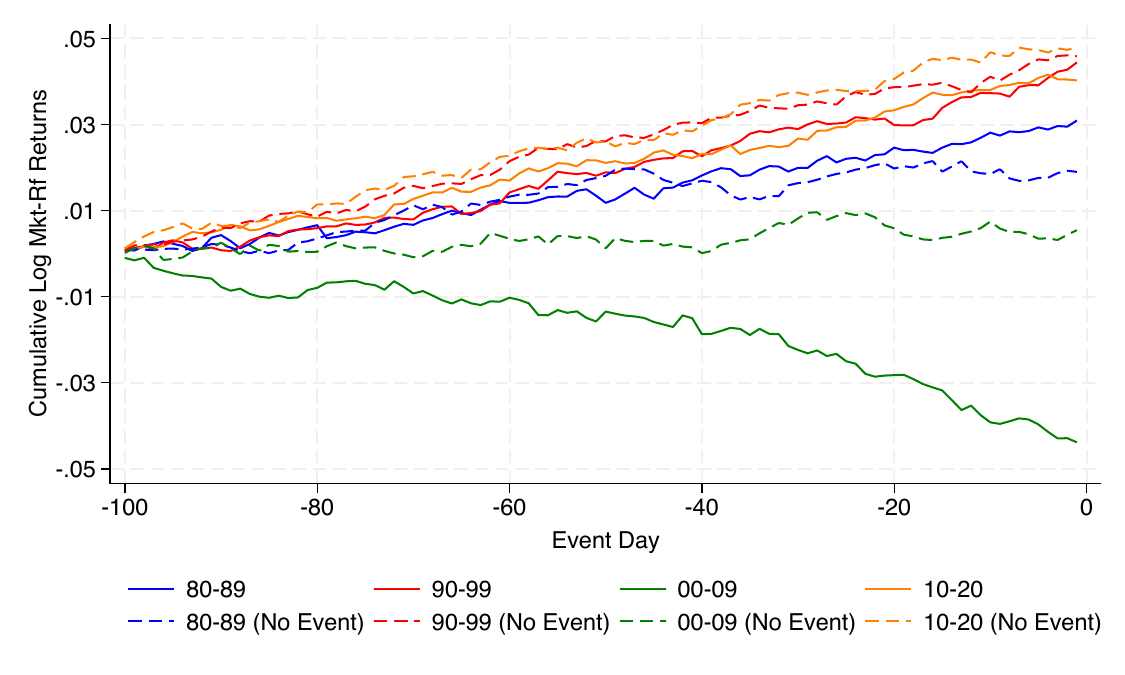}
    \end{subfigure}
        \begin{subfigure}[t]{0.5\textwidth}            
        \caption{Panel B: SMB Cumulative Returns}
    \includegraphics[width=\linewidth]{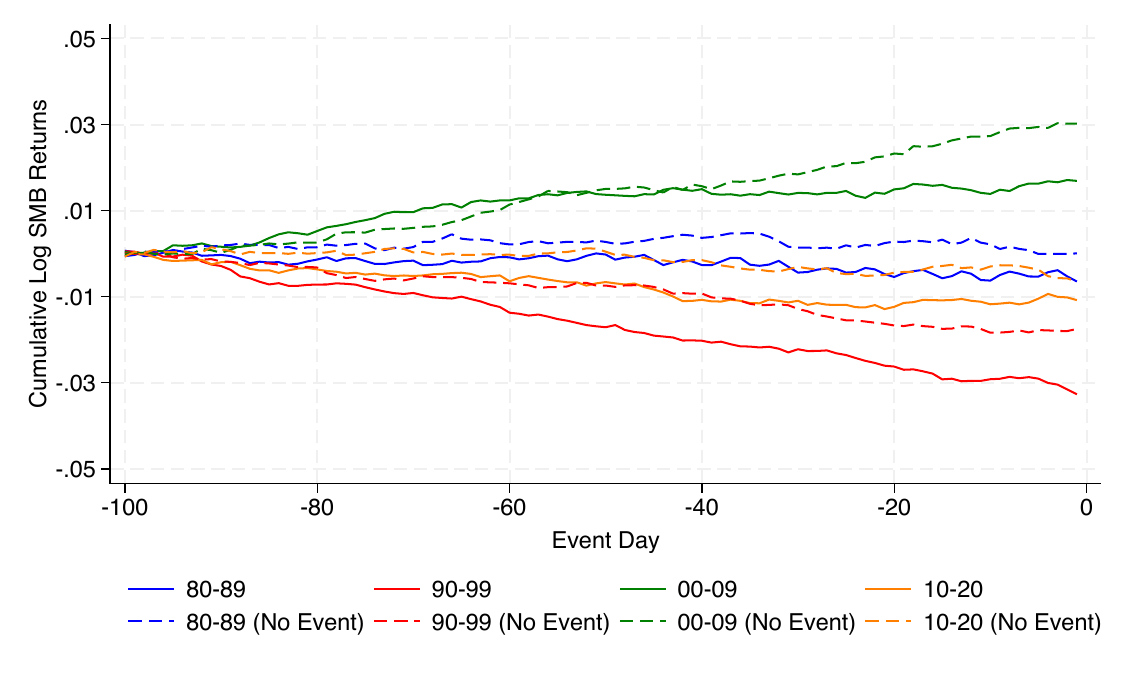}
    \end{subfigure}
%
\end{figure}

\clearpage

\begin{figure}[thb]
    \centering
    \caption{Cumulative pre-addition market-adjusted returns (Treated vs. propensity score matched)}
    \includegraphics[width=\linewidth]{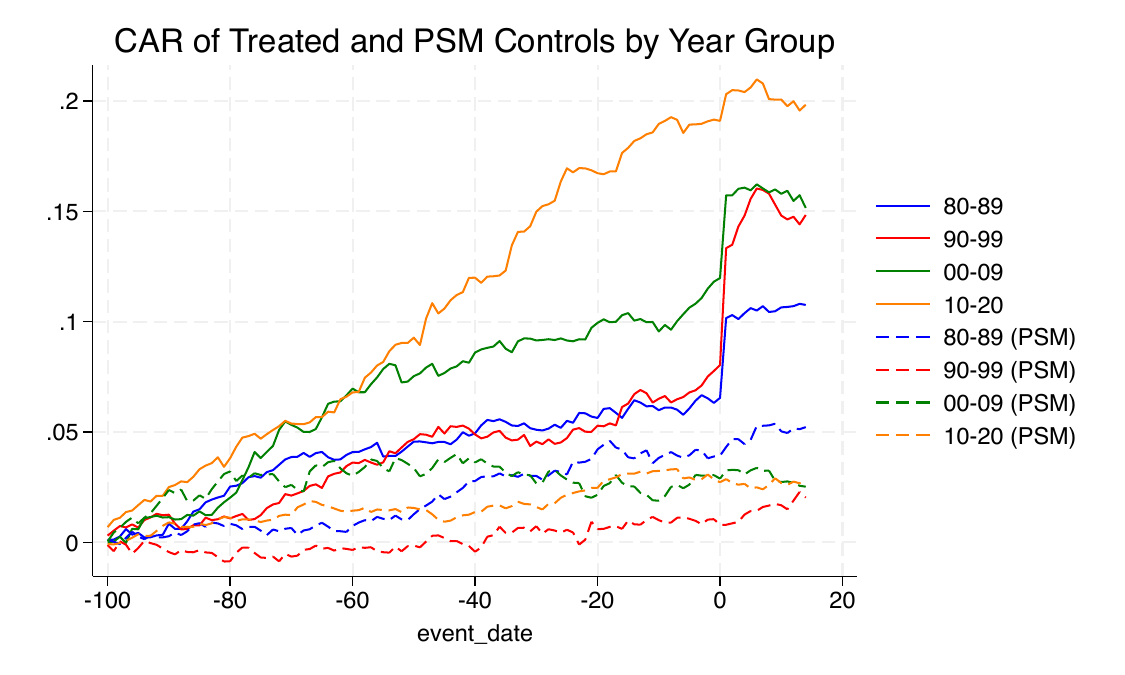}
\end{figure}

\clearpage
\begin{figure}[thb]
    \centering
    \caption{Cumulative pre-addition market-adjusted returns (Treated vs. synthetic method)}
    \includegraphics[width=\linewidth]{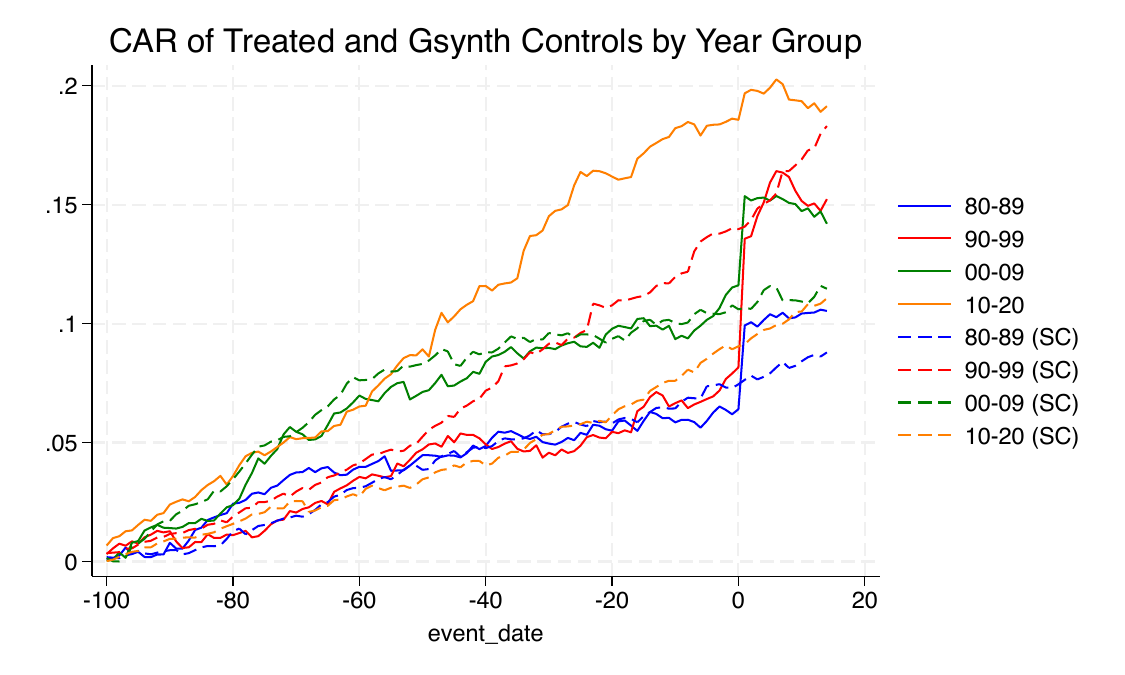}
\end{figure}

\clearpage

\section{Additional results on empirical example 3, merger announcments }

This section presents additional exhibits and results for our first mergers announcement example.

\clearpage

\begin{figure}[H]
\label{fig:ma_cdf_factor}
    \caption{\small\textbf{Cumulative Distributions of Factor Returns by Announcement Status}
    This figure plots the daily returns of the CRSP value-weighted index on the dates when there are merger announcements versus the dates without. The blue line plots the overall cumulative distribution function from 1962 to 2023, and the red lines plot the cumulative distribution function of daily returns on the days when there is a merger annnouncment event. }
        \centering
        \includegraphics[width=0.8\textwidth]{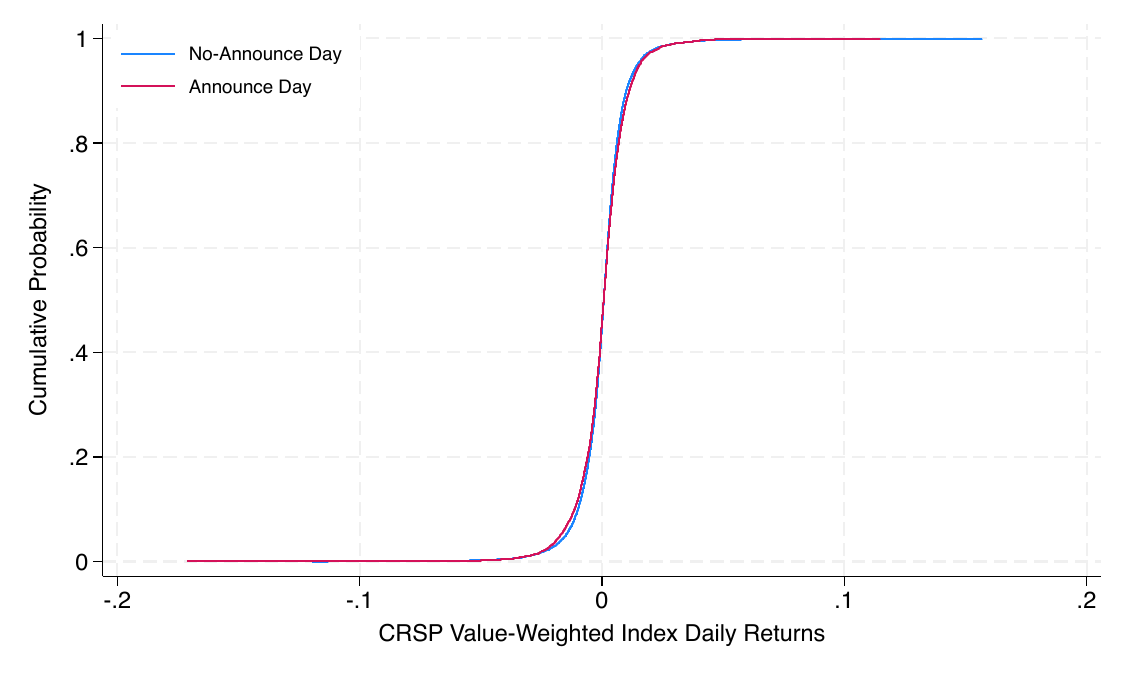}
\end{figure}

\clearpage

\begin{table}[htbp]
  
  \caption{\small \textbf{Beta Distributions of Acquirers in Merger} 
    This table presents the average CAPM and Fama-French three-factor betas for acquirers in merger transactions. We estimate firm-level betas using daily stock returns before and after the announcement. For pre betas, we use event date -280 to -30, while for post beta, we use event date 30-280 to estimate betas. We provide the mean and median of CAPM market beta and betas in Fama-French three-factor model. We test if the pre and post betas are statistically different using a two-sided t-test.} \label{tab:beta_ma_sl_m}  
    \centering
    \begin{tabular}{lccccc}
    \toprule
          & \multicolumn{2}{c}{Pre} & \multicolumn{2}{c}{Post} & Mean t-test \\
          & Mean  & Median & Mean  & Median & Pre - Post \\
          \midrule
    CAPM Beta & 0.952 & 0.909 & 0.968 & 0.926 & -0.016*** \\
    FF3F Mkt Beta & 1.020 & 0.993 & 1.022 & 1.001 & -0.002 \\
    FF3F SMB Beta & 0.690 & 0.620 & 0.678 & 0.613 & 0.012* \\
    FF3F HML Beta & 0.101 & 0.141 & 0.132 & 0.179 & -0.031*** \\
    \bottomrule
    \end{tabular}%
  \label{tab:addlabel}%
\end{table}%

\clearpage

\section{Additional results on empirical example 4, close mergers with winners and losers}

This section presents additional exhibits and results for our second mergers announcement example.

\subsubsection{Estimation Window and Model Stability}
We estimate different counterfactual models with different estimation window lengths. We vary the estimation window from 35 months as default to 12 month. We again leave the period $t=0$ as the placebo period. We plot the average treatment effects in the estimation window, placebo period, treatment period, and post treatment windows, respectively.

\begin{figure}[H]
    \caption{\small\textbf{Average Treatment Effects by Estimation Window Length}\label{fig:close_contest_att_prelength}
    This figure plots the average treatment effects of winners in merger contests in the estimation, placebo, treatment, and post-treatment windows, by the length of the estimation window. We estimate betas and train synthetic control and gsynth models using pre-announcement periods from 12 months to 35 months. We then compute the counterfactual returns from different models with the design-based loser portfolio. }
    \begin{subfigure}[t]{0.5\textwidth}
            \caption{Panel A: Treatment Period $t=1$}
    \includegraphics[width=\linewidth]{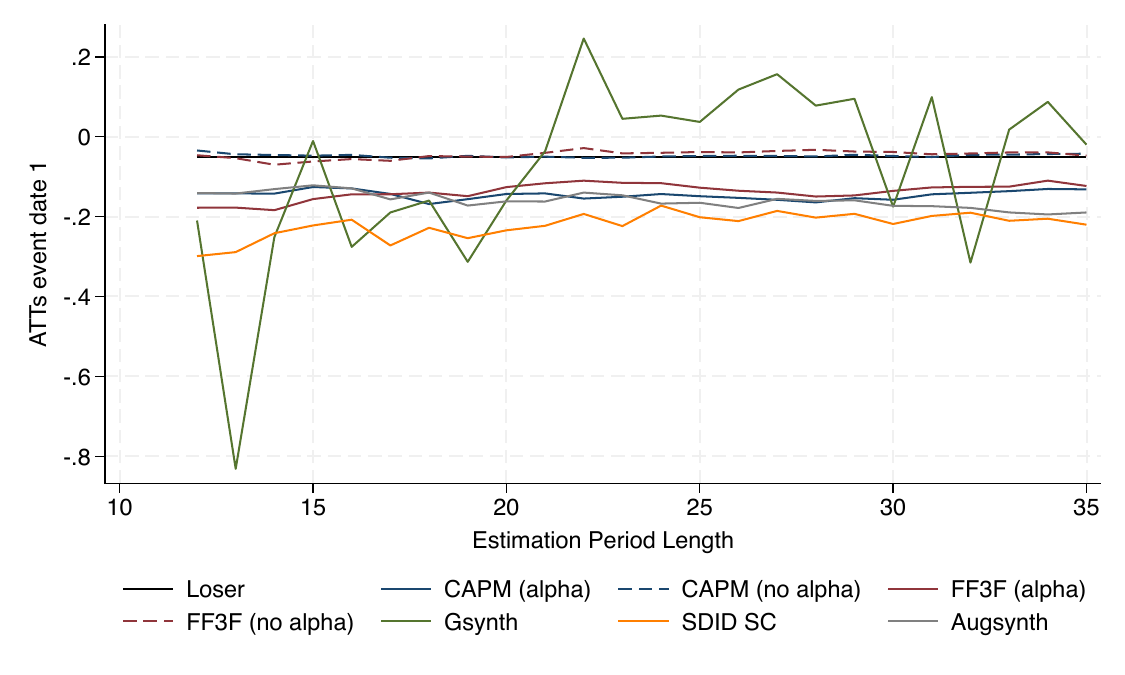}
    \end{subfigure}
        \begin{subfigure}[t]{0.5\textwidth}            
        \caption{Panel B: Post-Treatment Periods $t>1$}
    \includegraphics[width=\linewidth]{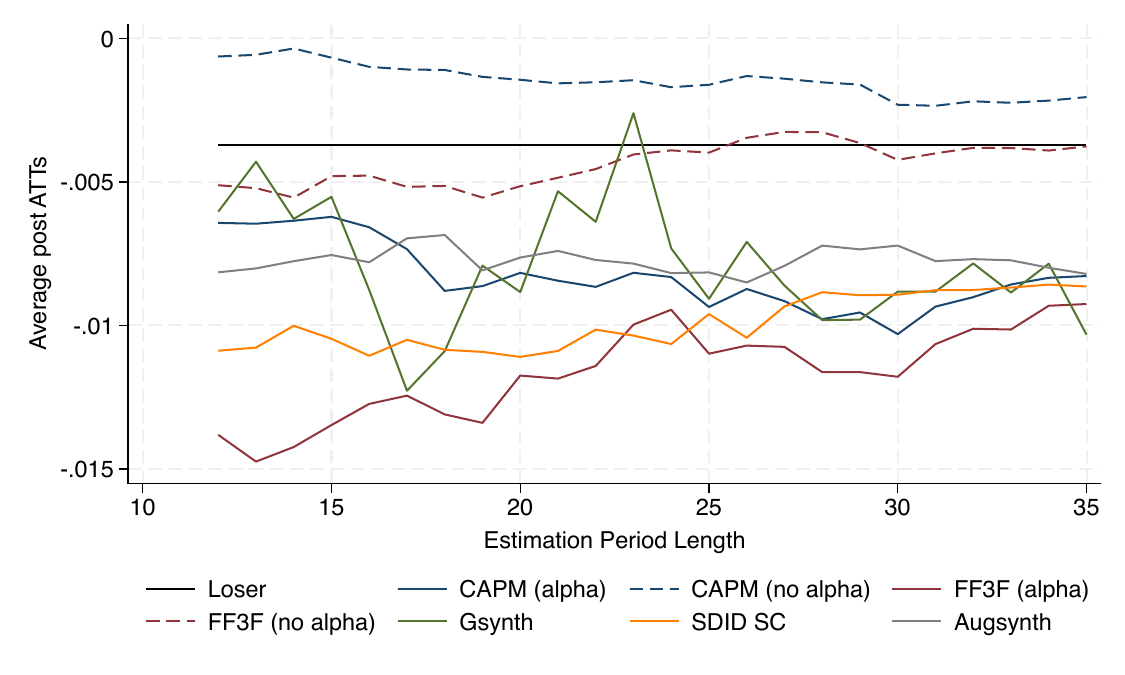}
    \end{subfigure}
\begin{subfigure}[t]{0.5\textwidth}
            \caption{Panel A: Placebo Period $t=0$}
    \includegraphics[width=\linewidth]{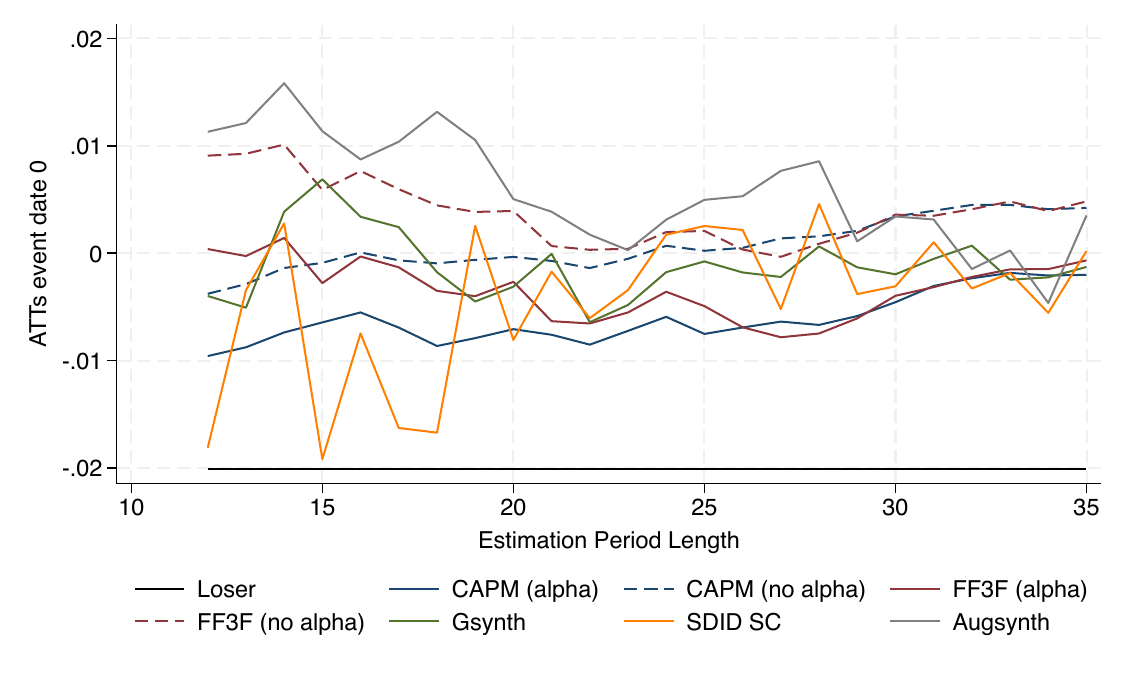}
    \end{subfigure}
        \begin{subfigure}[t]{0.5\textwidth}            
        \caption{Panel B: Estimation Periods $t<0$}
    \includegraphics[width=\linewidth]{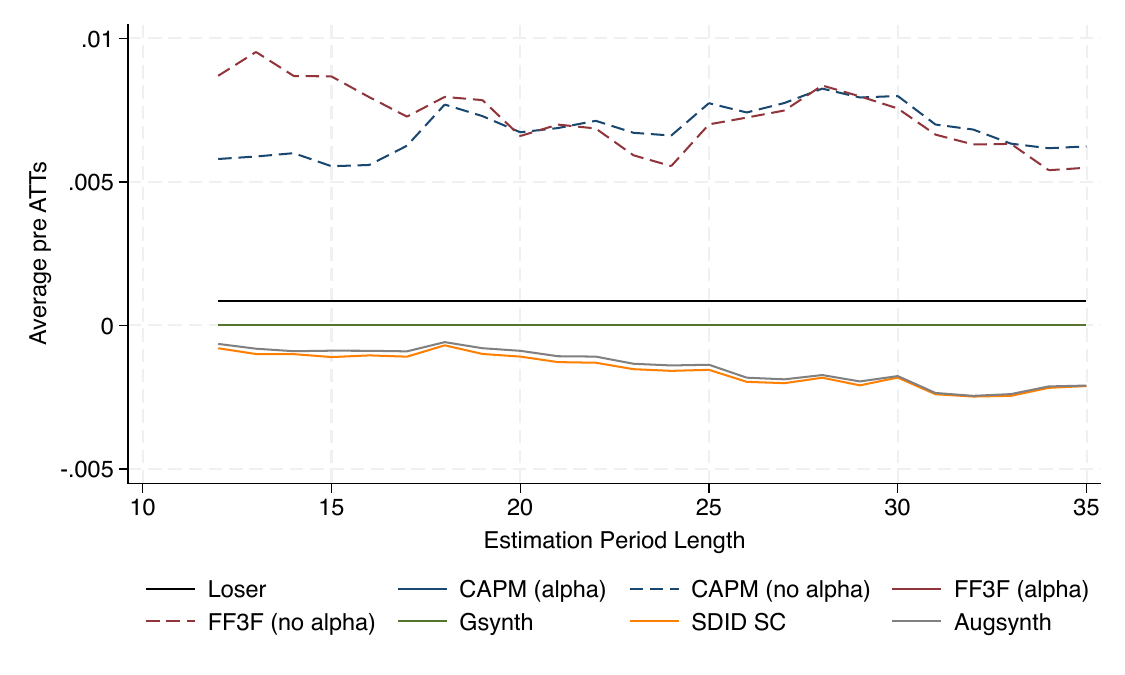}
    \end{subfigure}
\end{figure}